\documentclass[11pt]{amsart}
\usepackage[margin=1in]{geometry}
\usepackage{amssymb,xcolor,colortbl,diagbox,comment,graphicx,mathpazo,cite}
\usepackage{MnSymbol}
\usepackage{color,float,fullpage,mathrsfs,mathtools}
\usepackage{enumerate,enumitem}
\usepackage[colorlinks=true, pdfstartview=FitV, linkcolor=blue, citecolor=red, urlcolor=blue]{hyperref}
\usepackage{multirow,array}
\usepackage{makecell}
\usepackage{epstopdf}

\newtheorem{proposition}{Proposition}
\newtheorem{theorem}{Theorem}
\newtheorem{corollary}{Corollary}

\theoremstyle{definition}


\title[Defocusing-defocusing coupled Hirota equations]{Inverse scattering transform for the defocusing-defocusing coupled Hirota equations with non-zero boundary conditions: double-pole solutions}

\author{Peng-Fei Han$^{1}$}
\author{Wen-Xiu Ma$^{1,2,3,4}$}
\author{Ru-Suo Ye$^{1}$}
\author{Yi Zhang$^{1,*}$}

\thanks{$*$ Corresponding author. E-mail address: hanpf1995@163.com(Peng-Fei Han), mawx@cas.usf.edu(Wen-Xiu Ma), rusuoye@163.com(Ru-Suo Ye), zy2836@163.com(Yi Zhang)}

\dedicatory{$^{1}$Department of Mathematics, Zhejiang Normal University, Jinhua 321004, People's Republic of China. \\
$^{2}$Department of Mathematics, King Abdulaziz University, Jeddah 21589, Saudi Arabia. \\
$^{3}$Department of Mathematics and Statistics, University of South Florida, Tampa, FL 33620-5700, USA. \\
$^{4}$Material Science Innovation and Modelling, Department of Mathematical Sciences, North-West University, Mafikeng Campus, Mmabatho 2735, South Africa.}


\date{\today}

\begin{document}

\begin{abstract}
The inverse scattering transform for the defocusing-defocusing coupled Hirota equations with non-zero boundary conditions at infinity is thoroughly discussed. We delve into the analytical properties of the Jost eigenfunctions and scrutinize the characteristics of the scattering coefficients. To enhance our investigation of the fundamental eigenfunctions, we have derived additional auxiliary eigenfunctions with the help of the adjoint problem. Two symmetry conditions are studied to constrain the behavior of the eigenfunctions and scattering coefficients. Utilizing these symmetries, we precisely delineate the discrete spectrum and establish the associated symmetries of the scattering data. By framing the inverse problem within the context of the Riemann-Hilbert problem, we develop suitable jump conditions to express the eigenfunctions. Consequently, we deduce the pure soliton solutions from the defocusing-defocusing coupled Hirota equations, and the double-poles solutions are provided explicitly for the first time in this work.
\end{abstract}
\maketitle

\tableofcontents

\section{Introduction}
\label{s:intro}

In the vastness of nature, nonlinear wave phenomena play an important role. Its dynamic behavior is described by nonlinear wave equations, which are a class of evolutionary nonlinear partial differential equations~\cite{1}. These equations are the basis for understanding a large number of phenomena in various scientific fields. In particular, the nonlinear Schr\"odinger (NLS) equation which is expressed in scalar~\cite{2}, vector~\cite{3}, and matrix~\cite{4} forms, acting as a universal model for the evolution of weakly nonlinear dispersive wave trains. This makes it an indispensable tool in the study of deep-water waves and nonlinear optics. Furthermore, the NLS equations are instrumental in elucidating complex phenomena like modulational instability and the genesis of rogue waves.

In order to solve the mystery of these nonlinear systems, the inverse scattering transform (IST) is a powerful analytical technique. IST was first proposed by Gardner, Greene, Kruskal, and Miura in 1967~\cite{5}, aiming to provide an exact solution to the initial value problem of the Korteweg-de-Vries equation by using the Lax pairs~\cite{6}. This pioneering method was later extended to many integrable systems characterized by Lax pairs, providing a principled method for solving the initial value problems. When applicable, the IST serves as an efficacious instrument for dissecting the intricate behavior of solutions. Despite its utility, the formulation of the IST remains an open challenge in certain scenarios, indicating that there is still much to explore and understand in this intricate field of study.

AKNS hierarchy is an important research field which provides a powerful mathematical framework for understanding various nonlinear phenomena~\cite{C1}. In the case of local group reduction, the AKNS hierarchy can be simplified by applying two local group reductions to obtain a specific instance of the Sasa-Satsuma type matrix integrable hierarchies~\cite{C2}. In the case of nonlocal group reduction, a new form of nonlocal NLS equation can be explored~\cite{C3,C4}. In some nonlocal cases with one nonlocal group reduction, researchers have presented the most general Darboux transformations, which contain generalized Darboux transformations~\cite{C5}. From a mathematical perspective, the general coupled Hirota equations are a linear combination of a special case of a reduced AKNS hierarchy with one local group reduction~\cite{C6,C7}. The general coupled Hirota equations~\cite{A1,A2,A3,A4,A5,A6,A7,A8,A9,A10,A11,A12,A13} offer a comprehensive model for studying the wave propagation of two ultrashort optical fields in optical fibers, accounting for the intricate interplay of nonlinear and dispersive effects that govern the evolution of the pulses as they travel through the fiber.
\begin{subequations}\label{1.1}
\begin{align}
\mathrm{i}q_{1,t}+& q_{1,xx}+2(\sigma_{1}\left|q_{1}\right|^{2}+\sigma_{2}\left|q_{2}\right|^{2})q_{1}
+\mathrm{i}\sigma[q_{1,xxx}+(6\sigma_{1}\left|q_{1}\right|^{2}+3\sigma_{2}\left|q_{2}\right|^{2})q_{1,x} +3\sigma_{2}q_{1}q_{2}^{*}q_{2,x}]=0, \\
\mathrm{i}q_{2,t}+& q_{2,xx}+2(\sigma_{1}\left|q_{1}\right|^{2}+\sigma_{2}\left|q_{2}\right|^{2})q_{2}
+\mathrm{i}\sigma[q_{2,xxx}+(6\sigma_{2}\left|q_{2}\right|^{2}+3\sigma_{1}\left|q_{1}\right|^{2})q_{2,x}
+3\sigma_{1}q_{2}q_{1}^{*}q_{1,x}]=0,
\end{align}
\end{subequations}
where $q_{1}$ and $q_{2}$ are the two-component electric fields, the parameters $\sigma_{1}$, $\sigma_{2}$ and $\sigma$ are real constants~\cite{A1}. Given the positivity of $\sigma_{1}$ and $\sigma_{2}$, the equations~\eqref{1.1} are designated as the focusing-focusing coupled Hirota equations, which lead to the energy concentration of the wave during the interaction and the focusing effect. With the assumption that $\sigma_{1}$ and $\sigma_{2}$ are negative, the set of equations~\eqref{1.1} are referred to as the defocusing-defocusing coupled Hirota equations, which lead to the energy concentration of the wave during the interaction and the defocusing effect. Supposing $\sigma_{1}$ and $\sigma_{2}$ exhibit opposite signs, the equations \eqref{1.1} are categorized as the mixed coupled Hirota equations~\cite{A1}, which lead to the coexistence of focusing and defocusing of waves, resulting in more complex wave dynamics.

In 1992, Tasgal and Potasek~\cite{A2} employed the IST to derive soliton solutions for the coupled higher-order NLS equations within a specific parameter regime. This work underscores the integrability of the general coupled Hirota equations, characterized by the presence of the Lax pair, $N$-th order Darboux transformation and a variety of localized wave solutions, as further elaborated in~\cite{A1,A3}. It has been shown that the general coupled Hirota equations also admit the dark soliton solutions~\cite{A4}, the high-order rational rogue waves and multi-dark soliton structures~\cite{A5}, rogue wave solutions~\cite{A6,A7,A8}, semirational solutions~\cite{A9}, analytical solutions~\cite{A10}, dark-bright-rogue wave solutions~\cite{A11}, the interactions between breathers and rogue waves~\cite{A12}, the interactions between dark-bright solitons and rogue waves~\cite{A12}. Utilizing nonlinear steepest descent techniques, the leading-order asymptotic expressions and consistent error bounds for solutions to the the coupled Hirota equations were meticulously examined, as detailed in~\cite{A13}.

Recently, the IST and Riemann-Hilbert (RH) method have been extensively applied to investigate soliton solutions for the Hirota equation, as evidenced by recent studies~\cite{7,8,9,10}. Furthermore, these methods have yet to be explored in the context of the more complex, general coupled Hirota equations with the non-zero boundary conditions (NZBC) at infinity. In light of this gap in the literature, the present paper aims to delve into the application of the IST to the defocusing-defocusing coupled Hirota equations, offering a novel perspective on this under-explored area. When $\sigma_{1}=\sigma_{2}=-1$, Eqs. \eqref{1.1} can be convert the defocusing-defocusing coupled Hirota equations
\begin{equation}\label{1.2}
\begin{split}
\mathrm{i}\widetilde{\mathbf{q}}_{t}+\widetilde{\mathbf{q}}_{xx}-2\Vert \widetilde{\mathbf{q}} \Vert^{2}\widetilde{\mathbf{q}}+\mathrm{i}\sigma[\widetilde{\mathbf{q}}_{xxx}-3\Vert \widetilde{\mathbf{q}} \Vert^{2}\widetilde{\mathbf{q}}_{x}
-3(\widetilde{\mathbf{q}}^{\dagger}\widetilde{\mathbf{q}}_{x})\widetilde{\mathbf{q}}]=\mathbf{0},
\end{split}
\end{equation}
where $\widetilde{\mathbf{q}}(x,t)=(\widetilde{q}_{1}(x,t),
\widetilde{q}_{2}(x,t))^{T}$. Employing the variable transformation $\widetilde{\mathbf{q}}(x,t)=\mathbf{q}(x,t)\mathrm{e}^{-2\mathrm{i}q_{0}^{2}t}$, we can derive the defocusing-defocusing coupled Hirota equations with the NZBC, namely
\begin{equation}\label{1.3}
\begin{split}
\mathrm{i}\mathbf{q}_{t}+\mathbf{q}_{xx}-2(\Vert \mathbf{q} \Vert^{2}-q_{0}^{2})\mathbf{q}+\mathrm{i}\sigma[\mathbf{q}_{xxx}-3\Vert \mathbf{q} \Vert^{2}\mathbf{q}_{x}
-3(\mathbf{q}^{\dagger}\mathbf{q}_{x})\mathbf{q}]=\mathbf{0},
\end{split}
\end{equation}
and the corresponding NZBC at infinity are
\begin{equation}\label{1.4}
\begin{split}
\lim_{x\rightarrow\pm\infty}{\mathbf{q}}(x,t)=\mathbf{q}_{\pm}=\mathbf{q}_{0}
\mathbf{e}^{\mathrm{i}\delta_{\pm}},
\end{split}
\end{equation}
where $\mathbf{q}=\mathbf{q}(x,t)=(q_{1}(x,t),q_{2}(x,t))^{T}$ and $\mathbf{q}_{0}$ are two-component vectors, $q_{0}=\Vert \mathbf{q}_{0} \Vert$, with $\delta_{\pm}$ are real numbers. For the defocusing-defocusing coupled Hirota equations with the NZBC $\mathbf{q}_{+}$ and $\mathbf{q}_{-}$ at infinity, the scenarios of parallel and non-parallel orientations have not yet been explored. This paper initially focuses on the case where $\mathbf{q}_{+}$ is parallel to $\mathbf{q}_{-}$, with the intention to address the non-parallel case of $\mathbf{q}_{+}$ and $\mathbf{q}_{-}$ in future research endeavors.

Compared with the defocusing Hirota equation~\cite{7}, the defocusing-defocusing coupled Hirota equations~\cite{A1} are associated with a $3\times3$ matrix Lax pair, which make the study of spectral analysis very difficult. In the case of the NZBC, the study of $3\times3$ matrix Lax pair usually has the problem that the analytical Jost eigenfunction is not analytical, which will make it more difficult to construct the IST. Specifically, the IST has demonstrated its unique value in the study of nonlinear wave equations with specific boundary conditions~\cite{11,12}. For example, when studying the focusing~\cite{13} and defocusing~\cite{B1} Manakov systems with the NZBC at infinity, the application of IST enables researchers to analyze the dynamic characteristics of dark-dark and dark-bright solitons in detail. In addition, the IST has also proved its effectiveness in solving the coupled Gerdjikov-Ivanov equation~\cite{14} with the NZBC, which not only helps to reveal the existence of the dark-dark solitons, bright-bright and breather-breather, but also provides insight into understanding their interaction. Based on these studies~\cite{11,12,13,B1,14}, it is planned to further explore the application of IST in solving nonlinear equations with more complex boundary conditions~\cite{B2}. In this paper, we will use the IST method to study the defocusing-defocusing coupled Hirota equations with the parallel NZBC at infinity in order to more fully understand and predict their analytical and asymptotic properties.

The structure of the remaining sections of this paper is outlined as follows. Section~\ref{s:Direct scattering problem} delves into the intricacies of the direct scattering problem. Initially, we delineate the Jost functions associated with the Lax pair, ensuring they adhere to the stipulated boundary conditions. Subsequently, we scrutinize the analytical characteristics of the modified eigenfunctions, leveraging the foundational definitions of the Jost functions. Furthermore, we rigorously establish the analytical properties of the corresponding coefficients of the scattering matrix, grounded in the precise formulation of the scattering matrix itself. Ultimately, we address the adjoint problem by delineating the auxiliary eigenfunctions, which in turn facilitates the derivation of symmetries for the Jost eigenfunctions, scattering coefficients and auxiliary eigenfunctions. In Section~\ref{s:Discrete spectrum and asymptotic behavior}, we delve into the characterization of the discrete spectrum. Additionally, we systematically analyze the asymptotic behavior of the modified Jost eigenfunctions and the scattering matrix elements for $z\rightarrow\infty$ and $z\rightarrow 0$. In Section \ref{s:Inverse problem}, based on the RH problem to formulate the inverse problem, we construct appropriate jump conditions to express the eigenfunctions. By using the meromorphic matrices, the corresponding residue conditions and norming constants are obtained. We construct the formal solutions of the RH problem and reconstruction formula with the help of the Plemelj formula. The pure soliton solutions are derived within the framework of reflectionless potentials and complete with a comprehensive proof. Subsequently, the discussion delves into the categorization of solitons possessing discrete eigenvalues, both within and beyond the specified circumference. In Section~\ref{s:Double-pole solutions}, we derive the solutions associated with the double zeros of the analytic scattering coefficients, and explicitly present the solutions for double poles. The results are summarized in Section~\ref{s:Conclusion}.

\section{Direct scattering problem}
\label{s:Direct scattering problem}

Generally speaking, the IST of integrable nonlinear equations needs to be studied through the formula of their Lax pairs. Our calculations are based on the following $3\times3$ Lax pair, which corresponds to the defocusing-defocusing coupled Hirota equations~\eqref{1.3}
\begin{equation}\label{2.1}
\begin{split}
\psi_{x}=\mathbf{X}\psi, \quad \psi_{t}=\mathbf{T}\psi,
\end{split}
\end{equation}
where $\psi=\psi(x,t)$, the matrices $\mathbf{X}$ and $\mathbf{T}$ are written as
\begin{equation}\label{2.2}
\begin{split}
\mathbf{X}=\mathbf{X}(k;x,t)&=\mathrm{i}k\mathbf{J}+\mathrm{i}\mathbf{Q}, \\
\mathbf{T}=\mathbf{T}(k;x,t)&=4\mathrm{i}\sigma k^{3}\mathbf{J}-\mathrm{i}q_{0}^{2}\mathbf{J}
+k^{2}(4\mathrm{i}\sigma\mathbf{Q}+2\mathrm{i}\mathbf{J})
+k(2\mathrm{i}\mathbf{Q}-2\sigma\mathbf{Q}_{x}\mathbf{J}-2\mathrm{i}\sigma\mathbf{J}\mathbf{Q}^{2})
-2\mathrm{i}\sigma\mathbf{Q}^{3} \\
&-\mathrm{i}\mathbf{J}\mathbf{Q}^{2}-\mathrm{i}\sigma\mathbf{Q}_{xx}
-\mathbf{Q}_{x}\mathbf{J}+\sigma[\mathbf{Q},\mathbf{Q}_{x}],
\end{split}
\end{equation}
while the definitions of $\mathbf{J}$ and $\mathbf{Q}=\mathbf{Q}(x,t)$ are as follows
\begin{equation}\label{2.3}
\begin{split}
\mathbf{J}=\begin{pmatrix}
    1 & \mathbf{0_{1\times2}}  \\
    \mathbf{0_{2\times1}} & -\mathbf{I_{2\times2}}
  \end{pmatrix},  \quad
\mathbf{Q}=\begin{pmatrix}
    0 & -\mathbf{q}^{\dagger}   \\
    \mathbf{q} & \mathbf{0_{2\times2}}
  \end{pmatrix}, \quad \mathbf{J}\mathbf{Q}=-\mathbf{Q}\mathbf{J},
\end{split}
\end{equation}
where $k$ is the spectral parameter and $\dagger$ denotes conjugate transpose. The compatibility condition $\psi_{xt}=\psi_{tx}$ is ascertained through the zero-curvature equation $\mathbf{X}_{t}-\mathbf{T}_{x}+[\mathbf{X},\mathbf{T}]=\mathbf{0}$.

\subsection{Jost solutions and scattering matrix}

Taking into account the Jost eigenfunctions as $x\to\pm\infty$, the spatial and temporal evolution of the solutions for the asymptotic Lax pair can be described as follows:
\begin{equation}\label{2.4}
\begin{split}
\psi_{x}=\mathbf{X}_{\pm}\psi, \quad \psi_{t}=\mathbf{T}_{\pm}\psi,
\end{split}
\end{equation}
where
\begin{subequations}\label{2.5}
\begin{align}
\lim_{x\rightarrow\pm\infty}{\mathbf{X}}=\mathbf{X}_{\pm}&=\mathrm{i}k\mathbf{J}+\mathrm{i}\mathbf{Q}_{\pm}, \\
\lim_{x\rightarrow\pm\infty}{\mathbf{T}}=\mathbf{T}_{\pm}&=4\mathrm{i}\sigma k^{3}\mathbf{J}-\mathrm{i}q_{0}^{2}\mathbf{J}
+k^{2}(4\mathrm{i}\sigma\mathbf{Q}_{\pm}+2\mathrm{i}\mathbf{J})
+k(2\mathrm{i}\mathbf{Q}_{\pm}-2\mathrm{i}\sigma\mathbf{J}\mathbf{Q}_{\pm}^{2})
-2\mathrm{i}\sigma\mathbf{Q}_{\pm}^{3}-\mathrm{i}\mathbf{J}\mathbf{Q}_{\pm}^{2}.
\end{align}
\end{subequations}
By the definition of $\mathbf{X}_{\pm}$ and $\mathbf{T}_{\pm}$ in~\eqref{2.5}, the eigenvalues of the corresponding matrix can be derived as follows
\begin{equation}\label{2.6}
\begin{split}
\mathbf{X}_{\pm,1}=-\mathrm{i}k, \quad
\mathbf{X}_{\pm,2,3}=\pm\mathrm{i}\lambda, \quad
\mathbf{T}_{\pm,1}=-\mathrm{i}(\lambda^{2}+k^{2}+4\sigma k^{3}), \quad
\mathbf{T}_{\pm,2,3}=\pm 2\mathrm{i}\lambda[k+\sigma(3k^{2}-\lambda^{2})],
\end{split}
\end{equation}
where
\begin{equation}\label{2.7}
\begin{split}
\lambda(k)=\sqrt{k^{2}-q_{0}^{2}}.
\end{split}
\end{equation}
Biondini~\cite{B1} et al. introduced the two-sheeted Riemann surface defined by~\eqref{2.7}, and $\lambda(k)$ is a single-valued function of $k$ that satisfies $\lambda(\pm q_{0})=0$. Then, the branch points are $k=\pm \mathrm{i}q_{0}$. We define the uniformization variable
\begin{equation}\label{2.8}
\begin{split}
z=k+\lambda,
\end{split}
\end{equation}
whose corresponding inverse map is given by
\begin{equation}\label{2.9}
\begin{split}
k(z)=\frac{1}{2}(z+\frac{q_{0}^{2}}{z}), \quad \lambda(z)=\frac{1}{2}(z-\frac{q_{0}^{2}}{z}).
\end{split}
\end{equation}
The relevant theories and property descriptions of the two-sheeted Riemann surface can be referred to in~\cite{B1}. Therefore, it can be defined that
\begin{equation}\label{2.10}
\everymath{\displaystyle}
\begin{split}
\mathbb{D}^{+}&=\left\{z \in \mathbb{C}: \mathfrak{Im} z>0\right\}, \quad
\mathbb{D}^{-}=\left\{z \in \mathbb{C}: \mathfrak{Im} z<0\right\}.
\end{split}
\end{equation}

The analytical regions of the eigenfunctions are determined by the sign of $\mathfrak{Im}\lambda$. Therefore, all $k$ dependencies will be rewritten as dependencies on $z$. The continuous spectrum of $k$ is given by $k\in \mathbb{R} \backslash (-q_{0}, q_{0})$. In the complex $z$-plane, the corresponding set is the whole real axis. Let us define a two-component vector $\mathbf{v}=(v_{1},v_{2})^{T}$ and its corresponding orthogonal vector as $\mathbf{v}^{\perp}=(v_{2},-v_{1})^{\dagger}$. Consider the eigenvector matrix of the asymptotic Lax pair~\eqref{2.4} as follows
\begin{equation}\label{2.11}
\everymath{\displaystyle}
\begin{split}
\mathbf{Y}_{\pm}(z)=\begin{pmatrix}
    \mathrm{i} & 0 &  \frac{q_{0}}{z}  \\
    \frac{\mathrm{i}\mathbf{q}_{\pm}}{z} &
    \frac{\mathbf{q}_{\pm}^{\bot}}{q_{0}} &
     \frac{\mathbf{q}_{\pm}}{q_{0}}
  \end{pmatrix}, \quad  \det\mathbf{Y}_{\pm}(z)=\mathrm{i}\rho(z), \quad
  \rho(z)=1-\frac{q_{0}^{2}}{z^{2}},
\end{split}
\end{equation}
where
\begin{equation}\label{2.12}
\everymath{\displaystyle}
\begin{split}
 \quad
\mathbf{Y}_{\pm}^{-1}(z)=\frac{1}{\mathrm{i}\rho(z)}\begin{pmatrix}
    1 & -\frac{\mathbf{q}_{\pm}^{\dagger}}{z}  \\
    0 & \frac{\mathrm{i}\rho(z)}{q_{0}}(\mathbf{q}_{\pm}^{\bot})^{\dagger}  \\
    -\frac{\mathrm{i}q_{0}}{z} &
    \frac{\mathrm{i}\mathbf{q}_{\pm}^{\dagger}}{q_{0}}  \\
  \end{pmatrix}, \quad  \det\mathbf{Y}_{\pm}^{-1}(z)=\frac{1}{\mathrm{i}\rho(z)},
\end{split}
\end{equation}
so that
\begin{equation}\label{2.13}
\begin{split}
\mathbf{X}_{\pm}\mathbf{Y}_{\pm}=\mathrm{i}\mathbf{Y}_{\pm}\mathbf{\Lambda}_{1}, \quad
\mathbf{T}_{\pm}\mathbf{Y}_{\pm}=\mathrm{i}\mathbf{Y}_{\pm}\mathbf{\Lambda}_{2},
\end{split}
\end{equation}
where $\mathbf{Y}_{\pm}=\mathbf{Y}_{\pm}(z)$, the definitions of $\mathbf{\Lambda}_{1}=\mathbf{\Lambda}_{1}(z)$ and $\mathbf{\Lambda}_{2}=\mathbf{\Lambda}_{2}(z)$ are as follows
\begin{subequations}\label{2.14}
\begin{align}
\mathbf{\Lambda}_{1}(z)&=\operatorname{diag}\left( \lambda,-k,-\lambda \right), \\
\mathbf{\Lambda}_{2}(z)&=\operatorname{diag}\left( 2\lambda[k+\sigma(3k^{2}-\lambda^{2})],
-(\lambda^{2}+k^{2}+4\sigma k^{3}), -2\lambda[k+\sigma(3k^{2}-\lambda^{2})] \right).
\end{align}
\end{subequations}
Since the NBCS are constant, the relationship $[\mathbf{X}_{\pm},\mathbf{T}_{\pm}]=\mathbf{0}$ can be calculated using expression~\eqref{2.5}. Therefore, $\mathbf{X}_{\pm}$ and $\mathbf{T}_{\pm}$ have a common eigenvector. Then, the Jost solutions $\psi(z;x,t)$ of the Lax pair~\eqref{2.4} on $z\in\mathbb{R}$ satisfying the boundary conditions
\begin{equation}\label{2.15}
\begin{split}
\psi_{\pm}(z;x,t)=\mathbf{Y}_{\pm}(z) \mathbf{e}^{\mathrm{i}\mathbf{\Delta}(z;x,t)}+o(1), \quad x\rightarrow\pm\infty,
\end{split}
\end{equation}
where $\mathbf{\Delta}(z;x,t)=\operatorname{diag} \left( \delta_{1},\delta_{2},-\delta_{1} \right)$ and
\begin{equation}\label{2.16}
\begin{split}
\delta_{1}=\delta_{1}(z;x,t)=\lambda x+2\lambda[k+\sigma(3k^{2}-\lambda^{2})]t, \quad
\delta_{2}=\delta_{2}(z;x,t)=-kx-(\lambda^{2}+k^{2}+4\sigma k^{3})t.
\end{split}
\end{equation}

Consistently, we define the modified eigenfunctions
\begin{equation}\label{2.17}
\begin{split}
\nu_{\pm}(z;x,t)=\psi_{\pm}(z;x,t) \mathbf{e}^{-\mathrm{i}\mathbf{\Delta}(z;x,t)},
\end{split}
\end{equation}
so that $\lim_{x\rightarrow\pm\infty}{\nu_{\pm}(z;x,t)}=\mathbf{Y}_{\pm}(z)$. We perform factor decomposition on the asymptotic behavior of the potential and rewrite the Lax pair~\eqref{2.4} as
\begin{equation}\label{2.18}
\begin{split}
(\psi_{\pm})_{x}&=\mathbf{X}_{\pm}\psi_{\pm}+(\mathbf{X}-\mathbf{X}_{\pm})\psi_{\pm}, \quad
(\psi_{\pm})_{t}=\mathbf{T}_{\pm}\psi_{\pm}+(\mathbf{T}-\mathbf{T}_{\pm})\psi_{\pm},
\end{split}
\end{equation}
where $\psi_{\pm}=\psi_{\pm}(z;x,t)$ and the following systems can be exported through the Lax pair~\eqref{2.18}
\begin{subequations}\label{2.19}
\begin{align}
(\mathbf{Y}_{\pm}^{-1}\nu_{\pm})_{x}&=
\mathrm{i}\mathbf{\Lambda}_{1}\mathbf{Y}_{\pm}^{-1}\nu_{\pm}
-\mathrm{i}\mathbf{Y}_{\pm}^{-1}\nu_{\pm}\mathbf{\Lambda}_{1}
+\mathbf{Y}_{\pm}^{-1}(\mathbf{X}-\mathbf{X}_{\pm})\nu_{\pm} ,  \\
(\mathbf{Y}_{\pm}^{-1}\nu_{\pm})_{t}&=
\mathrm{i}\mathbf{\Lambda}_{2}\mathbf{Y}_{\pm}^{-1}\nu_{\pm}-
\mathrm{i}\mathbf{Y}_{\pm}^{-1}\nu_{\pm}\mathbf{\Lambda}_{2}
+\mathbf{Y}_{\pm}^{-1}(\mathbf{T}-\mathbf{T}_{\pm})\nu_{\pm},
\end{align}
\end{subequations}
where $\mathbf{Y}_{\pm}^{-1}=\mathbf{Y}_{\pm}^{-1}(z)$ and $\nu_{\pm}=\nu_{\pm}(z;x,t)$. Subsequently, the systems~\eqref{2.19} can be expressed in complete differential form.
\begin{equation}\label{2.20}
\begin{split}
\mathrm{d}[\mathbf{e}^{-\mathrm{i}\mathbf{\Delta}}\mathbf{Y}_{\pm}^{-1}\nu_{\pm}
\mathbf{e}^{\mathrm{i}\mathbf{\Delta}}]&=\mathbf{e}^{-\mathrm{i}\mathbf{\Delta}}
[\mathbf{Y}_{\pm}^{-1}(\mathbf{X}-\mathbf{X}_{\pm})\nu_{\pm}\,\mathrm{d}x
+\mathbf{Y}_{\pm}^{-1}(\mathbf{T}-\mathbf{T}_{\pm})\nu_{\pm}\,\mathrm{d}t]
\mathbf{e}^{\mathrm{i}\mathbf{\Delta}},
\end{split}
\end{equation}
where $\mathbf{\Delta}=\mathbf{\Delta}(z;x,t)$. The choice of integration path is independent of $t$, it has been confirmed that the spectral problem concerning $\nu_{\pm}(z;x,t)$ is equivalent to the Volterra integral equations.
\begin{equation}\label{2.21}
\begin{split}
\nu_{\pm}(z;x,t)&=\mathbf{Y}_{\pm}+\int_{\pm\infty}^{x}\mathbf{Y}_{\pm}
\mathbf{e}^{\mathrm{i}(x-y)\mathbf{\Lambda}_{1}}
[\mathbf{Y}_{\pm}^{-1}(\mathbf{X}(z;y,t)-\mathbf{X}_{\pm}(z;y,t))\nu_{\pm}(z;y,t)]
\mathbf{e}^{-\mathrm{i}(x-y)\mathbf{\Lambda}_{1}}\,\mathrm{d}y.
\end{split}
\end{equation}
The proofs of Theorem~\ref{thm:1} and Theorem~\ref{thm:2} are based on the proofs of related Theorems given in~\cite{B1}.

\begin{theorem}\label{thm:1}
Suppose that $\mathbf{q}(\cdot, t)-\mathbf{q}_{-} \in L^{1}(-\infty, \sigma_{3})$ $(\mathbf{q}(\cdot, t)-\mathbf{q}_{+} \in L^{1}(\sigma_{3}, \infty))$ for every fixed $\sigma_{3}\in\mathbb{R}$, the ensuing columns of $\nu_{-}(z;x,t)$ $(\nu_{+}(z;x,t))$ fulfill the requisite properties
\begin{equation}\label{2.22}
\begin{split}
\nu_{-,1}(z;x,t) \;\; \text{and} \;\; \nu_{+,3}(z;x,t) \; : \; z\in \mathbb{D}^{-},  \qquad
\nu_{-,3}(z;x,t) \;\; \text{and} \;\; \nu_{+,1}(z;x,t) \; : \; z\in \mathbb{D}^{+}.
\end{split}
\end{equation}
\end{theorem}

Assuming $\psi(z;x,t)$ is a solution of the Lax pair~\eqref{2.1}, it can be obtained that
\begin{subequations}\label{2.23}
\begin{align}
\partial_{x}[\det \psi_{\pm}]&=\operatorname{tr} \mathbf{X} \, \det\psi_{\pm}
=-\mathrm{i}k \, \det\psi_{\pm},  \\
\partial_{t}[\det\psi_{\pm}]&=\operatorname{tr} \mathbf{T} \, \det\psi_{\pm}
=-\mathrm{i}(\lambda^{2}+k^{2}+4\sigma k^{3}) \, \det\psi_{\pm}.
\end{align}
\end{subequations}
From Abel's theorem, it can be inferred that $\partial x[\det\nu_{\pm}]=0$ and
$\partial t[\det\nu_{\pm}]=0$. Then~\eqref{2.15} implies
\begin{equation}\label{2.24}
\begin{split}
\det\psi_{\pm}(z;x,t)=\mathrm{i}\rho(z)\mathrm{e}^{\mathrm{i}\delta_{2}(z;x,t)}, \quad (x,t)\in\mathbb{R}^{2}, \quad z\in \mathbb{R}\backslash\{\pm q_{0}\}.
\end{split}
\end{equation}

The scattering matrix $\mathbf{S}(z)$ and $\mathbf{H}(z)$ are characterized through the following definition
\begin{equation}\label{2.25}
\begin{split}
\psi_{+}(z;x,t)=\psi_{-}(z;x,t)\mathbf{H}(z), \quad z\in \mathbb{R} \backslash \{\pm q_{0}\},
\end{split}
\end{equation}
where $\mathbf{H}(z)=(h_{ij}(z))$. According to the definition, $\mathbf{H}(z)$ is independent of $x$ and $t$ variables. Using~\eqref{2.24} and~\eqref{2.25} can indicate that
\begin{equation}\label{2.26}
\begin{split}
\det\mathbf{H}(z)=1, \quad z\in \mathbb{R} \backslash \{\pm q_{0}\}.
\end{split}
\end{equation}
Similarly, define $\mathbf{S}(z)=\mathbf{H}^{-1}(z)=(s_{ij}(z))$.

\begin{theorem}\label{thm:2}
According to the same assumption in Theorem~\ref{thm:1}, the scattering coefficients have the following properties
\begin{equation}\label{2.27}
\begin{split}
s_{33}(z) \;\; \text{and} \;\; h_{11}(z) \; : \; z\in \mathbb{D}^{+},  \qquad
s_{11}(z) \;\; \text{and} \;\; h_{33}(z) \; : \; z\in \mathbb{D}^{-}.
\end{split}
\end{equation}
\end{theorem}

\subsection{Adjoint problem}

Since the $\nu_{\pm,2}(z;x,t)$ is not analytical, then solving the inverse problem requires handling this non analytical term. In order to set up the scattering problem, it is essential to possess a fully analytic function. We now turn our attention to the so-called ``adjoint" Lax pair, which is a key component in this context.
\begin{equation}\label{2.28}
\begin{split}
\widetilde{\psi}_{x}&=\widetilde{\mathbf{X}}\widetilde{\psi}, \quad
\widetilde{\psi}_{t}=\widetilde{\mathbf{T}}\widetilde{\psi},
\end{split}
\end{equation}
where
\begin{equation}\label{2.29}
\begin{split}
\widetilde{\mathbf{X}}=\widetilde{\mathbf{X}}(z;x,t)&=-\mathrm{i}k\mathbf{J}-\mathrm{i}\mathbf{Q}^{*}, \\
\widetilde{\mathbf{T}}=\widetilde{\mathbf{T}}(z;x,t)&=-4\mathrm{i}\sigma k^{3}\mathbf{J}+\mathrm{i}q_{0}^{2}\mathbf{J}
-k^{2}(4\mathrm{i}\sigma\mathbf{Q}^{*}+2\mathrm{i}\mathbf{J})
+k(2\mathrm{i}\sigma\mathbf{J}(\mathbf{Q}^{*})^{2}-2\mathrm{i}\mathbf{Q}^{*}
-2\sigma\mathbf{Q}_{x}^{*}\mathbf{J}) \\
&+2\mathrm{i}\sigma(\mathbf{Q}^{*})^{3}+\mathrm{i}\mathbf{J}(\mathbf{Q}^{*})^{2}
+\mathrm{i}\sigma\mathbf{Q}_{xx}^{*}-\mathbf{Q}_{x}^{*}\mathbf{J}
+\sigma[\mathbf{Q}^{*},\mathbf{Q}_{x}^{*}],
\end{split}
\end{equation}
where $\widetilde{\psi}=\widetilde{\psi}(z,x,t)$, with $\widetilde{\mathbf{X}}=\mathbf{X}^{*}$ and $\widetilde{\mathbf{T}}=\mathbf{T}^{*}$ for all $z\in \mathbb{R}$. The following proposition can be directly proven through properties $\mathbf{J}\mathbf{Q}=-\mathbf{Q}\mathbf{J}$, $\mathbf{J}\mathbf{Q}^{*}=-\mathbf{Q}^{*}\mathbf{J}$, $\mathbf{J}\mathbf{Q}^{*}=\mathbf{Q}^{T}\mathbf{J}$ and the identity in~\cite{B1}.

\begin{proposition}\label{pro:1}
If $\widetilde{\mathbf{v}}_{2}(z;x,t)$ and $\widetilde{\mathbf{v}}_{3}(z;x,t)$ are two arbitrary solutions of the ``adjoint" Lax pair~\eqref{2.28}, while ``$\times$" denotes the usual cross product, then
\begin{equation}\label{2.30}
\begin{split}
\mathbf{v}_{1}(z;x,t)=\mathrm{e}^{\mathrm{i}\delta_{2}(z;x,t)}\mathbf{J}
[\widetilde{\mathbf{v}}_{2}(z;x,t) \times \widetilde{\mathbf{v}}_{3}(z;x,t)],
\end{split}
\end{equation}
is a solution of the Lax pair~\eqref{2.1}.
\end{proposition}

As $x\to\pm\infty$, the behavior of the solutions derived from the ``adjoint" Lax pair~\eqref{2.28} will approach an asymptotic state in terms of both spatial and temporal.
\begin{equation}\label{2.31}
\begin{split}
\widetilde{\psi}_{x}=\widetilde{\mathbf{X}}_{\pm}\widetilde{\psi}, \quad \widetilde{\psi}_{t}=\widetilde{\mathbf{T}}_{\pm}\widetilde{\psi},
\end{split}
\end{equation}
where
\begin{subequations}\label{2.32}
\begin{align}
\lim_{x\rightarrow\pm\infty}{\widetilde{\mathbf{X}}}&=\widetilde{\mathbf{X}}_{\pm}
=-\mathrm{i}k\mathbf{J}-\mathrm{i}\mathbf{Q}_{\pm}^{*}, \\
\lim_{x\rightarrow\pm\infty}{\widetilde{\mathbf{T}}}&=\widetilde{\mathbf{T}}_{\pm}
=-4\mathrm{i}\sigma k^{3}\mathbf{J}+\mathrm{i}q_{0}^{2}\mathbf{J}
-k^{2}(4\mathrm{i}\sigma\mathbf{Q}_{\pm}^{*}+2\mathrm{i}\mathbf{J})
+k(2\mathrm{i}\sigma\mathbf{J}(\mathbf{Q}_{\pm}^{*})^{2}-2\mathrm{i}\mathbf{Q}_{\pm}^{*})
+2\mathrm{i}\sigma(\mathbf{Q}_{\pm}^{*})^{3}+\mathrm{i}\mathbf{J}(\mathbf{Q}_{\pm}^{*})^{2}.
\end{align}
\end{subequations}
The eigenvalues of $\widetilde{\mathbf{X}}_{\pm}$ are $\mathrm{i}k$ and $\pm\mathrm{i}\lambda$,
the eigenvalues of $\widetilde{\mathbf{T}}_{\pm}$ are $\mathrm{i}(\lambda^{2}+k^{2}+4\sigma k^{3})$ and $\pm 2\mathrm{i}\lambda[k+\sigma(3k^{2}-\lambda^{2})]$. Additionally, properties $\widetilde{\mathbf{Y}}_{\pm}(z)=\mathbf{Y}_{\pm}^{*}(z)$ and $\det \widetilde{\mathbf{Y}}_{\pm}(z)=-\mathrm{i}\rho(z)$ are present. It can be straightforwardly determined that $\widetilde{\mathbf{X}}_{\pm}$ and $\widetilde{\mathbf{T}}_{\pm}$ fulfill the subsequent conditions.
\begin{equation}\label{2.33}
\begin{split}
\widetilde{\mathbf{X}}_{\pm}\widetilde{\mathbf{Y}}_{\pm}(z)&=
-\mathrm{i}\widetilde{\mathbf{Y}}_{\pm}(z)\mathbf{\Lambda}_{1}(z), \quad
\widetilde{\mathbf{T}}_{\pm}\widetilde{\mathbf{Y}}_{\pm}(z)=
-\mathrm{i}\widetilde{\mathbf{Y}}_{\pm}(z)\mathbf{\Lambda}_{2}(z).
\end{split}
\end{equation}

Similarly, the Jost solutions of the ``adjoint" Lax pair~\eqref{2.28}
\begin{equation}\label{2.34}
\begin{split}
\widetilde{\psi}_{\pm}(z;x,t)=\widetilde{\mathbf{Y}}_{\pm}(z) \mathbf{e}^{-\mathrm{i}\mathbf{\Delta}(z;x,t)}+o(1), \quad x\rightarrow\pm\infty,
\quad z\in \mathbb{R}.
\end{split}
\end{equation}
Introducing the modified Jost solutions
\begin{equation}\label{2.35}
\begin{split}
\widetilde{\nu}_{\pm}(z;x,t)=\widetilde{\psi}_{\pm}(z;x,t) \mathbf{e}^{\mathrm{i}\mathbf{\Delta}(z;x,t)},
\end{split}
\end{equation}
the subsequent columns of the function $\widetilde{\nu}_{\pm}(z;x,t)$ adhere to the ensuing properties.
\begin{equation}\label{2.36}
\begin{split}
\widetilde{\nu}_{-,3}(z;x,t) \;\; \text{and} \;\; \widetilde{\nu}_{+,1}(z;x,t) \; : \; z\in \mathbb{D}^{-},  \qquad
\widetilde{\nu}_{-,1}(z;x,t) \;\; \text{and} \;\; \widetilde{\nu}_{+,3}(z;x,t) \; : \; z\in \mathbb{D}^{+}.
\end{split}
\end{equation}

The modified Jost solutions imply that the columns of $\widetilde{\psi}_{\pm}(z;x,t)$ exhibit analogous properties of analyticity and boundedness. In a similar fashion, the ``adjoint" scattering matrix can likewise be defined as follows:
\begin{equation}\label{2.37}
\begin{split}
\widetilde{\psi}_{+}(z;x,t)=\widetilde{\psi}_{-}(z;x,t)\widetilde{\mathbf{H}}(z),
\end{split}
\end{equation}
where $\widetilde{\mathbf{H}}(z)=(\widetilde{h}_{ij}(z))$. Similarly, define $\widetilde{\mathbf{S}}(z)=\widetilde{\mathbf{H}}^{-1}(z)=(\widetilde{s}_{ij}(z))$. The following scattering coefficients satisfy the following properties:
\begin{equation}\label{2.38}
\begin{split}
\widetilde{s}_{33}(z) \;\; \text{and} \;\; \widetilde{h}_{11}(z) \; : \; z\in \mathbb{D}^{-}, \qquad
\widetilde{s}_{11}(z) \;\; \text{and} \;\; \widetilde{h}_{33}(z) \; : \; z\in \mathbb{D}^{+}.
\end{split}
\end{equation}

In order to fully construct the analytical eigenfunctions, two new solutions for the original Lax pair~\eqref{2.1} are defined
\begin{equation}\label{2.39}
\begin{split}
\gamma=-\frac{\mathrm{i}\mathrm{e}^{\mathrm{i}\delta_{2}}\mathbf{J}
[\widetilde{\psi}_{-,3} \times \widetilde{\psi}_{+,1}]}{\rho(z)}, \quad
z\in \mathbb{D}^{-},  \quad
\widetilde{\gamma}=-\frac{\mathrm{i}\mathrm{e}^{\mathrm{i}\delta_{2}}
\mathbf{J}[\widetilde{\psi}_{-,1} \times \widetilde{\psi}_{+,3}]}{\rho(z)}, \quad
z\in \mathbb{D}^{+},
\end{split}
\end{equation}
where $\gamma=\gamma(z;x,t)$, $\widetilde{\gamma}=\widetilde{\gamma}(z;x,t)$ and $\widetilde{\psi}_{\pm,j}=\widetilde{\psi}_{\pm,j}(z;x,t)$ for $j=1,3$. Then we can directly derive the following three conclusions:

\begin{corollary}\label{cor:1}
For all cyclic indices $j$, $l$ and $m$ with $z\in \mathbb{R}$,
\begin{equation}\label{2.40}
\begin{split}
\psi_{\pm,j}=-\frac{\mathrm{i}\mathrm{e}^{\mathrm{i}\delta_{2}}\mathbf{J}
[\widetilde{\psi}_{\pm,l} \times \widetilde{\psi}_{\pm,m}]}{\rho_{j}(z)}, \quad
\widetilde{\psi}_{\pm,j}=\frac{\mathrm{i}\mathrm{e}^{-\mathrm{i}\delta_{2}}\mathbf{J}
[\psi_{\pm,l} \times \psi_{\pm,m}]}{\rho_{j}(z)},
\end{split}
\end{equation}
where
\begin{equation}\label{2.41}
\begin{split}
\rho_{1}(z)=-1, \quad \rho_{2}(z)=\rho(z), \quad \rho_{3}(z)=1.
\end{split}
\end{equation}
\end{corollary}

\begin{corollary}\label{cor:2}
The scattering matrix $\mathbf{S}(z)$ and $\widetilde{\mathbf{S}}(z)$ are related as follows
\begin{equation}\label{2.42}
\begin{split}
\widetilde{\mathbf{S}}^{-1}(z)=\mathbf{J}_{1}(z)\mathbf{S}^{T}(z)\mathbf{J}_{1}^{-1}(z), \quad  \mathbf{J}_{1}(z)=\operatorname{diag}\,(-1,\rho(z),1).
\end{split}
\end{equation}
\end{corollary}

\begin{corollary}\label{cor:3}
The Jost eigenfunctions exhibit the following decompositions for $z\in \mathbb{R}$
\begin{equation}\label{2.43}
\begin{split}
\psi_{-,2}=\frac{s_{32}(z)\psi_{-,3}-\widetilde{\gamma}}{s_{33}(z)}
=\frac{s_{12}(z)\psi_{-,1}+\gamma}{s_{11}(z)}, \quad
\psi_{+,2}=\frac{h_{12}(z)\psi_{+,1}-\widetilde{\gamma}}{h_{11}(z)}
=\frac{h_{32}(z)\psi_{+,3}+\gamma}{h_{33}(z)}.
\end{split}
\end{equation}
\end{corollary}

Furthermore, the modified auxiliary eigenfunctions are delineated as follows:
\begin{equation}\label{2.44}
\begin{split}
d(z;x,t)=\gamma\mathrm{e}^{-\mathrm{i}\delta_{2}}, \quad z\in \mathbb{D}^{-}, \qquad
\widetilde{d}(z;x,t)=\widetilde{\gamma}\mathrm{e}^{-\mathrm{i}\delta_{2}}, \quad z\in \mathbb{D}^{+}.
\end{split}
\end{equation}

\subsection{Symmetries}

Compared with the equation with the initial value condition of ZBCs, when dealing with the equation with NZBCs, the corresponding symmetry becomes complicated due to the existence of Riemann surface. Hence, it is imperative to recognize that the symmetry inherent in the potential within the Lax pair engenders the corresponding symmetry in the scattering data. To ensure symmetry, it is essential to contemplate the subsequent involutions: $z\mapsto z^{*}$ and $z\mapsto -q_{0}^{2}/z$.

\subsubsection{First symmetry}

\begin{proposition}\label{pro:2}
If $\psi(z;x,t)$ is a non-singular solution of the Lax pair~\eqref{2.1}, so is
\begin{equation}\label{2.45}
\begin{split}
\mathbf{v}_{4}(z;x,t)=\mathbf{J}[\psi^{\dagger}(z^{*};x,t)]^{-1}.
\end{split}
\end{equation}
\end{proposition}

By using property $[\mathbf{e}^{\mathrm{i}\mathbf{\Delta}(z^{*};x,t)}]^{\dagger}
=\mathbf{e}^{-\mathrm{i}\mathbf{\Delta}(z;x,t)}$, the Jost eigenfunctions exhibit the specific symmetries
\begin{equation}\label{2.46}
\begin{split}
\psi_{\pm}(z;x,t)=\mathbf{J}[\psi_{\pm}^{\dagger}(z^{*};x,t)]^{-1}\mathbf{J}_{2}(z), \quad
\mathbf{J}_{2}(z)=\operatorname{diag}\,(\rho(z),-1,-\rho(z)), \quad z\in \mathbb{R}.
\end{split}
\end{equation}
Subsequently, employing the Schwarz reflection principle, we obtain the subsequent results:
\begin{subequations}\label{2.47}
\begin{align}
\psi_{-,1}^{*}(z^{*})&=\frac{\mathrm{i}\mathbf{J}\mathrm{e}^{-\mathrm{i}\delta_{2}(z)}
[\widetilde{\gamma}(z) \times \psi_{-,3}(z)]}{s_{33}(z)}, \quad \mathfrak{Im}z\geq0,  \\
\psi_{+,1}^{*}(z^{*})&=\frac{\mathrm{i}\mathbf{J}\mathrm{e}^{-\mathrm{i}\delta_{2}(z)}
[\gamma(z) \times \psi_{+,3}(z)]}{-h_{33}(z)}, \quad \mathfrak{Im}z\leq0,  \\
\psi_{-,3}^{*}(z^{*})&=\frac{\mathrm{i}\mathbf{J}\mathrm{e}^{-\mathrm{i}\delta_{2}(z)}
[\gamma(z) \times \psi_{-,1}(z)]}{-s_{11}(z)}, \quad \mathfrak{Im}z\leq0,  \\
\psi_{+,3}^{*}(z^{*})&=\frac{\mathrm{i}\mathbf{J}\mathrm{e}^{-\mathrm{i}\delta_{2}(z)}
[\widetilde{\gamma}(z) \times \psi_{+,1}(z)]}{h_{11}(z)}, \quad \mathfrak{Im}z\geq0.
\end{align}
\end{subequations}

In addition, using the scattering relationship~\eqref{2.25}, ~\eqref{2.46} and the property $\mathbf{J}_{2}(z)=-\rho(z)\mathbf{J}_{1}^{-1}(z)$, the scattering matrices are interrelated as follows:
\begin{equation}\label{2.48}
\begin{split}
\mathbf{S}^{\dagger}(z)=\mathbf{J}_{1}^{-1}(z)\mathbf{H}(z)\mathbf{J}_{1}(z), \quad
z\in \mathbb{R}.
\end{split}
\end{equation}
Accordingly, it can be deduced that for $z\in \mathbb{R}$
\begin{subequations}\label{2.49}
\begin{align}
h_{11}(z)&=s_{11}^{*}(z), \quad \quad \quad \,\,\,\, h_{12}(z)=-\frac{s_{21}^{*}(z)}{\rho(z)}, \quad  \quad  h_{13}(z)=-s_{31}^{*}(z),  \\
h_{21}(z)&=-\rho(z)s_{12}^{*}(z), \quad \, h_{22}(z)=s_{22}^{*}(z), \quad  \quad \,\,\,\,\,
h_{23}(z)=\rho(z)s_{32}^{*}(z),  \\
h_{31}(z)&=-s_{13}^{*}(z), \quad \quad \quad  h_{32}(z)=\frac{s_{23}^{*}(z)}{\rho(z)}, \quad  \quad \,\,\,\, h_{33}(z)=s_{33}^{*}(z).
\end{align}
\end{subequations}
Then, we can draw this conclusion
\begin{equation}\label{2.50}
\begin{split}
h_{11}(z)=s_{11}^{*}(z^{*}), \quad  \mathfrak{Im}z\geq0,  \qquad
h_{33}(z)=s_{33}^{*}(z^{*}), \quad  \mathfrak{Im}z\leq0.
\end{split}
\end{equation}

The property $\psi_{\pm}^{*}(z^{*};x,t)=\widetilde{\psi}_{\pm}(z;x,t)$ is obtained, so the following conditions are established
\begin{equation}\label{2.51}
\begin{split}
\psi_{\pm,1}^{*}(z^{*})=\widetilde{\psi}_{\pm,1}(z), \quad  \mathfrak{Im}z\lessgtr0, \qquad
\psi_{\pm,3}^{*}(z^{*})=\widetilde{\psi}_{\pm,3}(z), \quad  \mathfrak{Im}z\gtrless0.
\end{split}
\end{equation}
Then, through the property~\eqref{2.51} and new auxiliary eigenfunctions~\eqref{2.39}, we derive the following conclusion.

\begin{corollary}\label{cor:4}
The new auxiliary eigenfunctions~\eqref{2.39} adhere to the symmetry relations:
\begin{subequations}\label{2.52}
\begin{align}
\gamma(z)&=-\frac{\mathrm{i}\mathrm{e}^{\mathrm{i}\delta_{2}(z)}\mathbf{J}
[\psi_{-,3}^{*}(z^{*}) \times \psi_{+,1}^{*}(z^{*})]}{\rho(z)}, \quad
z\in \mathbb{D}^{-},   \\
\widetilde{\gamma}(z)&=-\frac{\mathrm{i}\mathrm{e}^{\mathrm{i}\delta_{2}(z)}\mathbf{J}
[\psi_{-,1}^{*}(z^{*}) \times \psi_{+,3}^{*}(z^{*})]}{\rho(z)}, \quad
z\in \mathbb{D}^{+}.
\end{align}
\end{subequations}
\end{corollary}

Furthermore, the symmetrical properties are also present.
\begin{equation}\label{2.53}
\begin{split}
\psi_{\pm,j}^{*}(z)=\frac{\mathrm{i}\mathrm{e}^{-\mathrm{i}\delta_{2}(z)}\mathbf{J}
[\psi_{\pm,l}(z) \times \psi_{\pm,m}(z)]}{\rho_{j}(z)}, \quad
z\in \mathbb{R},
\end{split}
\end{equation}
where $j$, $l$ and $m$ are cyclic indices.

\subsubsection{Second symmetry}

\begin{proposition}\label{pro:3}
If $\psi(z;x,t)$ is a non-singular solution of the Lax pair~\eqref{2.1}, so is
\begin{equation}\label{2.54}
\begin{split}
\mathbf{v}_{5}(z;x,t)=\psi(\frac{q_{0}^{2}}{z};x,t).
\end{split}
\end{equation}
\end{proposition}

The subsequent properties can be derived from the principle of progressiveness.
\begin{equation}\label{2.55}
\everymath{\displaystyle}
\begin{split}
\psi_{\pm}(z;x,t)=\psi_{\pm}(\frac{q_{0}^{2}}{z};x,t)\mathbf{J}_{3}(z), \quad \mathbf{J}_{3}(z)=\begin{pmatrix}
    0 & 0 &  -\frac{\mathrm{i}q_{0}}{z}  \\
    0 & 1 &  0  \\
    \frac{\mathrm{i}q_{0}}{z} & 0 & 0
  \end{pmatrix}, \quad z\in \mathbb{R}.
\end{split}
\end{equation}
Consistent with the previous discussion, the eigenfunctions exhibit the following analytical characteristics:
\begin{subequations}\label{2.56}
\begin{align}
\psi_{\pm,1}(z)&=\frac{\mathrm{i}q_{0}}{z}\psi_{\pm,3}(\frac{q_{0}^{2}}{z}), \quad
\mathfrak{Im}z\gtrless0, \quad \psi_{\pm,3}(z)=-\frac{\mathrm{i}q_{0}}{z}\psi_{\pm,1}(\frac{q_{0}^{2}}{z}), \quad
\mathfrak{Im}z\lessgtr0, \\
\psi_{\pm,2}(z)&=\psi_{\pm,2}(\frac{q_{0}^{2}}{z}), \quad z\in \mathbb{R},
\end{align}
\end{subequations}
using the scattering relationship~\eqref{2.25} and~\eqref{2.55}, the scattering matrices have the following relation
\begin{equation}\label{2.57}
\begin{split}
\mathbf{S}(\frac{q_{0}^{2}}{z})=\mathbf{J}_{3}(z)\mathbf{S}(z)\mathbf{J}_{3}^{-1}(z), \quad
z\in \mathbb{R}.
\end{split}
\end{equation}
Accordingly, it can be deduced that
\begin{subequations}\label{2.58}
\begin{align}
s_{11}(z)&=s_{33}(\frac{q_{0}^{2}}{z}), \quad \quad  \,\,\,\, s_{12}(z)=-\frac{\mathrm{i}z}{q_{0}}s_{32}(\frac{q_{0}^{2}}{z}), \quad \,\,\,
s_{13}(z)=-s_{31}(\frac{q_{0}^{2}}{z}),  \\
s_{21}(z)&=\frac{\mathrm{i}q_{0}}{z}s_{23}(\frac{q_{0}^{2}}{z}), \quad  s_{22}(z)=s_{22}(\frac{q_{0}^{2}}{z}), \qquad \quad \,\,
s_{23}(z)=-\frac{\mathrm{i}q_{0}}{z}s_{21}(\frac{q_{0}^{2}}{z}),  \\
s_{31}(z)&=-s_{13}(\frac{q_{0}^{2}}{z}), \quad \quad
s_{32}(z)=\frac{\mathrm{i}z}{q_{0}}s_{12}(\frac{q_{0}^{2}}{z}), \qquad
s_{33}(z)=s_{11}(\frac{q_{0}^{2}}{z}),
\end{align}
\end{subequations}
the analytical domain of the scattering coefficients
\begin{equation}\label{2.59}
\begin{split}
s_{11}(z)&=s_{33}(\frac{q_{0}^{2}}{z}), \quad
h_{33}(z)=h_{11}(\frac{q_{0}^{2}}{z}),  \quad \mathfrak{Im}z\leq0.
\end{split}
\end{equation}

The auxiliary eigenfunctions exhibit the following characteristics:
\begin{equation}\label{2.60}
\begin{split}
\widetilde{\gamma}(z)&=-\gamma(\frac{q_{0}^{2}}{z}),  \quad  \mathfrak{Im}z\geq0.
\end{split}
\end{equation}
We introduce the new reflections as follows:
\begin{subequations}\label{2.61}
\begin{align}
\beta_{1}(z)&=\frac{s_{13}(z)}{s_{11}(z)}=-\frac{h_{31}^{*}(z)}{h_{11}^{*}(z)}, \qquad \,\,\,\,\,\,
\beta_{1}(\frac{q_{0}^{2}}{z})=-\frac{s_{31}(z)}{s_{33}(z)}=\frac{h_{13}^{*}(z)}{h_{33}^{*}(z)}, \\
\beta_{2}(z)&=\frac{h_{21}(z)}{h_{11}(z)}=-\rho(z)\frac{s_{12}^{*}(z)}{s_{11}^{*}(z)}, \quad \beta_{2}(\frac{q_{0}^{2}}{z})=\frac{\mathrm{i}z}{q_{0}}\frac{h_{23}(z)}{h_{33}(z)}
=\frac{\mathrm{i}z\rho(z)}{q_{0}}\frac{s_{32}^{*}(z)}{s_{33}^{*}(z)}.
\end{align}
\end{subequations}

\section{Discrete spectrum and asymptotic behavior}
\label{s:Discrete spectrum and asymptotic behavior}

A direct link is established between the zeros of the scattering coefficients and the discrete eigenvalues, each signifying the presence of a bound state within the system~\cite{B1}. $C_{o}$ is a circle with a radius of $q_{0}$ centered on the origin of the complex $z$-plane. It has been established that discrete eigenvalues are excluded from the continuous spectrum, hence they are confined to the domain within the circle $C_{o}$. Furthermore, the self-adjoint property of the scattering problem ensures that the discrete eigenvalues $k$ must be real numbers, and there are no discrete eigenvalues in the continuous spectrum. Therefore, these discrete eigenvalues only appear in the circle $C_{o}$ on the $z$-plane.

\begin{proposition}\label{pro:4}
Let $\mathbf{v}(z;x,t)$ denote a nontrivial solution to the scattering problem in~\eqref{2.1}. If $\mathbf{v}(z;x,t)\in L^{2}(\mathbb{R})$, then $z\in C_{o}$.
\end{proposition}

In order to fully represent the characteristics of the inverse problem, it is necessary to consider the zeros of the analytical scattering coefficient outside the circle $C_{o}$. This view does not conflict with Proposition~\ref{pro:4}, the zero point of the analytical scattering coefficient outside the circle $C_{o}$ does not lead to the generation of bound states. The discrete spectrum is the set of all $z\in\mathbb{C}\backslash \mathbb{R}$ such that $h_{11}(z)=0$ or $s_{11}(z^{*})=0$, values for which the Jost eigenfunctions are in $L^{2}(\mathbb{R})$. Consequently, the existence of zeros for $h_{11}(z)$ within $C_{o}$ is permissible, and such zeros lead to eigenfunctions that do not exhibit decay towards both spatial infinities.

\subsection{Discrete spectrum}

To delve into the discrete spectrum, we define two $3\times3$ matrices
\begin{subequations}\label{3.1}
\begin{align}
\mathbf{\Psi}^{+}(z)&=(\psi_{+,1}(z),-\widetilde{\gamma}(z),\psi_{-,3}(z)), \quad z\in \mathbb{D}^{+}, \\
\mathbf{\Psi}^{-}(z)&=(\psi_{-,1}(z),\gamma(z),\psi_{+,3}(z)), \quad \,\,\,\, z\in \mathbb{D}^{-}.
\end{align}
\end{subequations}
By the decompositions~\eqref{2.43} and taking the determinant, we can have
\begin{subequations}\label{3.2}
\begin{align}
\det \mathbf{\Psi}^{+}(z)&
=\mathrm{i}\mathrm{e}^{\mathrm{i}\delta_{2}(z)}h_{11}(z)s_{33}(z)\rho(z), \quad \mathfrak{Im}z\geq0, \\
\det \mathbf{\Psi}^{-}(z)&
=\mathrm{i}\mathrm{e}^{\mathrm{i}\delta_{2}(z)}s_{11}(z)h_{33}(z)\rho(z), \quad \mathfrak{Im}z\leq0.
\end{align}
\end{subequations}
Nevertheless, the symmetries inherent in the scattering coefficients imply that these zeros are interrelated and not mutually exclusive.

\begin{proposition}[Off the circle $C_{o}$]\label{pro:5}
Suppose that $h_{11}(z)$ possesses a zero $\theta_{g}$ within the upper half plane of $z$, then
\begin{equation}\label{3.3}
\begin{split}
h_{11}(\theta_{g})=0 \Longleftrightarrow s_{11}(\theta_{g}^{*})=0 \Longleftrightarrow s_{33}(\frac{q_{0}^{2}}{\theta_{g}^{*}})=0 \Longleftrightarrow h_{33}(\frac{q_{0}^{2}}{\theta_{g}})=0.
\end{split}
\end{equation}
\end{proposition}

Therefore, it can be considered that the discrete eigenvalues $z_{g}$ on the circle $C_{o}$ is $\{ z_{g},z_{g}^{*} \}$ and the discrete eigenvalues $\theta_{g}$ off circle $C_{o}$ is $\{ \theta_{g},\theta_{g}^{*},q_{0}^{2}/\theta_{g},q_{0}^{2}/\theta_{g}^{*} \}$.

\begin{proposition}\label{pro:6}
If $\mathfrak{Im}\theta_{g}>0$ and $\theta_{g}\notin C_{o}$, then $\widetilde{\gamma}(\theta_{g};x,t)\neq\mathbf{0}$.
\end{proposition}

\begin{proposition}\label{pro:7}
Suppose $\mathfrak{Im}\theta_{g}>0$, the following conclusions are equivalent
\begin{itemize} \itemsep0.75pt
\everymath{\displaystyle}
  \item[(1)] $\widetilde{\gamma}(\theta_{g})=\mathbf{0} \Longleftrightarrow \gamma(\theta_{g}^{*})=\mathbf{0} \Longleftrightarrow \gamma(\frac{q_{0}^{2}}{\theta_{g}})=\mathbf{0} \Longleftrightarrow
      \widetilde{\gamma}(\frac{q_{0}^{2}}{\theta_{g}^{*}})=\mathbf{0}$.
  \item[(2)] $\psi_{-,3}(\theta_{g})$ and $\psi_{+,1}(\theta_{g})$ are linearly correlated. $\psi_{-,1}(\theta_{g}^{*})$ and $\psi_{+,3}(\theta_{g}^{*})$ are linearly correlated.
  \item[(3)] $\psi_{-,3}(\frac{q_{0}^{2}}{\theta_{g}^{*}})$ and $\psi_{+,1}(\frac{q_{0}^{2}}{\theta_{g}^{*}})$ are linearly correlated. $\psi_{-,1}(\frac{q_{0}^{2}}{\theta_{g}})$ and $\psi_{+,3}(\frac{q_{0}^{2}}{\theta_{g}})$ are linearly correlated.
\end{itemize}
\end{proposition}

Given the assumption that the discrete eigenvalues are simple and non-repeating, then we can derive the following two Theorems.
\begin{theorem}\label{thm:3}
Let $z_{g}$ be a zero of $h_{11}(z)$ in the upper half plane with $|z_{g}|=q_{0}$, then $\widetilde{\gamma}(z_{g})=\gamma(z_{g}^{*})=\mathbf{0}$, there exist constants $c_{g}$, and $\bar{c}_{g}$ such that
\begin{equation}\label{3.4}
\begin{split}
\psi_{+,1}(z_{g})&=c_{g}\psi_{-,3}(z_{g}), \quad \psi_{+,3}(z_{g}^{*})=\bar{c}_{g}\psi_{-,1}(z_{g}^{*}).
\end{split}
\end{equation}
\end{theorem}

\begin{theorem}\label{thm:4}
Let $\theta_{g}$ be a zero of $h_{11}(z)$ in the upper half plane with $|\theta_{g}|\neq q_{0}$,
then $|\theta_{g}|<q_{0}$ and $s_{33}(\theta_{g})\neq 0$, there exist constants $f_{g}$, $\hat{f}_{g}$, $\check{f}_{g}$ and $\bar{f}_{g}$ such that
\begin{subequations}\label{3.5}
\begin{align}
\psi_{+,1}(\theta_{g})&=\frac{f_{g}}{s_{33}(\theta_{g})}\widetilde{\gamma}(\theta_{g}), \quad
\gamma(\theta_{g}^{*})=\bar{f}_{g}\psi_{-,1}(\theta_{g}^{*}), \\
\psi_{+,3}(\frac{q_{0}^{2}}{\theta_{g}})&=\check{f}_{g}\gamma(\frac{q_{0}^{2}}{\theta_{g}}), \qquad \,\,\,\,\,  \widetilde{\gamma}(\frac{q_{0}^{2}}{\theta_{g}^{*}})
=\hat{f}_{n}\psi_{-,3}(\frac{q_{0}^{2}}{\theta_{g}^{*}}).
\end{align}
\end{subequations}
\end{theorem}

\begin{corollary}\label{cor:5}
Suppose that $h_{11}(z)$ has simple zeros $\{z_{g}\}_{g=1}^{G_{1}}$ on $C_{o}$, it can be inferred that the norming constants adhere to the symmetry relationship:
\begin{equation}\label{3.6}
\begin{split}
\bar{c}_{g}=-c_{g}, \quad c_{g}^{*}=\frac{s_{11}^{\prime}(z_{g}^{*})}{h_{33}^{\prime}(z_{g}^{*})}\bar{c}_{g},
\quad g=1,2, \ldots, G_{1}.
\end{split}
\end{equation}
\end{corollary}

\begin{corollary}\label{cor:6}
Suppose that $h_{11}(z)$ has zeros $\{\theta_{g}\}_{g=1}^{N_{2}}$ off $C_{o}$, it is known that the norming constants adhere to the symmetry relationship:
\begin{equation}\label{3.7}
\begin{split}
\bar{f}_{g}=-\frac{f_{g}^{*}}{\rho(\theta_{g}^{*})}, \quad \hat{f}_{g}=\frac{\mathrm{i}q_{0}}{\theta_{g}^{*}\rho(\theta_{g}^{*})}f_{g}^{*}, \quad
\check{f}_{g}=\frac{\mathrm{i}\theta_{g}}{q_{0}s_{33}(\theta_{g})}f_{g},
\quad g=1,2, \ldots, G_{2}.
\end{split}
\end{equation}
\end{corollary}

\subsection{Asymptotic behavior}

The asymptotic behavior of the modified Jost eigenfunctions for $z\rightarrow\infty$ and $z\rightarrow 0$ can be scrutinized utilizing the Wentzel-Kramers-Brillouin approximation technique. Specifically, when integrated with the differential equations~\eqref{2.19}, it exhibits asymptotic characteristics:
\begin{corollary}\label{cor:7}
The asymptotic expansion as $z\rightarrow\infty$ is delineated as follows:
\begin{equation}\label{3.8}
\everymath{\displaystyle}
\begin{split}
\nu_{\pm,1}(z)=\begin{pmatrix}
    \mathrm{i}   \\
    \mathrm{i}\mathbf{q}/z
  \end{pmatrix}+O(\frac{1}{z^{2}}), \quad
\nu_{\pm,3}(z)=\begin{pmatrix}
    \mathbf{q}^{\dagger}\mathbf{q}_{\pm}/q_{0}z  \\
    \mathbf{q}_{\pm}/q_{0}
  \end{pmatrix}+O(\frac{1}{z^{2}}).
\end{split}
\end{equation}
Similarly, the asymptotic expansion as $z\rightarrow 0$ is
\begin{equation}\label{3.9}
\everymath{\displaystyle}
\begin{split}
\nu_{\pm,1}(z)=\begin{pmatrix}
    \mathrm{i}\mathbf{q}^{\dagger}\mathbf{q}_{\pm}/q_{0}^{2} \\
    \mathrm{i}\mathbf{q}_{\pm}/z
  \end{pmatrix}+O(z), \quad
\nu_{\pm,3}(z)=\begin{pmatrix}
    q_{0}/z  \\
    \mathbf{q}/q_{0}
  \end{pmatrix}+O(z).
\end{split}
\end{equation}
\end{corollary}

By combining the modified auxiliary eigenfunctions~\eqref{2.44} with the asymptotic properties~\eqref{3.8} and~\eqref{3.9}, the following results are obtained
\begin{equation}\label{3.10}
\everymath{\displaystyle}
\begin{split}
d(z)&=\begin{pmatrix}
    \mathbf{q}^{\dagger}\mathbf{q}_{-}^{\bot}/q_{0}z  \\
    \mathbf{q}_{-}^{\bot}/q_{0}
  \end{pmatrix}+O(\frac{1}{z^{2}}), \quad
\widetilde{d}(z)=\begin{pmatrix}
    -\mathbf{q}^{\dagger}\mathbf{q}_{+}^{\bot}/q_{0}z  \\
    -\mathbf{q}_{+}^{\bot}/q_{0}
  \end{pmatrix}+O(\frac{1}{z^{2}}), \quad z\rightarrow\infty,
\end{split}
\end{equation}
and
\begin{equation}\label{3.11}
\everymath{\displaystyle}
\begin{split}
d(z)&=\begin{pmatrix}
    0  \\
    \mathbf{q}_{+}^{\bot}/q_{0}
  \end{pmatrix}+O(z), \quad
\widetilde{d}(z)=\begin{pmatrix}
    0  \\
    -\mathbf{q}_{-}^{\bot}/q_{0}
  \end{pmatrix}+O(z), \quad z\rightarrow0.
\end{split}
\end{equation}

\begin{corollary}\label{cor:8}
The asymptotic behavior of scattering matrix entries as $z\rightarrow\infty$ is
\begin{subequations}\label{3.12}
\begin{align}
s_{11}(z)&=1+O(\frac{1}{z}),  \quad
s_{22}(z)=h_{33}(z)=\frac{\mathbf{q}_{-}^{\dagger}\mathbf{q}_{+}}{q_{0}^{2}}+O(\frac{1}{z}), \quad
s_{32}(z)=\frac{\mathbf{q}_{+}^{\dagger}\mathbf{q}_{-}^{\bot}}{q_{0}^{2}}+O(\frac{1}{z}), \\
h_{11}(z)&=1+O(\frac{1}{z}),  \quad
h_{22}(z)=s_{33}(z)=\frac{\mathbf{q}_{+}^{\dagger}\mathbf{q}_{-}}{q_{0}^{2}}+O(\frac{1}{z}), \quad
h_{32}(z)=\frac{\mathbf{q}_{-}^{\dagger}\mathbf{q}_{+}^{\bot}}{q_{0}^{2}}+O(\frac{1}{z}), \\
s_{23}(z)&=\frac{(\mathbf{q}_{+}^{\bot})^{\dagger}\mathbf{q}_{-}}{q_{0}^{2}}+O(\frac{1}{z}), \quad
h_{23}(z)=\frac{(\mathbf{q}_{-}^{\bot})^{\dagger}\mathbf{q}_{+}}{q_{0}^{2}}+O(\frac{1}{z}),
\end{align}
\end{subequations}
the asymptotic behavior of other entries as $z\rightarrow\infty$ in the scattering matrix is $O(1/z)$. Similarly, one can show that as $z\rightarrow 0$
\begin{subequations}\label{3.13}
\everymath{\displaystyle}
\begin{align}
s_{11}(z)&=h_{22}(z)=\frac{\mathbf{q}_{+}^{\dagger}\mathbf{q}_{-}}{q_{0}^{2}}+O(z), \quad
s_{21}(z)=\frac{\mathrm{i}(\mathbf{q}_{+}^{\bot})^{\dagger}\mathbf{q}_{-}}{q_{0}z}+O(1), \quad
s_{33}(z)=1+O(z),  \\
s_{22}(z)&=h_{11}(z)=\frac{\mathbf{q}_{-}^{\dagger}\mathbf{q}_{+}}{q_{0}^{2}}+O(z), \quad
h_{21}(z)=\frac{\mathrm{i}(\mathbf{q}_{-}^{\bot})^{\dagger}\mathbf{q}_{+}}{q_{0}z}+O(1), \quad
h_{33}(z)=1+O(z),
\end{align}
\end{subequations}
the asymptotic behavior of other entries as $z\rightarrow 0$ in the scattering matrix is $O(z)$.
\end{corollary}

\subsection{Behavior at the branch points}

Next, we will analyze the characteristics of the Jost eigenfunctions and scattering matrix at the branching point $k=\pm q_{0}$. At the branching point, the  matrices $\mathbf{Y}_{\pm}(z)$ are degenerate. However, the term $\mathbf{Y}_{\pm}(z)
\mathbf{e}^{\mathrm{i}(x-y)\mathbf{\Lambda}_{1}(z)}\mathbf{Y}_{\pm}^{-1}(z)$ is finite at the branching point.
\begin{equation}\label{3.14}
\everymath{\displaystyle}
\begin{split}
\lim_{z\rightarrow\pm q_{0}}{\mathbf{Y}_{\pm}
\mathbf{e}^{\mathrm{i}(x-y)\mathbf{\Lambda}_{1}}\mathbf{Y}_{\pm}^{-1}}=\begin{pmatrix}
    1\pm \mathrm{i}q_{0}(x-y) & \mathrm{i}(y-x)\mathbf{q}_{\pm}^{\dagger}   \\
    \mathrm{i}(x-y)\mathbf{q}_{\pm} & \frac{1}{q_{0}^{2}}\left[(1\mp \mathrm{i}q_{0}(x-y))
    \mathbf{q}_{\pm}\mathbf{q}_{\pm}^{\dagger} +\mathrm{e}^{\mp \mathrm{i}q_{0}(x-y)}\mathbf{q}_{\pm}^{\perp}(\mathbf{q}_{\pm}^{\perp})^{\dagger}\right] \\
  \end{pmatrix}.
\end{split}
\end{equation}
For all $(x,t)\in\mathbb{R}^{2}$, it can be inferred from expression~\eqref{2.24} that $\det\psi_{\pm}(\pm q_{0};x,t)=0$. Then, the columns of $\psi_{\pm}(q_{0};x,t)$ and $\psi_{\pm}(-q_{0};x,t)$ are linearly dependent, and with the help of analytical properties~\eqref{2.56}, the following conditions are obtained
\begin{equation}\label{3.15}
\begin{split}
\psi_{\pm,1}(q_{0};x,t)=\mathrm{i}\psi_{\pm,3}(q_{0};x,t), \quad
\psi_{\pm,1}(-q_{0};x,t)=-\mathrm{i}\psi_{\pm,3}(-q_{0};x,t).
\end{split}
\end{equation}

We examine the characteristics of the scattering matrix $\mathbf{S}(z)$ in the vicinity of the branch points, which are expressible through Wronskian determinants. Consequently, we document the scattering coefficients in terms of these Wronskians.
\begin{equation}\label{3.16}
\begin{split}
s_{jl}(z)=\frac{z^{2}}{\mathrm{i}(z^{2}-q_{0}^{2})}W_{jl}(z)\mathrm{e}^{-\mathrm{i}\delta_{2}(z)}
=\frac{W_{jl}(z)}{\mathrm{i}\rho(z)}\mathrm{e}^{-\mathrm{i}\delta_{2}(z)},
\end{split}
\end{equation}
where
\begin{equation}\label{3.17}
\begin{split}
W_{jl}(z)=\det\,(\psi_{-,l}(z),\psi_{+,j+1}(z),\psi_{+,j+2}(z)),
\end{split}
\end{equation}
where $j+1$ and $j+2$ are calculated modulo 3. The scattering coefficients as $z\rightarrow\pm q_{0}$ are articulated as follows:
\begin{equation}\label{3.18}
\begin{split}
s_{ij}(z)&=\frac{s_{ij,\pm}}{z\mp q_{0}}+s_{ij,\pm}^{(o)}+O(z\mp q_{0}), \quad
z\in \mathbb{R} \backslash \{\pm q_{0}\},
\end{split}
\end{equation}
where
\begin{subequations}\label{3.19}
\begin{align}
s_{ij,\pm}&=\mp\frac{\mathrm{i}q_{0}}{2}W_{ij}(\pm q_{0};x,t)
\exp\left[ \mathrm{i}q_{0}^{2}t \pm \mathrm{i}q_{0}(x+4\sigma q_{0}^{2}t) \right], \\
s_{ij,\pm}^{(o)}&=\left[\mp \frac{\mathrm{i}q_{0}}{2}
\frac{\partial}{\partial z} \left. W_{ij}(z;x,t) \right|_{z=\pm q_{0}}-\mathrm{i}W_{ij}(\pm q_{0};x,t) \right]
\exp\left[ \mathrm{i}q_{0}^{2}t \pm \mathrm{i}q_{0}(x+4\sigma q_{0}^{2}t) \right].
\end{align}
\end{subequations}

Subsequently, the asymptotic series for $\mathbf{S}(z)$ in the vicinity of the branch point can be articulated in the following manner:
\begin{equation}\label{3.20}
\begin{split}
\mathbf{S}(z)=\frac{\mathbf{S}_{\pm}}{z\mp q_{0}}+\mathbf{S}_{\pm}^{(o)}+O(z\mp q_{0}), \quad
\mathbf{S}_{\pm}^{(o)}=\left(s_{ij,\pm}^{(o)}\right),
\end{split}
\end{equation}
where
\begin{equation}\label{3.21}
\begin{split}
\mathbf{S}_{\pm}=s_{11,\pm}\begin{pmatrix}
    1 & 0 & \mp \mathrm{i}  \\
    0 & 0 & 0  \\
    \mp \mathrm{i} & 0 & -1
  \end{pmatrix}+s_{12,\pm}\begin{pmatrix}
    0 & 1 & 0  \\
    0 & 0 & 0  \\
    0 & \mp \mathrm{i} & 0
  \end{pmatrix}.
\end{split}
\end{equation}
The asymptotic behavior of the reflection coefficient~\eqref{2.61} at the branching point can be directly obtained through~\eqref{3.21} and the symmetry~\eqref{2.49}
\begin{equation}\label{3.22}
\begin{split}
\lim_{z\rightarrow\pm q_{0}}\beta_{1}(z)=\mp \mathrm{i}, \quad
\lim_{z\rightarrow\pm q_{0}}\beta_{2}(z)=0.
\end{split}
\end{equation}

\section{Inverse problem}
\label{s:Inverse problem}

In general, the IST is represented by a suitable RH problem, and then its different properties can be studied. Therefore, the meromorphic eigenfunctions in the upper half $z$-plane are related to the meromorphic eigenfunctions in the lower half $z$-plane with the help of a suitable jump condition.

\subsection{Riemann-Hilbert problem}

To formulate the matrix RH problem, it is essential to establish suitable transition conditions that define the behavior of eigenfunctions, which are characterized by their meromorphic nature within the specified domain. Given that certain Jost eigenfunctions lack analytic properties, it is necessary to define new modified meromorphic functions in the corresponding regions.

\begin{proposition}\label{pro:8}
Define the piecewise meromorphic function $\mathbf{R}^{\pm}(z;x,t)=(\mathbf{r}_{1}^{\pm},\mathbf{r}_{2}^{\pm},\mathbf{r}_{3}^{\pm})$, where
\begin{subequations}\label{4.1}
\everymath{\displaystyle}
\begin{align}
\mathbf{R}^{+}(z;x,t)&=\mathbf{\Psi}^{+}\mathbf{e}^{-\mathrm{i}\mathbf{\Delta}}
\operatorname{diag} \left( \frac{1}{h_{11}(z)},\frac{1}{s_{33}(z)},1 \right)=
\left[ \frac{\nu_{+,1}}{h_{11}(z)}, -\frac{\widetilde{d}}{s_{33}(z)}, \nu_{-,3} \right], \quad z\in \mathbb{D}^{+}, \\
\mathbf{R}^{-}(z;x,t)&=\mathbf{\Psi}^{-}\mathbf{e}^{-\mathrm{i}\mathbf{\Delta}}
\operatorname{diag} \left( 1, \frac{1}{s_{11}(z)},\frac{1}{h_{33}(z)} \right)=
\left[ \nu_{-,1}, \frac{d}{s_{11}(z)}, \frac{\nu_{+,3}}{h_{33}(z)}, \right], \quad \, z\in \mathbb{D}^{-},
\end{align}
\end{subequations}
where $\mathbf{\Psi}^{\pm}=\mathbf{\Psi}^{\pm}(z;x,t)$, $d=d(z;x,t)$ and $\widetilde{d}=\widetilde{d}(z;x,t)$. The corresponding jump condition is
\begin{equation}\label{4.2}
\begin{split}
\mathbf{R}^{+}(z;x,t)&=\mathbf{R}^{-}(z;x,t)[\mathbf{I}-\mathbf{e}^{\mathrm{i}\mathbf{\Delta}(z;x,t)}
\mathbf{L}(z)\mathbf{e}^{-\mathrm{i}\mathbf{\Delta}(z;x,t)}], \quad z\in \mathbb{R},
\end{split}
\end{equation}
where
\begin{equation}\label{4.3}
\everymath{\displaystyle}
\begin{split}
\mathbf{L}(z)=\begin{pmatrix}
\left[\frac{\left|\beta_{2}\right|^{2}}{\rho}-\beta_{1}^{*}\widetilde{\beta}_{1}^{*}
    -\frac{\mathrm{i}q_{0}}{z\rho}\beta_{1}^{*}\beta_{2}^{*}\widetilde{\beta}_{2} \right] & \left[\frac{\beta_{2}^{*}}{\rho}+\frac{q_{0}^{2}}{z^{2}\rho^{2}}\beta_{2}^{*}
    \left|\widetilde{\beta}_{2}\right|^{2}-\frac{\mathrm{i}q_{0}}{z\rho}\widetilde{\beta}_{1}^{*} \widetilde{\beta}_{2}^{*} \right] & \left[\frac{\mathrm{i}q_{0}}{z\rho}\beta_{2}^{*}\widetilde{\beta}_{2}+\widetilde{\beta}_{1}^{*} \right] \\
    \frac{\mathrm{i}q_{0}}{z}\beta_{1}^{*}\widetilde{\beta}_{2}-\beta_{2} & -\frac{q_{0}^{2}}{z^{2}\rho}\left|\widetilde{\beta}_{2}\right|^{2} & -\frac{\mathrm{i}q_{0}}{z}\widetilde{\beta}_{2} \\
    \beta_{1}^{*} &
    \frac{\mathrm{i}q_{0}}{z\rho}\widetilde{\beta}_{2}^{*} &
    0
\end{pmatrix},
\end{split}
\end{equation}
where $\rho=\rho(z)$, where $\beta_{j}=\beta_{j}(z)$ and $\widetilde{\beta}_{j}=\beta_{j}(q_{0}^{2}/z)$, for $j=1,2$.
\end{proposition}

To guarantee a unique solution to the aforementioned RH problem, it is imperative to establish an appropriate normalization condition. By considering the asymptotic behavior of $z\rightarrow\infty$ and $z\rightarrow 0$
\begin{subequations}\label{4.4}
\begin{align}
\mathbf{R}^{\pm}(z;x,t)&=\mathbf{R}_{\infty}+O(\frac{1}{z}), \quad z\rightarrow\infty, \quad
z\in \mathbb{D}^{\pm}, \\
\mathbf{R}^{\pm}(z;x,t)&=\frac{1}{z}\mathbf{R}_{0}+O(1), \quad z\rightarrow 0, \quad \,\,
z\in \mathbb{D}^{\pm},
\end{align}
\end{subequations}
where
\begin{equation}\label{4.5}
\everymath{\displaystyle}
\begin{split}
\mathbf{R}_{\infty}+\frac{1}{z}\mathbf{R}_{0}=\mathbf{Y}_{-}(z), \quad
\mathbf{R}_{\infty}=\begin{pmatrix}
    \mathrm{i} & 0 & 0  \\
    \mathbf{0_{2\times1}} & \frac{1}{q_{0}}\mathbf{q}_{-}^{\perp} & \frac{1}{q_{0}}\mathbf{q}_{-}
  \end{pmatrix}, \quad
\mathbf{R}_{0}=\begin{pmatrix}
    0 & 0 & q_{0}  \\
    \mathrm{i}\mathbf{q}_{-} & \mathbf{0_{2\times1}} & \mathbf{0_{2\times1}}
  \end{pmatrix}.
\end{split}
\end{equation}
Due to the scattering matrix breaking the symmetry between $\nu_{+}$ and $\nu_{-}$, the asymptotic behavior of $z\rightarrow\infty$ and $z\rightarrow 0$ is obtained using the potential value at $x\rightarrow -\infty$ (rather than $x\rightarrow \infty$). In addition to the asymptotic behavior outlined in equation~\eqref{4.1}, to fully delineate the RH problem presented in equation~\eqref{4.2}, it is also necessary to specify the residue conditions.

By using the meromorphic matrices $\mathbf{R}^{\pm}(z;x,t)$ in Proposition~\ref{pro:8}, the corresponding residue conditions are
\begin{subequations}\label{4.6}
\begin{align}
\left[ \operatorname{Res}_{z=z_{g}}\mathbf{R}^{+}(z;x,t) \right]
&=C_{g}\left[ \mathbf{r}_{3}^{+}(z_{g}),\mathbf{0},\mathbf{0} \right], \quad
\left[ \operatorname{Res}_{z=z_{g}^{*}}\mathbf{R}^{-}(z;x,t) \right]
=\bar{C}_{g}\left[ \mathbf{0},\mathbf{0}, \mathbf{r}_{1}^{-}(z_{g}^{*}) \right], \\
\left[ \operatorname{Res}_{z=\theta_{g}}\mathbf{R}^{+}(z;x,t) \right]
&=-F_{g}\left[ \mathbf{r}_{2}^{+}(\theta_{g}),\mathbf{0},\mathbf{0} \right], \quad
\left[ \operatorname{Res}_{z=q_{0}^{2}/\theta_{g}} \mathbf{R}^{-}(z;x,t) \right]
=\check{F}_{g}\left[ \mathbf{0},\mathbf{0},\mathbf{r}_{2}^{-}(q_{0}^{2}/\theta_{g}) \right], \\
\left[ \operatorname{Res}_{z=\theta_{g}^{*}}\mathbf{R}^{-}(z;x,t) \right]
&=\bar{F}_{g}\left[ \mathbf{0}, \mathbf{r}_{1}^{-}(\theta_{g}^{*}), \mathbf{0} \right], \quad
\left[ \operatorname{Res}_{z=q_{0}^{2}/\theta_{g}^{*}} \mathbf{R}^{+}(z;x,t) \right]
=-\hat{F}_{g}\left[ \mathbf{0}, \mathbf{r}_{3}^{+}(q_{0}^{2}/\theta_{g}^{*}), \mathbf{0} \right],
\end{align}
\end{subequations}
where $\delta_{j}(z)=\delta_{j}(z;x,t)$ for $j=1,2$ and $\mathbf{r}_{j}^{\pm}(z)=\mathbf{r}_{j}^{\pm}(z;x,t)$ for $j=1,2,3$, with norming constants
\begin{subequations}\label{4.7}
\begin{align}
C_{g}&=C_{g}(x,t)=\frac{c_{g}}{h_{11}^{\prime}(z_{g})}\mathrm{e}^{-2\mathrm{i}\delta_{1}(z_{g})}, \quad   F_{g}=F_{g}(x,t)=\frac{f_{g}}{h_{11}^{\prime}(\theta_{g})}
\mathrm{e}^{\mathrm{i}[\delta_{2}(\theta_{g})-\delta_{1}(\theta_{g})]},   \\
\bar{C}_{g}&=\bar{C}_{g}(x,t)=\frac{\bar{c}_{g}}{h_{33}^{\prime}(z_{g}^{*})}
\mathrm{e}^{2\mathrm{i}\delta_{1}(z_{g}^{*})}, \quad
\bar{F}_{g}=\bar{F}_{g}(x,t)=\frac{\bar{f}_{g}}{s_{11}^{\prime}(\theta_{g}^{*})}
\mathrm{e}^{\mathrm{i}[\delta_{1}(\theta_{g}^{*})-\delta_{2}(\theta_{g}^{*})]},  \\
\check{F}_{g}&=\check{F}_{g}(x,t)=\frac{\check{f}_{g}s_{11}(q_{0}^{2}
/\theta_{g})}{h_{33}^{\prime}(q_{0}^{2}/\theta_{g})}
\mathrm{e}^{\mathrm{i}[\delta_{2}(\theta_{g})-\delta_{1}(\theta_{g})]},  \quad
\hat{F}_{g}=\hat{F}_{g}(x,t)=\frac{\hat{f}_{g}}{s_{33}^{\prime}(q_{0}^{2}/\theta_{g}^{*})}
\mathrm{e}^{\mathrm{i}[\delta_{1}(\theta_{g}^{*})-\delta_{2}(\theta_{g}^{*})]},
\end{align}
\end{subequations}
where $g=1,\ldots,G_{1}$ for equations involving $z_{g}$, and $g=1,\ldots,G_{2}$ for equations involving $\theta_{g}$. Through the Corollaries~\ref{cor:5} and~\ref{cor:6}, it can be seen that the norming constants satisfy the following symmetry relationship
\begin{subequations}\label{4.8}
\begin{align}
C_{g}^{*}(x,t)&=\bar{C}_{g}(x,t)=\mathrm{e}^{-2\mathrm{i}\arg(z_{g})}C_{g}(x,t),  \\
\check{F}_{g}(x,t)&=-\frac{\mathrm{i}q_{0}}{\theta_{g}}F_{g}(x,t), \quad
\bar{F}_{g}(x,t)=-\frac{F_{g}^{*}(x,t)}{\rho(\theta_{g}^{*})}, \quad
\hat{F}_{g}(x,t)=-\frac{\mathrm{i}q_{0}^{3}}{(\theta_{g}^{*})^{3}}
\frac{F_{g}^{*}(x,t)}{\rho(\theta_{g}^{*})}.
\end{align}
\end{subequations}

\subsection{Formal solutions of the Riemann-Hilbert problem and reconstruction formula}

Regularize the RH problem by subtracting the leading asymptotics and any pole contributions from the discrete spectrum, then the solutions of the RH problem can be obtained with the help of the Plemelj formula.

\begin{theorem}\label{thm:5}
Pure soliton solutions of the RH problem is given by the following formula
\begin{equation}\label{4.9}
\begin{split}
\mathbf{R}(z;x,t)&=\mathbf{Y}_{-}(z)-\frac{1}{2\pi\mathrm{i}}
\int_{\mathbb{R}}\frac{\mathbf{R}^{-}(\xi)\widetilde{\mathbf{L}}(\xi)}{\xi-z}\,\mathrm{d}\xi
+\sum_{i=1}^{G_{1}}\left[ \frac{\operatorname{Res}_{z=z_{i}}\mathbf{R}^{+}}{z-z_{i}} +\frac{\operatorname{Res}_{z=z_{i}^{*}}\mathbf{R}^{-}}{z-z_{i}^{*}} \right] \\
&+\sum_{j=1}^{G_{2}}\left[ \frac{\operatorname{Res}_{z=\theta_{j}}\mathbf{R}^{+}}{z-\theta_{j}} +\frac{\operatorname{Res}_{z=\theta_{j}^{*}}\mathbf{R}^{-}}{z-\theta_{j}^{*}} \right]
+\sum_{j=1}^{G_{2}}\left[ \frac{\operatorname{Res}_{z=q_{0}^{2}/\theta_{j}^{*}}\mathbf{R}^{+}}{z-(q_{0}^{2}/\theta_{j}^{*})} +\frac{\operatorname{Res}_{z=q_{0}^{2}/\theta_{j}}\mathbf{R}^{-}}{z-(q_{0}^{2}/\theta_{j})} \right],
\end{split}
\end{equation}
where $\widetilde{\mathbf{L}}(z)=\mathbf{e}^{\mathrm{i}\mathbf{\Delta}(z;x,t)}
\mathbf{L}(z)\mathbf{e}^{-\mathrm{i}\mathbf{\Delta}(z;x,t)}$, with $\mathbf{R}(z;x,t)=\mathbf{R}^{\pm}(z;x,t)=(\mathbf{r}_{1}^{\pm},\mathbf{r}_{2}^{\pm},
\mathbf{r}_{3}^{\pm})$ for $z\in \mathbb{D}^{\pm}$. Furthermore, the eigenfunctions are given by
\begin{equation}\label{4.10}
\begin{split}
\mathbf{r}_{1}^{-}(z)&=\begin{pmatrix}
     \mathrm{i}   \\
     \mathrm{i}\mathbf{q}_{-}/z
  \end{pmatrix}
-\frac{1}{2\pi\mathrm{i}}
\int_{\mathbb{R}}\frac{\left[\mathbf{R}^{-}(\xi)\widetilde{\mathbf{L}}(\xi) \right]_{1}}{\xi-z}\,\mathrm{d}\xi
+\sum_{i=1}^{G_{1}}\left[ \frac{C_{i}\mathbf{r}_{3}^{+}(z_{i})}{z-z_{i}} \right]
-\sum_{j=1}^{G_{2}}\left[ \frac{F_{j}\mathbf{r}_{2}^{+}(\theta_{j})}{z-\theta_{j}} \right],
\end{split}
\end{equation}
where $z=z_{g}^{*},\theta_{g}^{*}$.
\begin{equation}\label{4.11}
\begin{split}
\mathbf{r}_{3}^{+}(z)&=\begin{pmatrix}
     q_{0}/z   \\
     \mathbf{q}_{-}/q_{0}
  \end{pmatrix}
-\frac{1}{2\pi\mathrm{i}} \int_{\mathbb{R}}\frac{\left[
\mathbf{R}^{-}(\xi)\widetilde{\mathbf{L}}(\xi) \right]_{3}}{\xi-z}\,\mathrm{d}\xi
+\sum_{i=1}^{G_{1}}\left[ \frac{\bar{C}_{i}\mathbf{r}_{1}^{-}(z_{i}^{*})}{z-z_{i}^{*}} \right]
+\sum_{j=1}^{G_{2}}\left[
\frac{\check{F}_{j}\mathbf{r}_{2}^{-}(q_{0}^{2}/\theta_{j})}{z-(q_{0}^{2}/\theta_{j})} \right],
\end{split}
\end{equation}
where $z=z_{g},q_{0}^{2}/\theta_{g}^{*}$.
\begin{equation}\label{4.12}
\begin{split}
\mathbf{r}_{2}^{+}(\theta_{g})&=\begin{pmatrix}
     0   \\
     \mathbf{q}_{-}^{\bot}/q_{0}
  \end{pmatrix}
-\frac{1}{2\pi\mathrm{i}} \int_{\mathbb{R}}\frac{\left[
\mathbf{R}^{-}(\xi)\widetilde{\mathbf{L}}(\xi) \right]_{2}}{\xi-\theta_{g}}\,\mathrm{d}\xi
+\sum_{j=1}^{G_{2}}\left[ \frac{\bar{F}_{j}\mathbf{r}_{1}^{-}(\theta_{j}^{*})}{\theta_{g}
-\theta_{j}^{*}} \right] -\sum_{j=1}^{G_{2}}\left[ \frac{\hat{F}_{j}\mathbf{r}_{3}^{+}(q_{0}^{2}
/\theta_{j}^{*})}{\theta_{g}-(q_{0}^{2}/\theta_{j}^{*})} \right],
\end{split}
\end{equation}
\begin{equation}\label{4.13}
\begin{split}
\mathbf{r}_{2}^{-}(\frac{q_{0}^{2}}{\theta_{g}})&=\begin{pmatrix}
     0   \\
     \mathbf{q}_{-}^{\bot}/q_{0}
  \end{pmatrix}
-\frac{1}{2\pi\mathrm{i}} \int_{\mathbb{R}}\frac{\left[\mathbf{R}^{-}(\xi)
\widetilde{\mathbf{L}}(\xi) \right]_{2}}{\xi-(q_{0}^{2}/\theta_{g})}\,\mathrm{d}\xi
+\sum_{j=1}^{G_{2}}\left[ \frac{\bar{F}_{j}\mathbf{r}_{1}^{-}(\theta_{j}^{*})}{(q_{0}^{2}
/\theta_{g})-\theta_{j}^{*}} \right] -\sum_{j=1}^{G_{2}}\left[ \frac{\hat{F}_{j}\mathbf{r}_{3}^{+}(q_{0}^{2}/\theta_{j}^{*})}{(q_{0}^{2}/\theta_{g})
-(q_{0}^{2}/\theta_{j}^{*})} \right].
\end{split}
\end{equation}
\end{theorem}

Through the examination of solutions to the regularized RH problem, one can derive certain conditions by juxtaposing the first column of the matrix $\mathbf{R}^{-}(z;x,t)$ with the asymptotic properties of the modified Jost eigenfunctions as depicted in equation~\eqref{3.8}.
\begin{equation}\label{4.14}
\begin{split}
q_{k}(x,t)&=-\mathrm{i}\lim_{z\rightarrow\infty}[z \nu_{-,(k+1)1}(z;x,t)], \quad k=1,2.
\end{split}
\end{equation}

\begin{theorem}[Reconstruction formula]\label{thm:6}
Pure soliton solutions $\mathbf{q}(x,t)$ of the defocusing-defocusing coupled Hirota equations with NZBC~\eqref{1.3} are reconstructed as
\begin{equation}\label{4.15}
\begin{split}
q_{k}(x,t)&=q_{k-}-\frac{1}{2\pi} \int_{\mathbb{R}}\left[
\mathbf{R}^{-}(\xi)\widetilde{\mathbf{L}}(\xi) \right]_{(k+1)1}\,\mathrm{d}\xi
-\sum_{i=1}^{N_{1}}\mathrm{i}C_{i}\mathbf{r}_{(k+1)3}^{+}(z_{i})
+\sum_{j=1}^{N_{2}}\mathrm{i}F_{j}\mathbf{r}_{(k+1)2}^{+}(\theta_{j}), \quad k=1,2.
\end{split}
\end{equation}
\end{theorem}

\subsection{Trace formulae}

We need to reconstruct the analytical scattering coefficients $h_{11}(z)$ and $s_{11}(z)$ based on the scattering data~\cite{B1}. We can define
\begin{subequations}\label{4.16}
\begin{align}
\chi^{+}(z)&=h_{11}(z)\prod_{g=1}^{G_{1}}\frac{z-z_{g}^{*}}{z-z_{g}}
\prod_{g=1}^{G_{2}}\frac{z-\theta_{g}^{*}}{z-\theta_{g}}, \quad z\in \mathbb{D}^{+}, \\
\chi^{-}(z)&=s_{11}(z)\prod_{g=1}^{G_{1}}\frac{z-z_{g}}{z-z_{g}^{*}}
\prod_{g=1}^{G_{2}}\frac{z-\theta_{g}}{z-\theta_{g}^{*}}, \quad \, z\in \mathbb{D}^{-}.
\end{align}
\end{subequations}
Using the definition of the reflection coefficients~\eqref{2.61} and the corresponding calculation of the scattering coefficients, then we have
\begin{equation}\label{4.17}
\begin{split}
\ln \chi^{+}(z)+\ln \chi^{-}(z)=-\ln \left[1-\left| \beta_{1}(z) \right|^{2}
-\frac{\left| \beta_{2}(z) \right|^{2}}{\rho(z)} \right], \quad z\in \mathbb{R}.
\end{split}
\end{equation}

By combining~\eqref{4.16} with~\eqref{4.17} and using the Plemelj formula, it can be concluded that
\begin{subequations}\label{4.18}
\begin{align}
\chi^{+}(z)&=\exp\left[ -\frac{1}{2\pi\mathrm{i}} \int_{\mathbb{R}} \ln \left[1-\left| \beta_{1}(\xi) \right|^{2}-\frac{\left| \beta_{2}(\xi) \right|^{2}}{\rho(\xi)} \right] \,\frac{\mathrm{d}\xi}{\xi-z} \right],  \quad z\in \mathbb{D}^{+}, \\
\chi^{-}(z)&=\exp\left[ \frac{1}{2\pi\mathrm{i}} \int_{\mathbb{R}} \ln \left[1-\left| \beta_{1}(\xi) \right|^{2}-\frac{\left| \beta_{2}(\xi) \right|^{2}}{\rho(\xi)} \right] \,\frac{\mathrm{d}\xi}{\xi-z} \right],  \quad \,\,\,\, z\in \mathbb{D}^{-}.
\end{align}
\end{subequations}
By substituting expression~\eqref{4.18} into definition~\eqref{4.16}, the scattering coefficient display expression can be solved
\begin{subequations}\label{4.19}
\begin{align}
h_{11}(z)&=\prod_{g=1}^{G_{1}}\frac{z-z_{g}}{z-z_{g}^{*}}
\prod_{g=1}^{G_{2}}\frac{z-\theta_{g}}{z-\theta_{g}^{*}}
\exp\left[ -\frac{1}{2\pi\mathrm{i}} \int_{\mathbb{R}} \ln \left[1-\left| \beta_{1}(\xi) \right|^{2}-\frac{\left| \beta_{2}(\xi) \right|^{2}}{\rho(\xi)} \right] \,\frac{\mathrm{d}\xi}{\xi-z} \right],  \quad z\in \mathbb{D}^{+}, \\
s_{11}(z)&=\prod_{g=1}^{G_{1}}\frac{z-z_{g}^{*}}{z-z_{g}}
\prod_{g=1}^{G_{2}}\frac{z-\theta_{g}^{*}}{z-\theta_{g}}\exp\left[ \frac{1}{2\pi\mathrm{i}} \int_{\mathbb{R}} \ln \left[1-\left| \beta_{1}(\xi) \right|^{2}-\frac{\left| \beta_{2}(\xi) \right|^{2}}{\rho(\xi)} \right] \,\frac{\mathrm{d}\xi}{\xi-z} \right],  \quad \,\,\,\, z\in \mathbb{D}^{-}.
\end{align}
\end{subequations}

By comparing the behavior of trace formula~\eqref{4.19} as $z\rightarrow 0$ and the asymptotic behavior of $h_{11}(z)$ and $s_{11}(z)$ in~\eqref{3.13}, we can calculate the asymptotic phase difference $\Delta\delta=\delta_{+}-\delta_{-}$.
\begin{equation}\label{4.20}
\begin{split}
\Delta\delta=2\sum_{g=1}^{G_{1}}\arg(z_{g})+2\sum_{g=1}^{G_{2}}\arg(\theta_{g})+\frac{1}{2\pi} \int_{\mathbb{R}} \ln \left[1-\left| \beta_{1}(\xi) \right|^{2}-\frac{\left| \beta_{2}(\xi) \right|^{2}}{\rho(\xi)} \right] \,\frac{\mathrm{d}\xi}{\xi}.
\end{split}
\end{equation}

\subsection{Pure soliton solutions}

Since the discrete eigenvalues on $C_{o}$ satisfy the restrictions $\arg(C_{g})=\arg(z_{g})$ for $g=1,\ldots,G_{1}$, the functions $C_{g}(x,t)$ in Theorem~\ref{thm:5} can be parameterized
\begin{equation}\label{4.21}
\begin{split}
C_{g}(x,t)\mathrm{e}^{2\mathrm{i}\delta_{1}(z_{g})}
=2\left|\lambda(z_{g})\right| \mathrm{e}^{2\left|\lambda(z_{g})\right| \kappa_{g}+\mathrm{i}\chi_{g}}, \quad g=1,\ldots,G_{1},
\end{split}
\end{equation}
where $\kappa_{g}$ and $\chi_{g}$ are real parameters and $\chi_{g}=\arg (z_{g})+k \pi$ for $k=0,1$.

\begin{theorem}\label{thm:7}
In the reflectionless case, the pure soliton solutions~\eqref{4.15} of the defocusing-defocusing coupled Hirota equations with NZBC~\eqref{1.3} may be written
\begin{equation}\label{4.22}
\begin{split}
\mathbf{q}(x,t)=\frac{1}{\det\mathbf{K}} \begin{pmatrix}
\det \mathbf{K}_{1}^{\mathrm{aug}} \\
\det \mathbf{K}_{2}^{\mathrm{aug}}
\end{pmatrix},  \quad
\mathbf{K}_{n}^{\mathrm{aug}}=\left(\begin{array}{ll}
q_{n-} & \mathbf{E} \\
\mathbf{A}_{n} & \mathbf{K}
\end{array}\right), \quad n=1,2,
\end{split}
\end{equation}
the vectors $\mathbf{A}_{n}$, $\mathbf{E}$ and matrix $\mathbf{K}$ are
\begin{equation}\label{4.23}
\begin{split}
\mathbf{A}_{n}&=\left(A_{n1}, \ldots, A_{n(G_{1}+G_{2})}\right)^{T}, \quad
\mathbf{E}=\left(E_{1}, \ldots, E_{G_{1}+G_{2}}\right), \quad \mathbf{K}=\mathbf{I}+\mathbf{P},
\end{split}
\end{equation}
where
\begin{align}\label{4.24}
\begin{split}
E_{g}=\left\{\begin{array}{ll}
\mathrm{i}C_{g}(x,t),  &g=1, \ldots , G_{1}, \\
-\mathrm{i}F_{g-G_{1}}(x,t), & g=G_{1}+1, \ldots , G_{1}+G_{2},
\end{array}\right.
\end{split}
\end{align}
the entries of matrix $\mathbf{P}=(P_{jk}(x,t))$ are defined as
\begin{align}\label{4.25}
\everymath{\displaystyle}
\begin{array}{l}
P_{jk}(x,t)=\left\{\begin{array}{ll}
-\frac{\mathrm{i}z_{k}}{q_{0}}f_{k}^{(2)}(z_{j};x,t),  &j,k=1, \ldots, G_{1}, \\
-f_{k-G_{1}}^{(5)}(z_{j};x,t),  & j=1, \ldots, G_{1}, \quad  k=G_{1}+1, \ldots, G_{1}+G_{2}, \\
-\sum_{a=1}^{G_{2}}f_{aj}(x,t)f_{k}^{(1)}(\theta_{a}^{*};x,t), & j=G_{1}+1, \ldots, G_{1}+G_{2}, \quad k=1, \ldots, G_{1}, \\
\sum_{a=1}^{G_{2}}f_{aj}(x,t)f_{k-G_{1}}^{(3)}(\theta_{a}^{*};x,t), & j, k=G_{1}+1, \ldots, G_{1}+G_{2},
\end{array}\right.
\end{array}
\end{align}
where
\begin{equation}\label{4.26}
\begin{split}
f_{jk}(x,t)=f_{j}^{(4)}(\theta_{k-G_{1}};x,t)+\frac{\mathrm{i}\theta_{j}^{*}}{q_{0}} f_{j}^{(6)}(\theta_{k-G_{1}};x,t),
\end{split}
\end{equation}
where
\begin{equation}\label{4.27}
\everymath{\displaystyle}
\begin{split}
A_{ni^{\prime}}=\left\{\begin{array}{ll}
\frac{q_{n-}}{q_{0}},  & i^{\prime}=1, \ldots, G_{1}, \\
(-1)^{n+1} \frac{q_{\bar{n}-}^{*}}{q_{0}}+\sum_{j=1}^{G_{2}}\frac{\mathrm{i}q_{n-}}{\theta_{j}^{*}} f_{ji^{\prime}}, & i^{\prime}=G_{1}+1, \ldots, G_{1}+G_{2},
\end{array}\right.
\end{split}
\end{equation}
and $\bar{n}=n+(-1)^{n+1}$.
\end{theorem}

\begin{proof}
For $i=1,\ldots,G_{1}$ and $j=1,\ldots,G_{2}$, define
\begin{subequations}\label{4.28}
\begin{align}
f_{i}^{(1)}(z;x,t)&=\frac{C_{i}(x,t)}{z-z_{i}}, \quad
f_{j}^{(3)}(z;x,t)=\frac{F_{j}(x,t)}{z-\theta_{j}}, \quad
f_{j}^{(5)}(z;x,t)=\frac{\check{F}_{j}(x,t)}{z-(q_{0}^{2}/\theta_{j})}, \\
f_{i}^{(2)}(z;x,t)&=\frac{\bar{C}_{i}(x,t)}{z-z_{i}^{*}}, \quad f_{j}^{(4)}(z;x,t)=\frac{\bar{F}_{j}(x,t)}{z-\theta_{j}^{*}}, \quad f_{j}^{(6)}(z;x,t)=\frac{\hat{F}_{j}(x,t)}{z-(q_{0}^{2}/\theta_{j}^{*})}.
\end{align}
\end{subequations}
Under the condition of the reflectionless, there are equations that hold true
\begin{subequations}\label{4.29}
\begin{align}
r_{21}^{-}(z_{i^{\prime}}^{*})&=\frac{\mathrm{i}q_{1-}}{z_{i^{\prime}}^{*}}
+\sum_{i=1}^{G_{1}}f_{i}^{(1)}(z_{i^{\prime}}^{*})r_{23}^{+}(z_{i})
-\sum_{j=1}^{G_{2}}f_{j}^{(3)}(z_{i^{\prime}}^{*})r_{22}^{+}(\theta_{j}), \\
r_{21}^{-}(\theta_{j^{\prime}}^{*})&=\frac{\mathrm{i}q_{1-}}{\theta_{j^{\prime}}^{*}}
+\sum_{i=1}^{G_{1}}f_{i}^{(1)}(\theta_{j^{\prime}}^{*})r_{23}^{+}(z_{i})
-\sum_{j=1}^{G_{2}}f_{j}^{(3)}(\theta_{j^{\prime}}^{*})r_{22}^{+}(\theta_{j}), \\
r_{23}^{+}(z_{i^{\prime}})&=\frac{q_{1-}}{q_{0}}
+\sum_{i=1}^{G_{1}}f_{i}^{(2)}(z_{i^{\prime}})r_{21}^{-}(z_{i}^{*})
+\sum_{j=1}^{G_{2}}f_{j}^{(5)}(z_{i^{\prime}})r_{22}^{-}(\frac{q_{0}^{2}}{\theta_{j}}), \\
r_{23}^{+}(\frac{q_{0}^{2}}{\theta_{j^{\prime}}^{*}})&=\frac{q_{1-}}{q_{0}}
+\sum_{i=1}^{G_{1}}f_{i}^{(2)}(\frac{q_{0}^{2}}{\theta_{j^{\prime}}^{*}})r_{21}^{-}(z_{i}^{*})
+\sum_{j=1}^{G_{2}}f_{j}^{(5)}(\frac{q_{0}^{2}}{\theta_{j^{\prime}}^{*}})
r_{22}^{-}(\frac{q_{0}^{2}}{\theta_{j}}), \\
r_{22}^{-}(\frac{q_{0}^{2}}{\theta_{j^{\prime}}})&=\frac{q_{2-}^{*}}{q_{0}}
+\sum_{j=1}^{G_{2}}f_{j}^{(4)}(\frac{q_{0}^{2}}{\theta_{j^{\prime}}})r_{21}^{-}(\theta_{j}^{*})
-\sum_{j=1}^{G_{2}}f_{j}^{(6)}(\frac{q_{0}^{2}}{\theta_{j^{\prime}}})
r_{23}^{+}(\frac{q_{0}^{2}}{\theta_{j}^{*}}),  \\
r_{22}^{+}(\theta_{j^{\prime}})&=\frac{q_{2-}^{*}}{q_{0}}
+\sum_{j=1}^{G_{2}}f_{j}^{(4)}(\theta_{j^{\prime}})r_{21}^{-}(\theta_{j}^{*})
-\sum_{j=1}^{G_{2}}f_{j}^{(6)}(\theta_{j^{\prime}})r_{23}^{+}(\frac{q_{0}^{2}}{\theta_{j}^{*}}),
\end{align}
\end{subequations}
where $i^{\prime}=1,\ldots,G_{1}$ and $j^{\prime}=1,\ldots,G_{2}$. With the help of the analytical properties~\eqref{2.56} and~\eqref{2.60}, the elements in piecewise meromorphic functions have the following properties:
\begin{equation}\label{4.30}
\begin{split}
r_{21}^{-}(z_{i}^{*})=\frac{\mathrm{i}z_{i}}{q_{0}}r_{23}^{+}(z_{i}), \quad
r_{22}^{-}(\frac{q_{0}^{2}}{\theta_{j}};x,t)=r_{22}^{+}(\theta_{j}), \quad
r_{23}^{+}(\frac{q_{0}^{2}}{\theta_{j}^{*}})
=\frac{\theta_{j}^{*}}{\mathrm{i}q_{0}}r_{21}^{-}(\theta_{j}^{*}).
\end{split}
\end{equation}
Substituting~\eqref{4.30} into equations~\eqref{4.29} yields
\begin{subequations}\label{4.31}
\begin{align}
r_{23}^{+}(z_{i^{\prime}})&=\frac{q_{1-}}{q_{0}}
+\sum_{i=1}^{G_{1}}\frac{\mathrm{i}z_{i}}{q_{0}} f_{i}^{(2)}(z_{i^{\prime}})r_{23}^{+}(z_{i})
+\sum_{j=1}^{G_{2}}f_{j}^{(5)}(z_{i^{\prime}})r_{22}^{+}(\theta_{j}), \\
r_{22}^{+}(\theta_{j^{\prime}})&=\frac{q_{2-}^{*}}{q_{0}}
+\sum_{j=1}^{G_{2}} \left[ f_{j}^{(4)}(\theta_{j^{\prime}})
+\frac{\mathrm{i}\theta_{j}^{*}}{q_{0}}f_{j}^{(6)}(\theta_{j^{\prime}}) \right] r_{21}^{-}(\theta_{j}^{*}),
\end{align}
\end{subequations}
and
\begin{equation}\label{4.32}
\begin{split}
r_{22}^{+}(\theta_{j^{\prime}})&=\sum_{j=1}^{G_{2}} \frac{\mathrm{i}q_{1-}}{\theta_{j}^{*}} \left[ f_{j}^{(4)}(\theta_{j^{\prime}})
+\frac{\mathrm{i}\theta_{j}^{*}}{q_{0}}f_{j}^{(6)}(\theta_{j^{\prime}}) \right]
+\sum_{j=1}^{G_{2}}\sum_{i=1}^{G_{1}} \left[ f_{j}^{(4)}(\theta_{j^{\prime}})
+\frac{\mathrm{i}\theta_{j}^{*}}{q_{0}}f_{j}^{(6)}(\theta_{j^{\prime}}) \right]
f_{i}^{(1)}(\theta_{j}^{*})r_{23}^{+}(z_{i}) \\
&+\frac{q_{2-}^{*}}{q_{0}}-\sum_{j=1}^{G_{2}}\sum_{j^{\prime\prime}=1}^{G_{2}} \left[ f_{j}^{(4)}(\theta_{j^{\prime}})
+\frac{\mathrm{i}\theta_{j}^{*}}{q_{0}}f_{j}^{(6)}(\theta_{j^{\prime}}) \right] f_{j^{\prime\prime}}^{(3)}(\theta_{j}^{*})r_{22}^{+}(\theta_{j^{\prime\prime}}).
\end{split}
\end{equation}
The equations pertaining to $r_{23}^{+}(z_{i}^{\prime})$ and $r_{22}^{+}(\theta_{j^{\prime}})$ constitute a self-contained set comprising $G_{1}+G_{2}$ equations, each with $G_{1}+G_{2}$ unknowns. In the same way, a closed system containing $r_{33}^{+}(z_{i}^{\prime})$ and $r_{32}^{+}(\theta_{j^{\prime}})$ can be found. These two systems can be written as $\mathbf{K}\mathbf{X}_{n}=\mathbf{A}_{n}$ for $n=1,2$, while $\mathbf{X}_{n}=\left(X_{n 1}, \ldots, X_{n\left(G_{1}+G_{2}\right)}\right)^{T}$ and
\begin{equation}\label{4.33}
\begin{split}
X_{n i^{\prime}}=\left\{\begin{array}{ll}
r_{(n+1)3}^{+}(z_{i^{\prime}};x,t), & i^{\prime}=1, \ldots, G_{1}, \\
r_{(n+1)2}^{+}(\theta_{i^{\prime}-G_{1}};x,t), & i^{\prime}=G_{1}+1, \ldots, G_{1}+G_{2}.
\end{array}\right.
\end{split}
\end{equation}
Using Cramer's rule, we have
\begin{equation}\label{4.34}
\begin{split}
X_{n i}=\frac{\det\mathbf{K}_{ni}^{\mathrm{aug}}}{\det\mathbf{K}}, \quad i=1, \ldots, G_{1}+G_{2}, \quad n=1,2,
\end{split}
\end{equation}
where $\mathbf{K}_{ni}^{\text{aug}}=\left(\mathbf{K}_{1}, \ldots, \mathbf{K}_{i-1}, \mathbf{A}_{n}, \mathbf{K}_{i+1}, \ldots, \mathbf{K}_{G_{1}+G_{2}}\right)$. By inserting the determinant representation of the solutions from equation~\eqref{4.34} into equation~\eqref{4.15}, one obtains~\eqref{4.22}.
\end{proof}

\subsection{Varieties of soliton solutions}

Here, the different possibilities of soliton solutions~\eqref{4.22} for the defocusing-defocusing coupled Hirota equations with NZBC~\eqref{1.3} are analyzed, and different schemes for soliton solutions~\eqref{4.22} are studied when there is only one or two discrete eigenvalues on or outside the circle $C_{o}$.

\subsubsection{Soliton solutions for the scenario where $G_{1}+G_{2}=1$}

Discuss the case where there is only one discrete eigenvalue on or outside the circle $C_{o}$, i.e. $G_{1}+G_{2}=1$. Firstly, we focus on a set of eigenvalues situated on the circle ($G_{1}=1$ and $G_{2}=0$) and express the discrete eigenvalues and normalization constants as
\begin{equation}\label{4.35}
\begin{split}
z_{1}=q_{0}\mathrm{e}^{\mathrm{i}\alpha_{1}}, \quad c_{1}=\mathrm{e}^{\kappa_{1}+\mathrm{i}[\alpha_{1}+(k-\frac{1}{2})\pi]}, \quad 0<\alpha_{1}<\pi, \quad k=0,1,
\end{split}
\end{equation}
from~\eqref{4.22} one obtains the one-soliton solutions of the defocusing-defocusing coupled Hirota equations with NZBC~\eqref{1.3}:
\begin{equation}\label{4.36}
\begin{split}
\mathbf{q}(x,t)=\mathrm{e}^{\mathrm{i}\alpha_{1}} \left[ \cos(\alpha_{1}) -\mathrm{i}\sin(\alpha_{1}) \left[ \tanh(Q_{1}) \right]^{(-1)^{k+1}} \right] \mathbf{q}_{-},
\end{split}
\end{equation}
where
\begin{equation}\label{4.37}
\begin{split}
Q_{1}=-q_{0}\sin(\alpha_{1})\left[ x+[2q_{0}\cos(\alpha_{1})+2q_{0}^{2}\sigma \cos(2\alpha_{1})+4q_{0}^{2}\sigma]t \right]-\frac{\kappa_{1}}{2}.
\end{split}
\end{equation}

In this case, two types of solitons are obtained. For $k=1$ and $\kappa_{1}\in \mathbb{R}$, one dark-dark soliton solution is given by (a1) and (a2) in Fig.~\ref{fig:1}. Moreover, setting $k=0$ and $\kappa_{1}\in \mathbb{C}$ generates one bright-bright soliton solution in (b1) and (b2) of Fig.~\ref{fig:1}.

\begin{figure}[htb]
\centering
\begin{tabular}{cccc}
\includegraphics[width=0.22\textwidth]{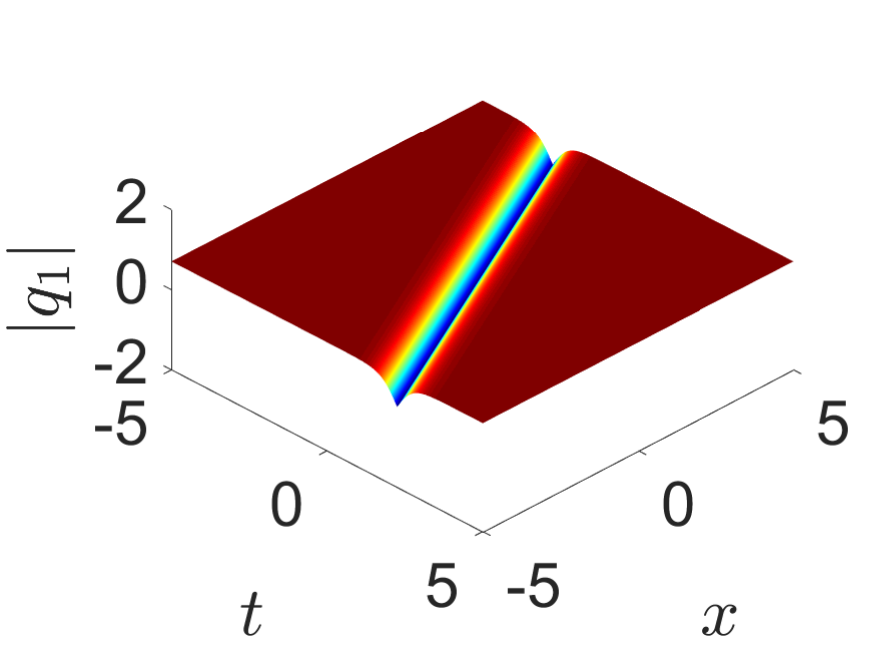} &
\includegraphics[width=0.22\textwidth]{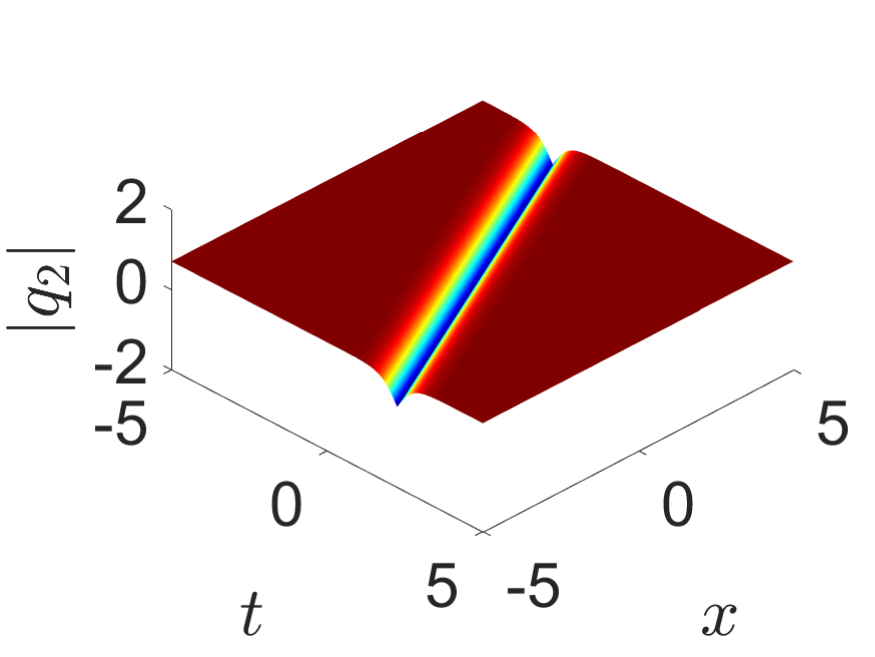} &
\includegraphics[width=0.22\textwidth]{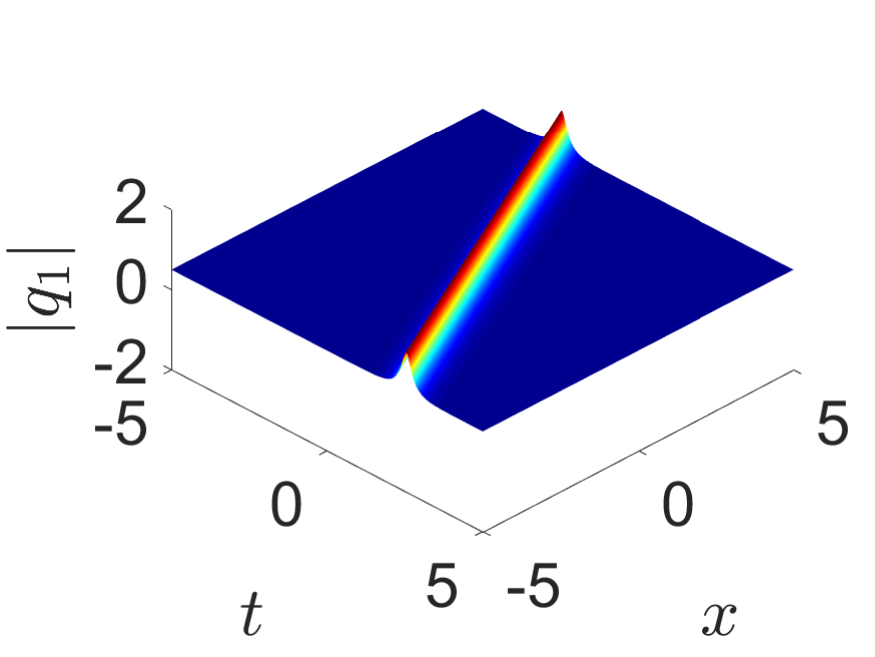} &
\includegraphics[width=0.22\textwidth]{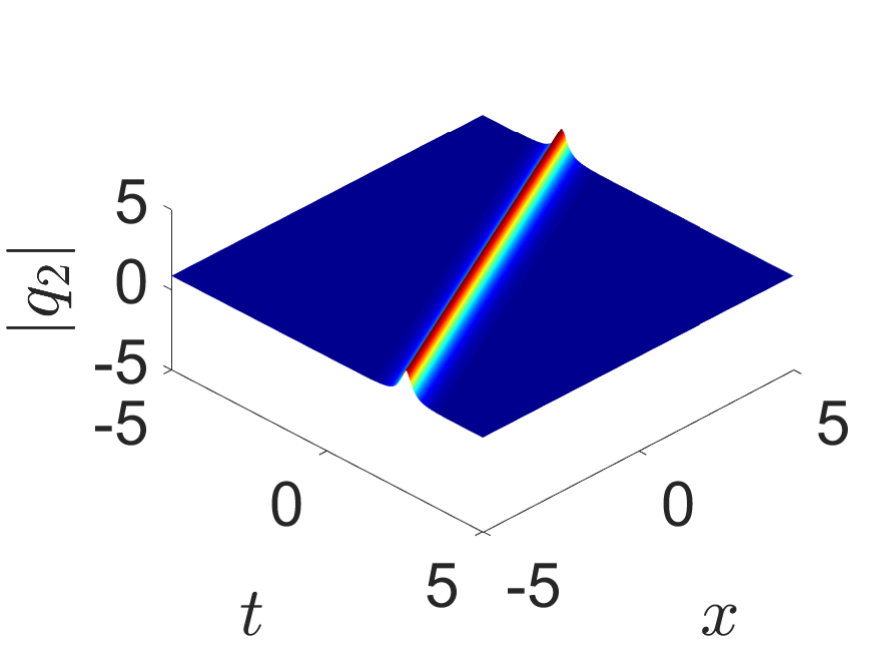} \\
		(a1) & (a2) & (b1) & (b2)
\end{tabular}
\caption{\small (a1) and (a2): One dark-dark soliton solution by taking $\mathbf{q}_{-}=(\frac{\sqrt{2}}{2},-\frac{\sqrt{2}}{2})^{T}$, $\sigma=\kappa_{1}=1$, $\alpha_{1}=\frac{1}{2}\pi$, $k=1$. (b1) and (b2): One bright-bright soliton solution by taking $\mathbf{q}_{-}=(\frac{1}{2},-\frac{\sqrt{3}}{2})^{T}$, $\sigma=1$, $\kappa_{1}=\exp(-\frac{2}{5}-\frac{9}{5}\mathrm{i})$, $\alpha_{1}=\frac{1}{2}\pi$, $k=0$.
}\label{fig:1}
\end{figure}

After that examining a couple of eigenvalues positioned away from the circle ($G_{1}=0$ and $G_{2}=1$) and setting another parameters as
\begin{equation}\label{4.38}
\begin{split}
\theta_{1}=K_{2}\mathrm{e}^{\mathrm{i}\alpha_{2}}, \quad f_{1}=\mathrm{e}^{\kappa_{2}+\mathrm{i}\chi_{2}}, \quad 0<K_{2}<q_{0}, \quad 0<\alpha_{2}<\pi, \quad \chi_{2}\in \mathbb{R},
\end{split}
\end{equation}
from~\eqref{4.22} one generates the corresponding one-soliton solutions
\begin{equation}\label{4.39}
\begin{split}
\mathbf{q}(x,t)=\frac{K_{2}^{2}(\mathrm{e}^{Q_{21}}-1)
+q_{0}^{2}}{K_{2}^{2}(\mathrm{e}^{Q_{22}}-1)+q_{0}^{2}} \mathbf{q}_{-}
+\frac{\mathrm{i}\mathrm{e}^{Q_{23}}K_{2}(K_{2}^{2}-q_{0}^{2})(1
-\mathrm{e}^{2\mathrm{i}\alpha_{2}})}{q_{0}^{3}+q_{0}K_{2}^{2}(\mathrm{e}^{Q_{22}}-1)} \mathbf{q}_{-}^{\perp},
\end{split}
\end{equation}
where
\begin{subequations}\label{4.40}
\begin{align}
Q_{21}&=2\kappa_{2}+2\mathrm{i}\alpha_{2}-\mathrm{i}K_{2}(x+3q_{0}^{2}\sigma t)(\mathrm{e}^{\mathrm{i}\alpha_{2}}-\mathrm{e}^{-\mathrm{i}\alpha_{2}})
-\mathrm{i}K_{2}^{2}t(\mathrm{e}^{2\mathrm{i}\alpha_{2}}-\mathrm{e}^{-2\mathrm{i}\alpha_{2}})
-\mathrm{i}K_{2}^{3}\sigma t(\mathrm{e}^{3\mathrm{i}\alpha_{2}}-\mathrm{e}^{-3\mathrm{i}\alpha_{2}}), \\
Q_{23}&=\kappa_{2}+\mathrm{i}(\chi_{2}-\alpha_{2})
-\mathrm{i}K_{2}\mathrm{e}^{\mathrm{i}\alpha_{2}}(x+3q_{0}^{2}\sigma t)
-\mathrm{i}K_{2}^{2}\mathrm{e}^{2\mathrm{i}\alpha_{2}}t
-\mathrm{i}K_{2}^{3}\sigma \mathrm{e}^{3\mathrm{i}\alpha_{2}}t, \\
Q_{22}&=2\kappa_{2}
-\mathrm{i}K_{2}(\mathrm{e}^{-\mathrm{i}\alpha_{2}}-\mathrm{e}^{-3\mathrm{i}\alpha_{2}})
[K_{2}t(\mathrm{e}^{\mathrm{i}\alpha_{2}}+\mathrm{e}^{3\mathrm{i}\alpha_{2}})
+K_{2}^{2}\sigma t(1+\mathrm{e}^{4\mathrm{i}\alpha_{2}})
+\mathrm{e}^{2\mathrm{i}\alpha_{2}}(x+(K_{2}^{2}+3q_{0}^{2})\sigma t)].
\end{align}
\end{subequations}

For $q_{1-}\times q_{2-}=0$ and $\kappa_{2}\in \mathbb{R}$, one dark-bright soliton solution is given by (a1) and (a2) in Fig.~\ref{fig:2}. Moreover, setting $q_{1-}\times q_{2-}=0$ and $\kappa_{2}\in \mathbb{C}$ generates one bright-bright soliton solution in (b1) and (b2) of Fig.~\ref{fig:2}. In addition, one breather-breather soliton solution is obtained by selecting parameters $q_{1-}\times q_{2-}\neq0$ in (c1) and (c2) of Fig.~\ref{fig:2}.

\begin{figure}[htb]
\centering
\begin{tabular}{cccc}
\includegraphics[width=0.22\textwidth]{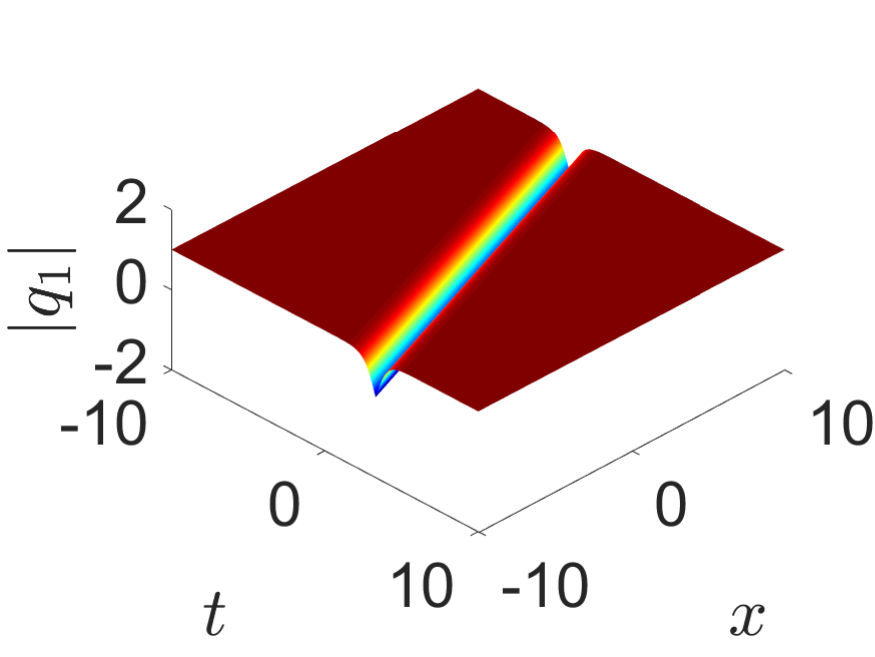} &
\includegraphics[width=0.22\textwidth]{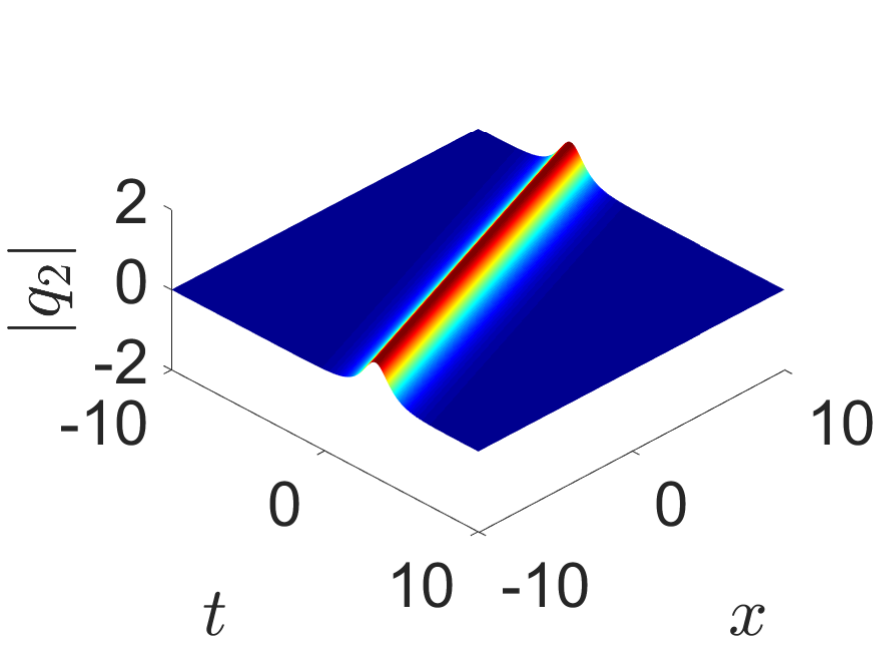} &
\includegraphics[width=0.22\textwidth]{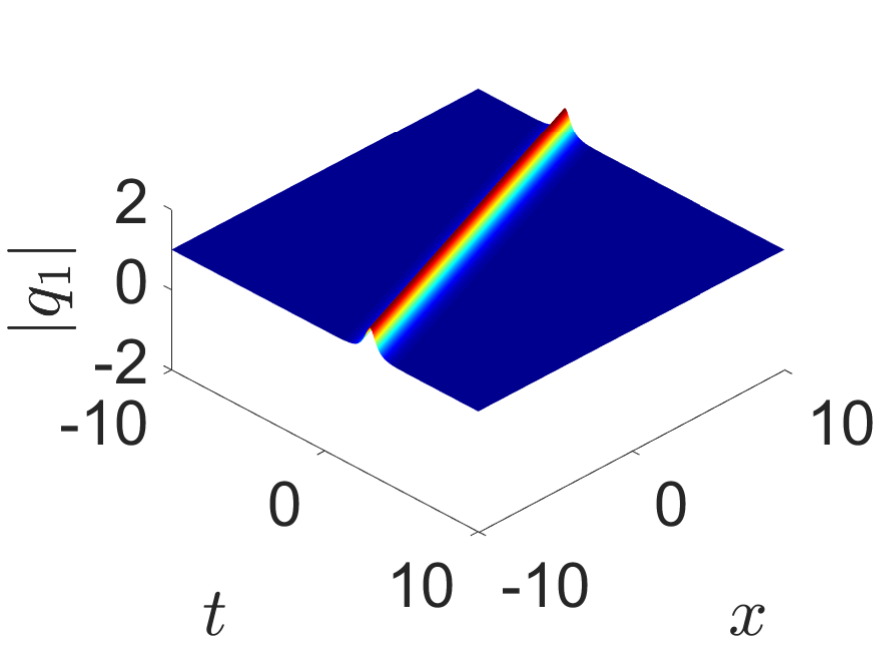} &
\includegraphics[width=0.22\textwidth]{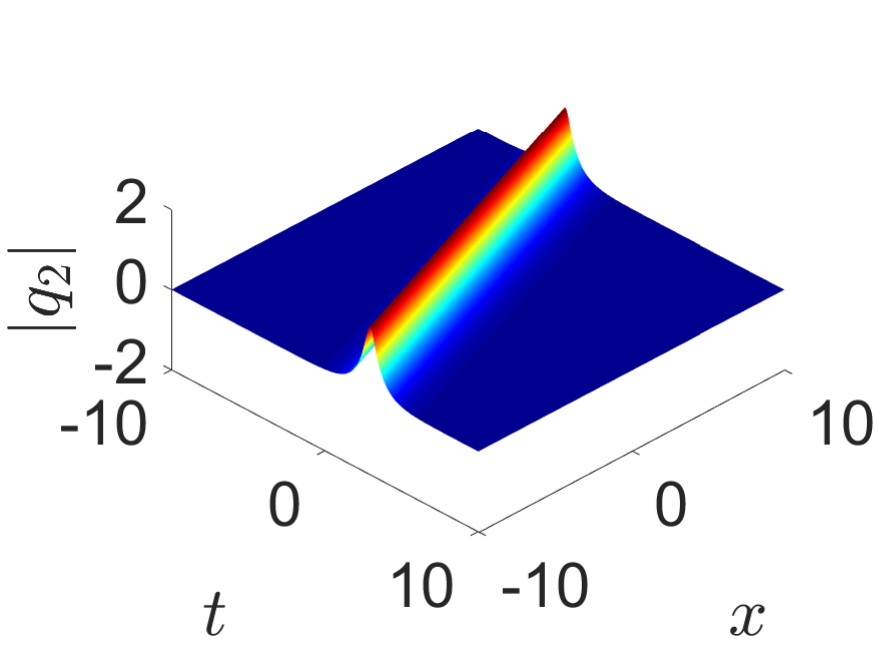} \\
		(a1) & (a2) & (b1) & (b2)
\end{tabular}
\begin{tabular}{cccc}
\includegraphics[width=0.22\textwidth]{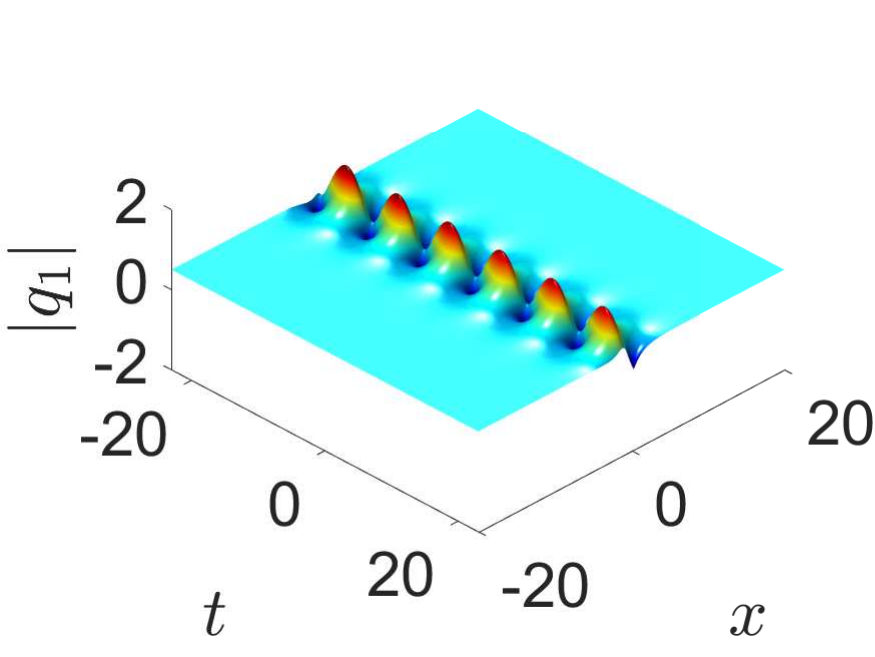} &
\includegraphics[width=0.22\textwidth]{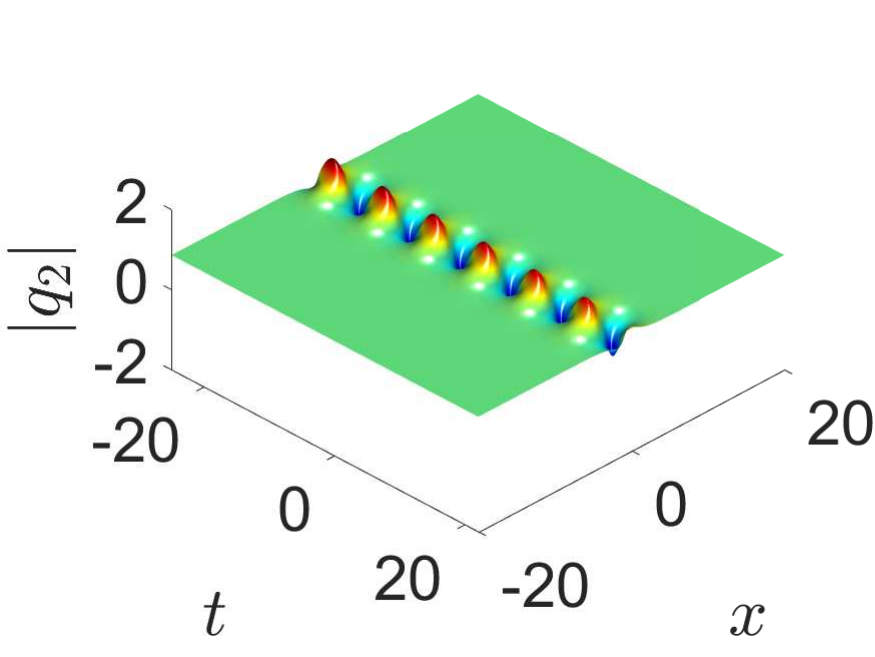} &
\includegraphics[width=0.22\textwidth]{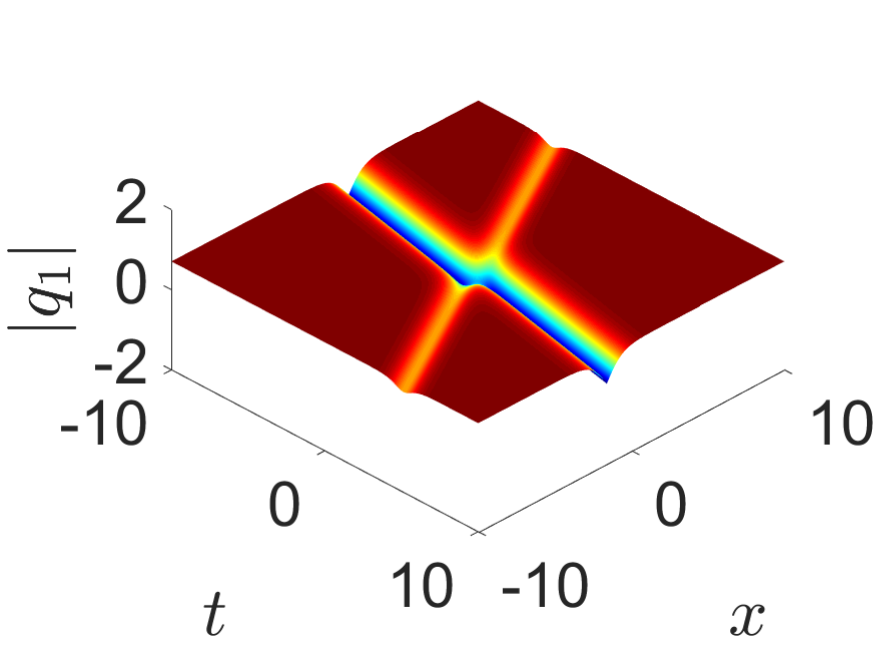} &
\includegraphics[width=0.22\textwidth]{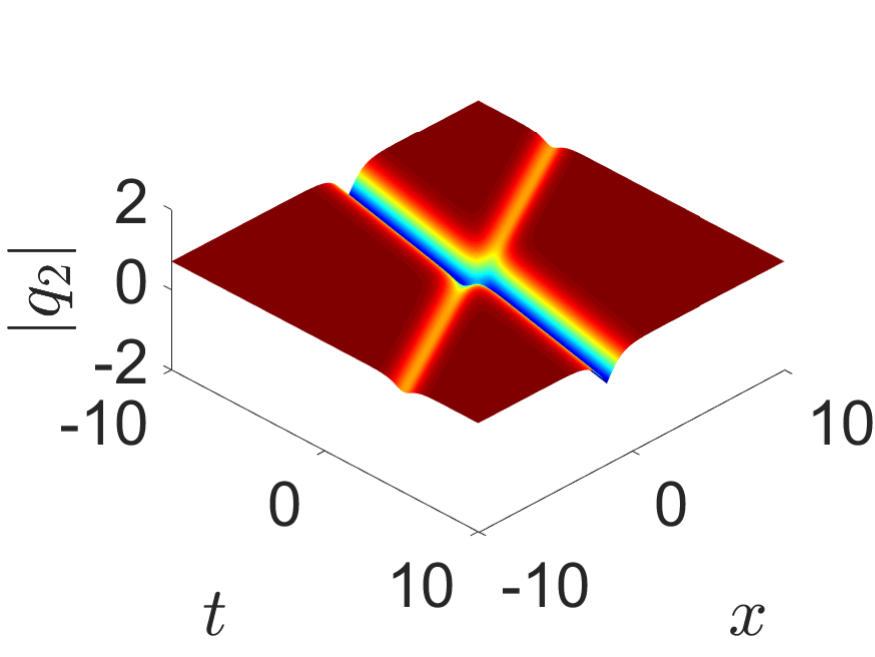} \\
		(c1) & (c2) & (d1) & (d2)
\end{tabular}
\caption{\small (a1) and (a2): One dark-bright soliton solution by taking $\mathbf{q}_{-}=(1,0)^{T}$, $\sigma=\kappa_{2}=1$, $\chi_{2}=-1$, $K_{2}=0.5$, $\alpha_{2}=\frac{1}{2}\pi$. (b1) and (b2): One bright-bright soliton solution by taking $\mathbf{q}_{-}=(1,0)^{T}$, $\sigma=\chi_{2}=1$, $K_{2}=0.5$, $\kappa_{2}=\exp(\frac{3}{5}-\frac{3}{5}\mathrm{i})$, $\alpha_{2}=\frac{1}{2}\pi$. (c1) and (c2): One breather-breather soliton solution by taking $\mathbf{q}_{-}=(\frac{1}{2}\mathrm{e}^{-\frac{1}{10}\mathrm{i}\pi},
-\frac{\sqrt{3}}{2}\mathrm{e}^{-\frac{1}{10}\mathrm{i}\pi})^{T}$, $\sigma=10^{-3}$, $\kappa_{2}=0.5$, $\chi_{2}=1$, $K_{2}=0.9$, $\alpha_{2}=\frac{1}{2}\pi$. (d1) and (d2): Two dark-dark solitons solution by taking $\mathbf{q}_{-}=(\frac{\sqrt{2}}{2},
-\frac{\sqrt{2}}{2})^{T}$, $\sigma=10^{-1}$, $\kappa_{3}=\kappa_{4}=1$, $\alpha_{3}=\frac{1}{2}\pi$, $\alpha_{4}=\frac{1}{4}\pi$, $k_{1}=k_{2}=1$.
}\label{fig:2}
\end{figure}

\subsubsection{Soliton solutions for the scenario where $G_{1}+G_{2}=2$}

Discuss the case where there is only one discrete eigenvalue on or outside the circle $C_{o}$, i.e. $G_{1}+G_{2}=2$. Firstly, taking into account a pair of eigenvalues located on the circumference ($G_{1}=2$ and $G_{2}=0$) and considering the discrete eigenvalues and normalization constants as
\begin{subequations}\label{4.41}
\begin{align}
z_{1}&=q_{0}\mathrm{e}^{\mathrm{i}\alpha_{3}}, \quad
c_{1}=\mathrm{e}^{\kappa_{3}+\mathrm{i}[\alpha_{3}+(k_{1}-\frac{1}{2})\pi]}, \quad
0<\alpha_{3}<\pi, \quad k_{1}=0,1, \\
z_{2}&=q_{0}\mathrm{e}^{\mathrm{i}\alpha_{4}}, \quad
c_{2}=\mathrm{e}^{\kappa_{4}+\mathrm{i}[\alpha_{4}+(k_{2}-\frac{1}{2})\pi]}, \quad
0<\alpha_{4}<\pi, \quad k_{2}=0,1.
\end{align}
\end{subequations}
Through the expressions~\eqref{4.41}, soliton solutions and the reflectionless potentials, it can be inferred that the different structures of the two-soliton solutions are obtained. For $k_{1}=k_{2}=1$ and $\kappa_{3},\kappa_{4}\in \mathbb{R}$, two dark-dark solitons solution is given by (d1) and (d2) in Fig.~\ref{fig:2}. Moreover, setting $k_{1}=k_{2}=0$ and $\kappa_{3},\kappa_{4}\in \mathbb{C}$ generates two bright-bright solitons solution in (a1) and (a2) of Fig.~\ref{fig:3}. In addition, one dark-dark and one bright-bright solitons solution is obtained by selecting parameters $k_{1}=1$, $k_{2}=0$, $\kappa_{3}\in \mathbb{R}$ and $\kappa_{4}\in \mathbb{C}$ in (b1) and (b2) of Fig.~\ref{fig:3}. Finally, setting $k_{1}=0$, $k_{2}=1$, $\kappa_{3}\in \mathbb{C}$ and $\kappa_{4}\in \mathbb{R}$ yields one bright-bright and one dark-dark solitons solution in (c1) and (c2) of Fig.~\ref{fig:3}.

\begin{figure}[htb]
\centering
\begin{tabular}{cccc}
\includegraphics[width=0.22\textwidth]{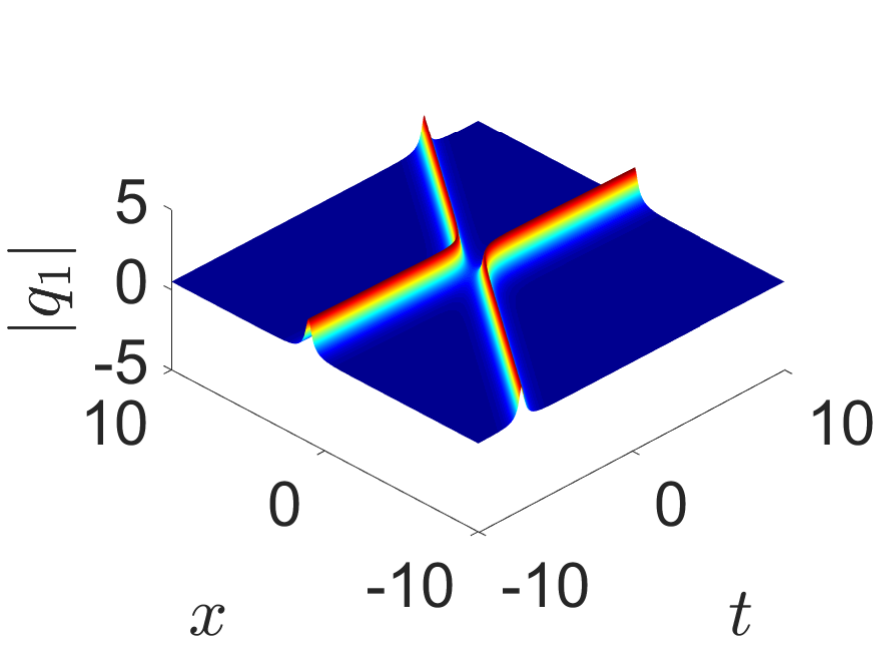} &
\includegraphics[width=0.22\textwidth]{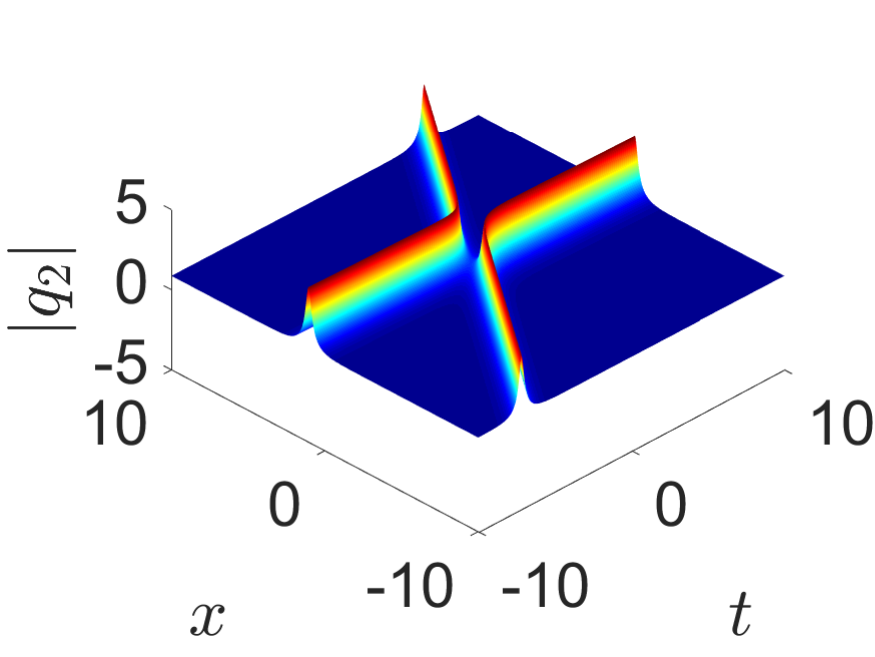} &
\includegraphics[width=0.22\textwidth]{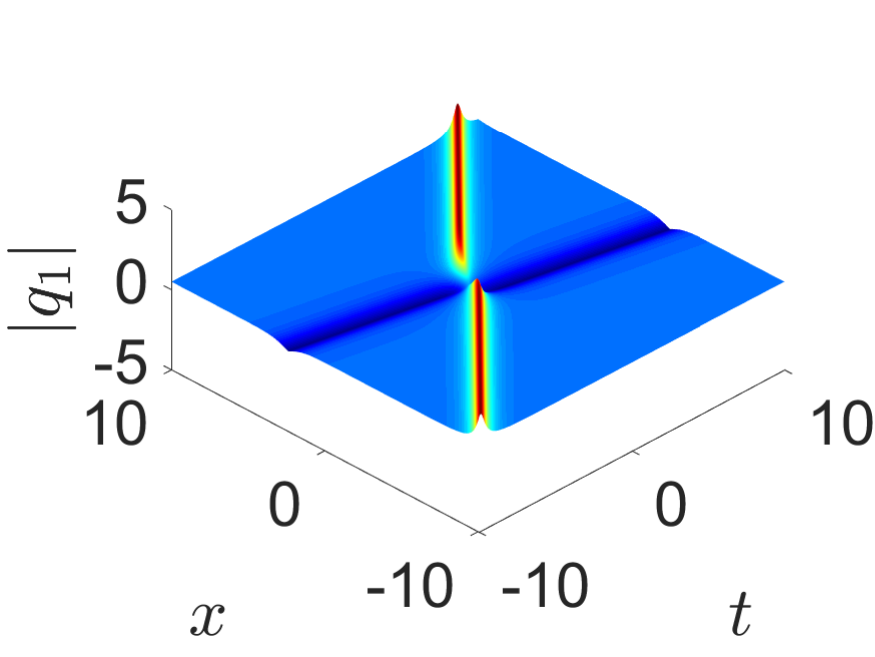} &
\includegraphics[width=0.22\textwidth]{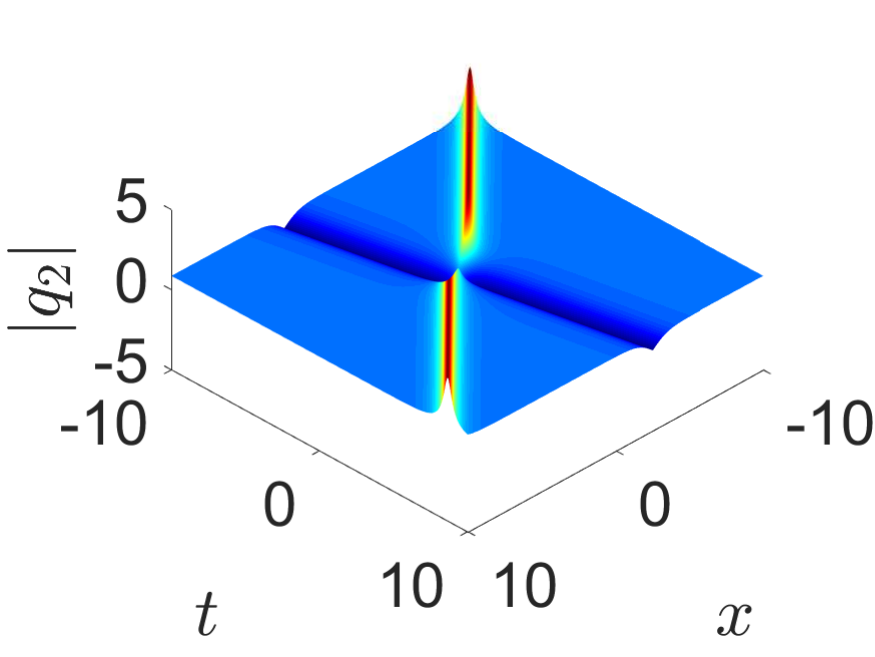} \\
		(a1) & (a2) & (b1) & (b2)
	\end{tabular}
\begin{tabular}{cccc}
\includegraphics[width=0.22\textwidth]{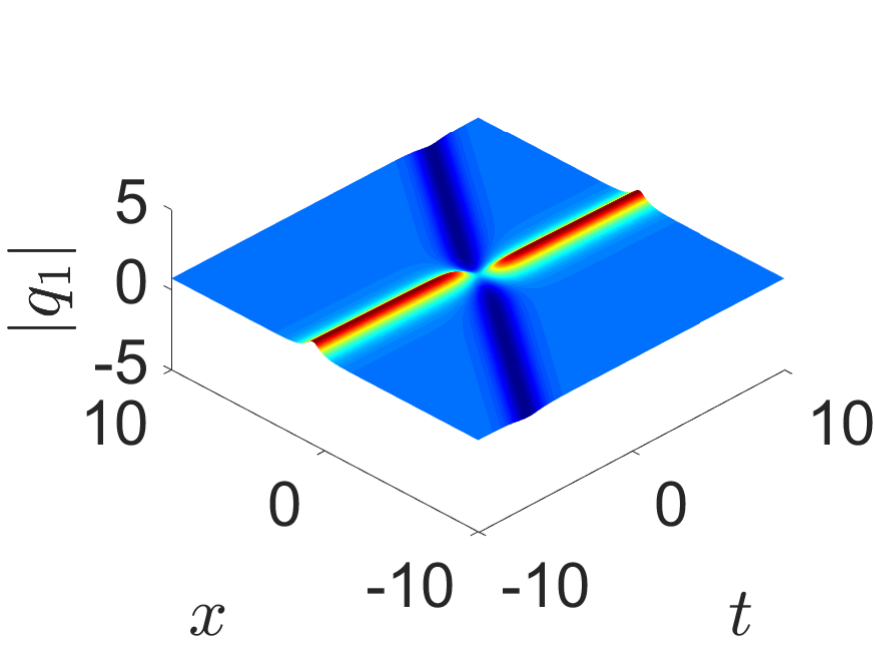} &
\includegraphics[width=0.22\textwidth]{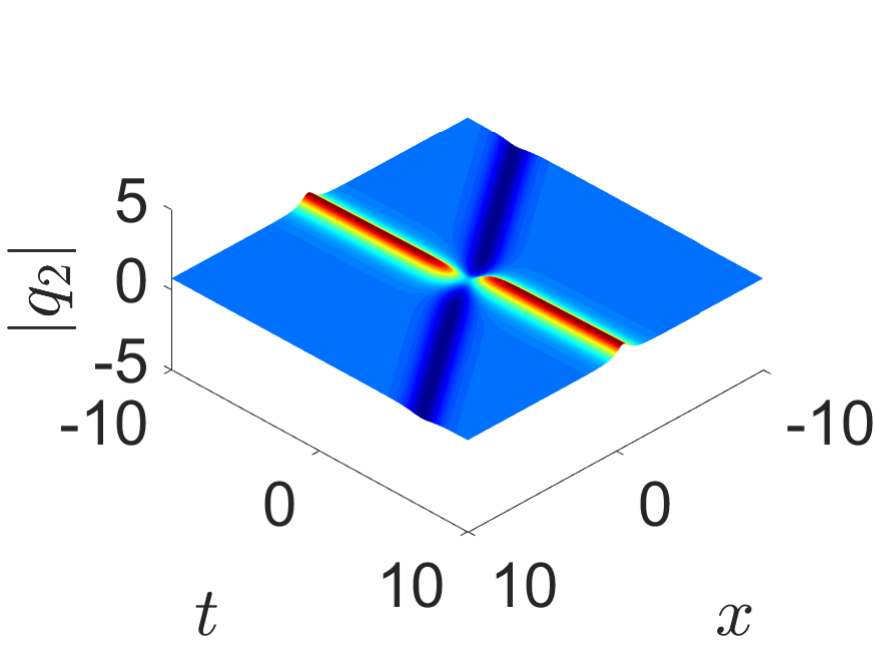} &
\includegraphics[width=0.22\textwidth]{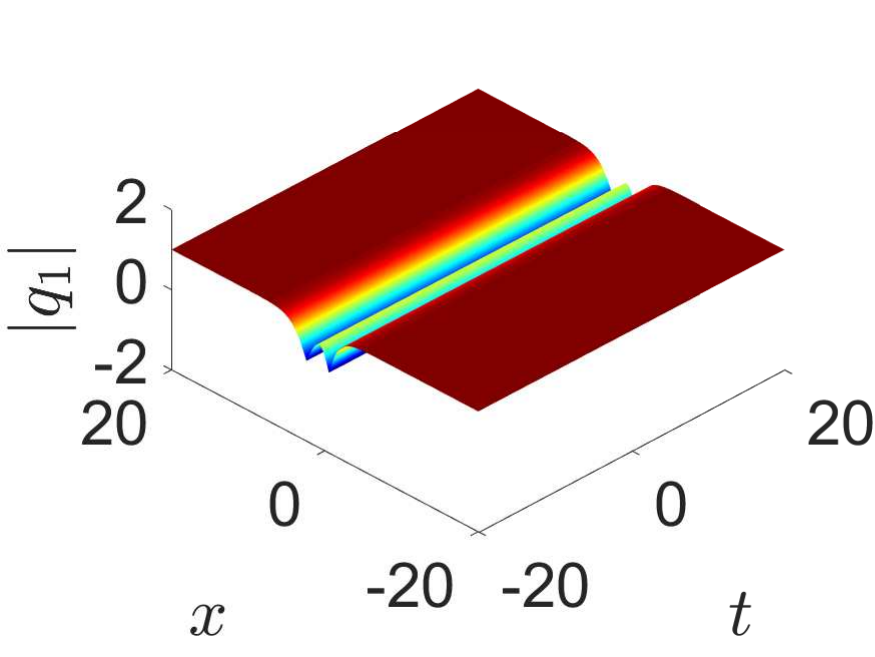} &
\includegraphics[width=0.22\textwidth]{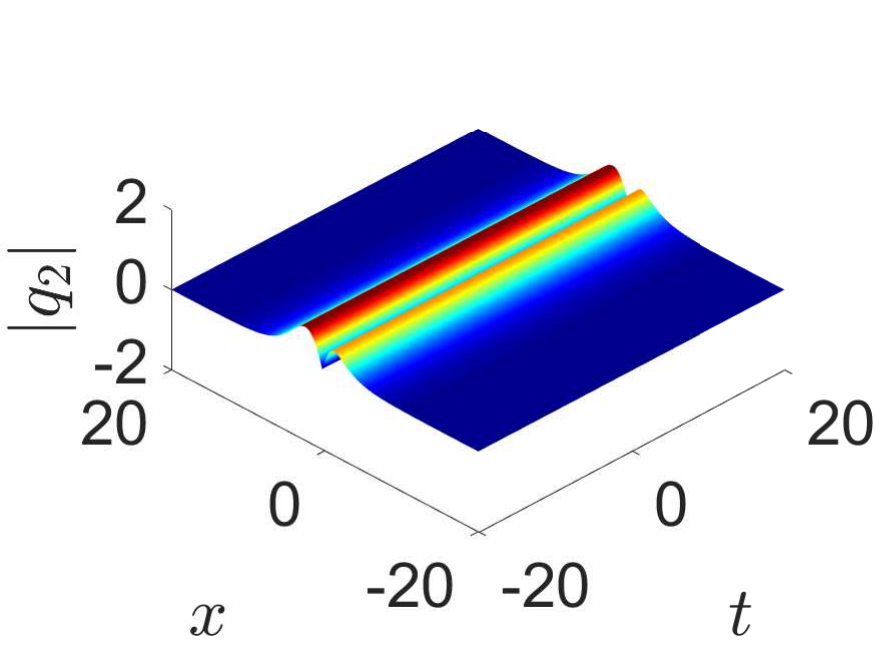} \\
		(c1) & (c2) & (d1) & (d2)
	\end{tabular}
\caption{\small (a1) and (a2): Two bright-bright solitons solution by taking $\mathbf{q}_{-}=(\frac{1}{2},-\frac{\sqrt{3}}{2})^{T}$, $\sigma=10^{-2}$, $\kappa_{3}=\mathrm{e}^{-1-\frac{3}{2}\mathrm{i}}$, $\kappa_{4}=\mathrm{e}^{-1-\mathrm{i}}$, $\alpha_{3}=\frac{1}{2}\pi$, $\alpha_{4}=\frac{3}{4}\pi$, $k_{1}=k_{2}=0$. (b1) and (b2): One dark-dark and one bright-bright solitons solution by taking $\mathbf{q}_{-}=(\frac{1}{2},-\frac{\sqrt{3}}{2})^{T}$, $\sigma=10^{-1}$, $\kappa_{3}=1$, $\kappa_{4}=\mathrm{e}^{-1-\frac{3}{2}\mathrm{i}}$, $\alpha_{3}=\frac{1}{2}\pi$, $\alpha_{4}=\frac{3}{4}\pi$, $k_{1}=1$, $k_{2}=0$. (c1) and (c2): One bright-bright and one dark-dark solitons solution by taking $\mathbf{q}_{-}=(\frac{\sqrt{2}}{2},-\frac{\sqrt{2}}{2})^{T}$, $\sigma=10^{-2}$, $\kappa_{3}=\mathrm{e}^{\frac{1}{10}+\mathrm{i}}$, $\kappa_{4}=1$, $\alpha_{3}=\frac{1}{2}\pi$, $\alpha_{4}=\frac{3}{4}\pi$, $k_{1}=0$, $k_{2}=1$. (d1) and (d2): Two parallel dark-bright solitons solution by taking $\mathbf{q}_{-}=(1,0)^{T}$, $\sigma=10^{-3}$, $K_{6}=\kappa_{5}=\kappa_{6}=\chi_{6}=0.5$, $\alpha_{5}=\alpha_{6}=\frac{1}{2}\pi$, $k=1$.
}\label{fig:3}
\end{figure}

One can discern a scenario in which one discrete eigenvalue lies on the circle while the other is situated off the circle: $G_{1}=G_{2}=1$. Subsequently, we encapsulate the discrete eigenvalues and norming constants within
\begin{subequations}\label{4.42}
\begin{align}
z_{1}&=q_{0}\mathrm{e}^{\mathrm{i}\alpha_{5}}, \quad
c_{1}=\mathrm{e}^{\kappa_{5}+\mathrm{i}[\alpha_{5}+(k-\frac{1}{2})\pi]}, \quad k=0,1,
\quad 0<\alpha_{5},\alpha_{6}<\pi,  \\
\theta_{1}&=K_{6}\mathrm{e}^{\mathrm{i}\alpha_{6}}, \quad f_{1}=\mathrm{e}^{\kappa_{6}+\mathrm{i}\chi_{6}}, \quad 0<K_{6}<q_{0}, \quad
\chi_{6}\in \mathbb{R}.
\end{align}
\end{subequations}
Based on the analysis of two types of one-soliton solutions, the combination of these two types of solitons has six results in the case of $G_{1}=G_{2}=1$. For $k=1$, $q_{1-}\times q_{2-}=0$ and $\kappa_{5},\kappa_{6}\in \mathbb{R}$, two parallel dark-bright solitons solution is given by (d1) and (d2) in Fig.~\ref{fig:3}. Moreover, setting $k=1$, $q_{1-}\times q_{2-}=0$, $\kappa_{5}\in \mathbb{R}$ and $\kappa_{6}\in \mathbb{C}$ generates one parallel dark-bright and one parallel bright-bright solitons solution in (a1) and (a2) of Fig.~\ref{fig:4}. In addition, two parallel bright-bright solitons solution is obtained by selecting parameters $k=0$, $q_{1-}\times q_{2-}=0$ and $\kappa_{5},\kappa_{6}\in \mathbb{C}$ in (b1) and (b2) of Fig.~\ref{fig:4}. Then, setting $k=0$, $q_{1-}\times q_{2-}=0$, $\kappa_{5}\in \mathbb{C}$ and $\kappa_{6}\in \mathbb{R}$ yields one parallel bright-bright and one parallel dark-bright solitons solution in (c1) and (c2) of Fig.~\ref{fig:4}. Moreover, setting $k=1$, $q_{1-}\times q_{2-}\neq0$ and $\kappa_{5}\in \mathbb{R}$ generates one dark-dark and one breather-breather solitons solution in (d1) and (d2) of Fig.~\ref{fig:4}. In addition, one bright-bright and one breather-breather solitons solution is obtained by selecting parameters $k=0$, $q_{1-}\times q_{2-}=0$ and $\kappa_{5}\in \mathbb{C}$ in (a1) and (a2) of Fig.~\ref{fig:5}.

In Fig.~\ref{fig:3}, (d1) and (d2) can be regarded as W-type soliton and M-type soliton solutions. In Fig.~\ref{fig:4}, (b1) and (b2) can be regarded as M-type soliton and M-type soliton solutions.

\begin{figure}[htb]
\centering
\begin{tabular}{cccc}
\includegraphics[width=0.22\textwidth]{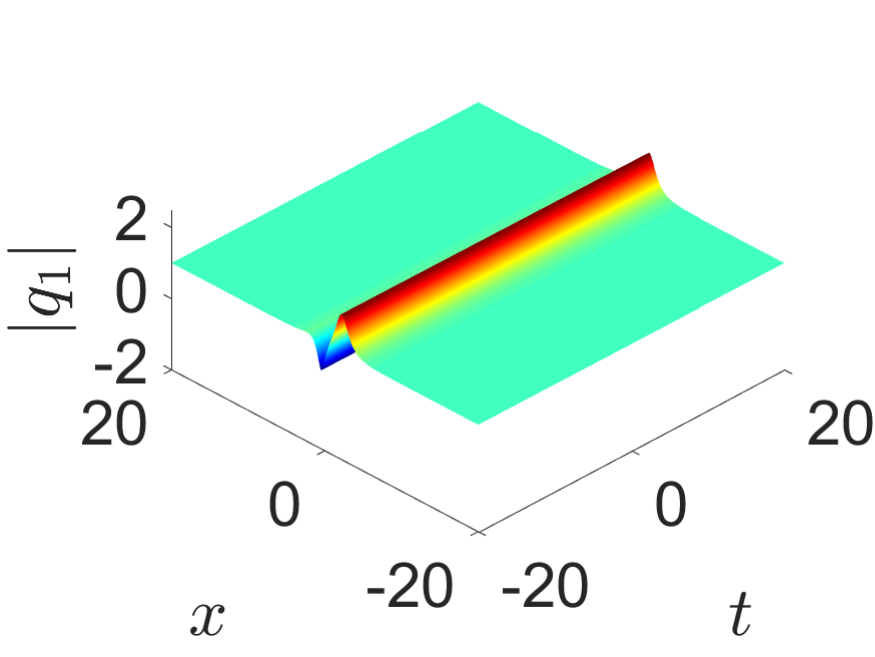} &
\includegraphics[width=0.22\textwidth]{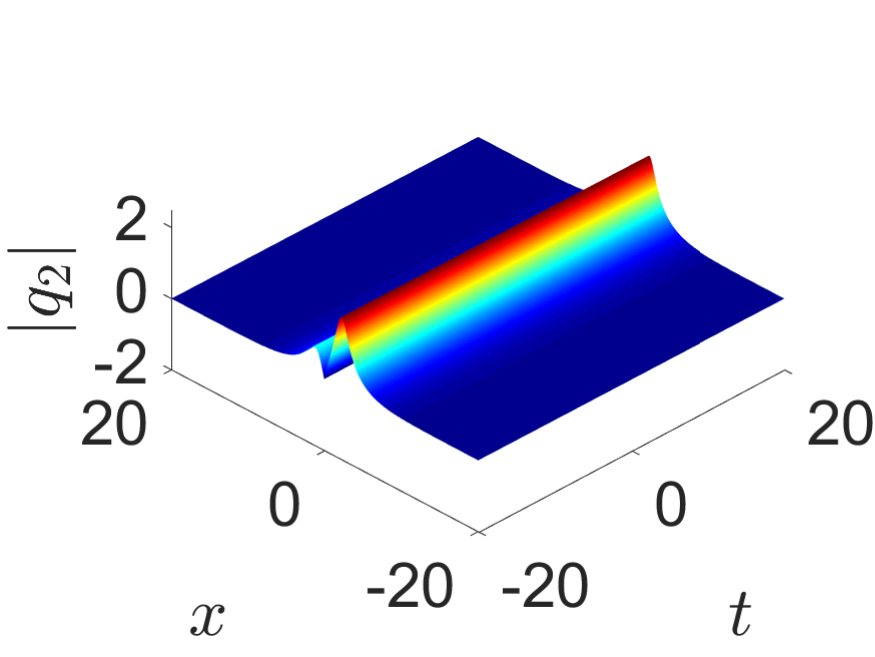} &
\includegraphics[width=0.22\textwidth]{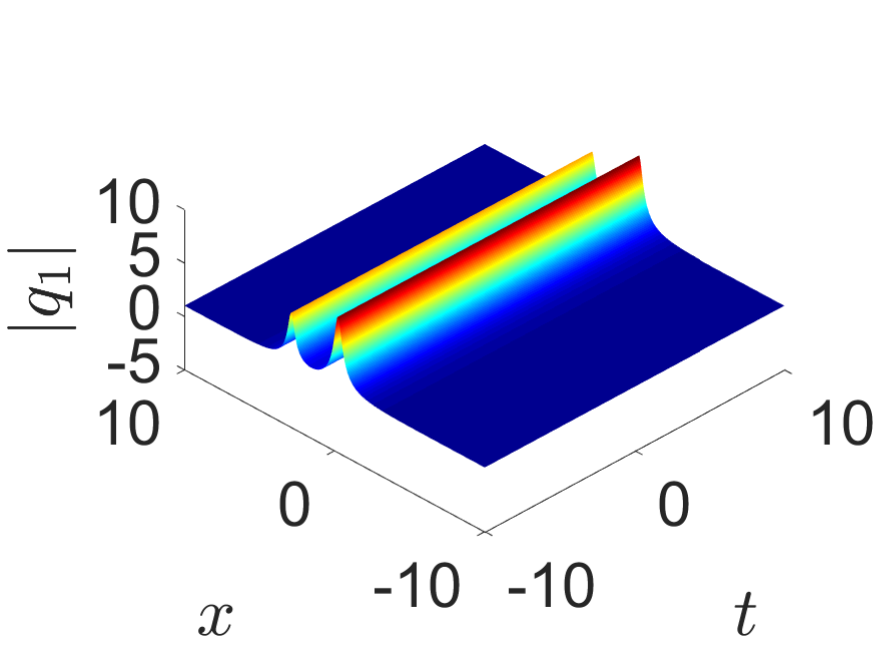} &
\includegraphics[width=0.22\textwidth]{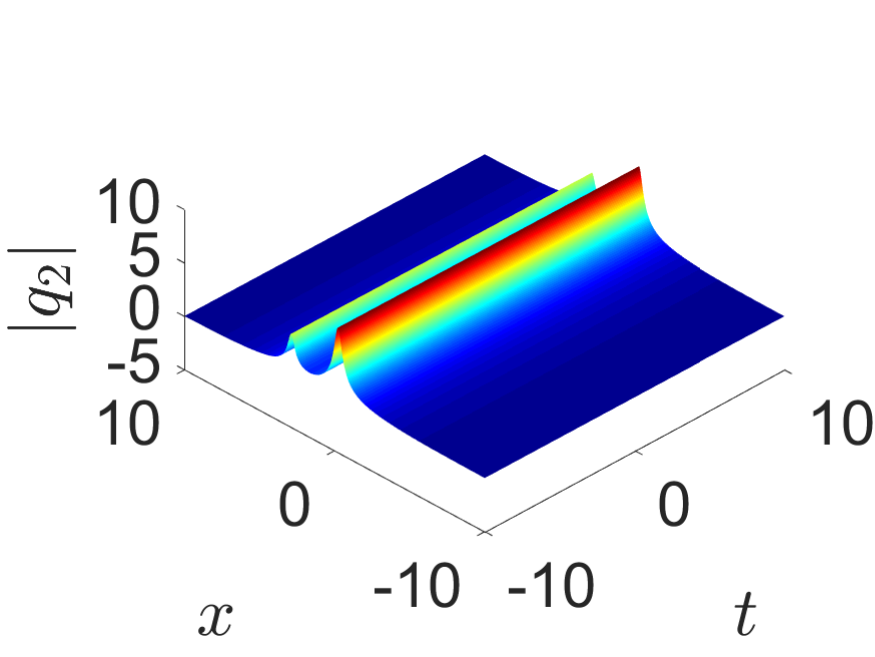} \\
		(a1) & (a2) & (b1) & (b2)
	\end{tabular}
\begin{tabular}{cccc}
\includegraphics[width=0.22\textwidth]{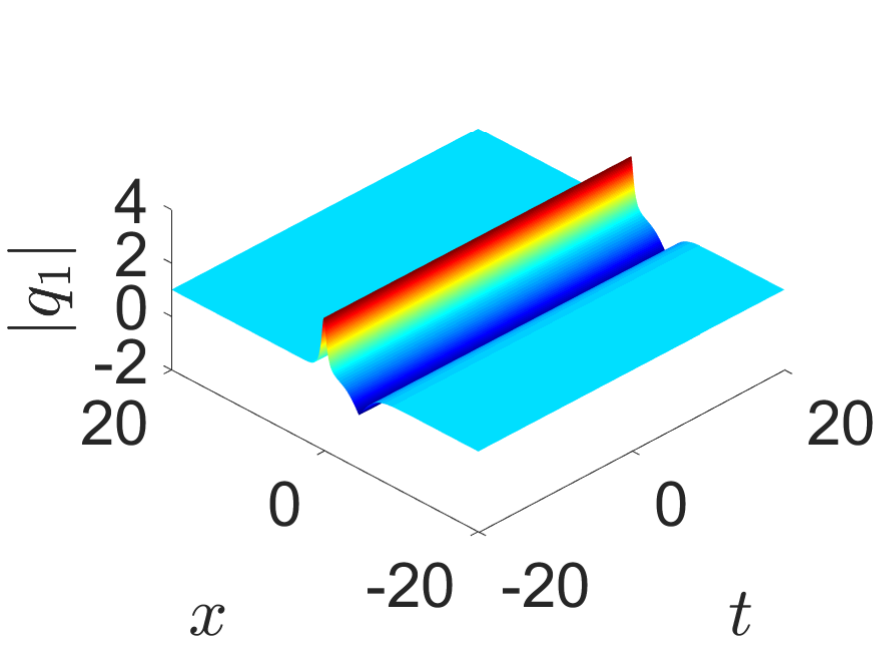} &
\includegraphics[width=0.22\textwidth]{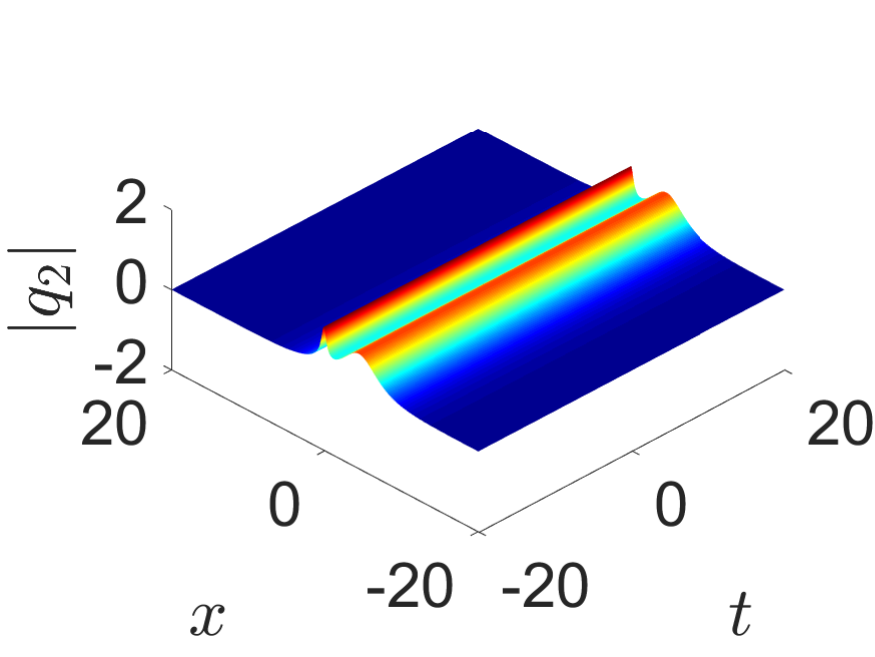} &
\includegraphics[width=0.22\textwidth]{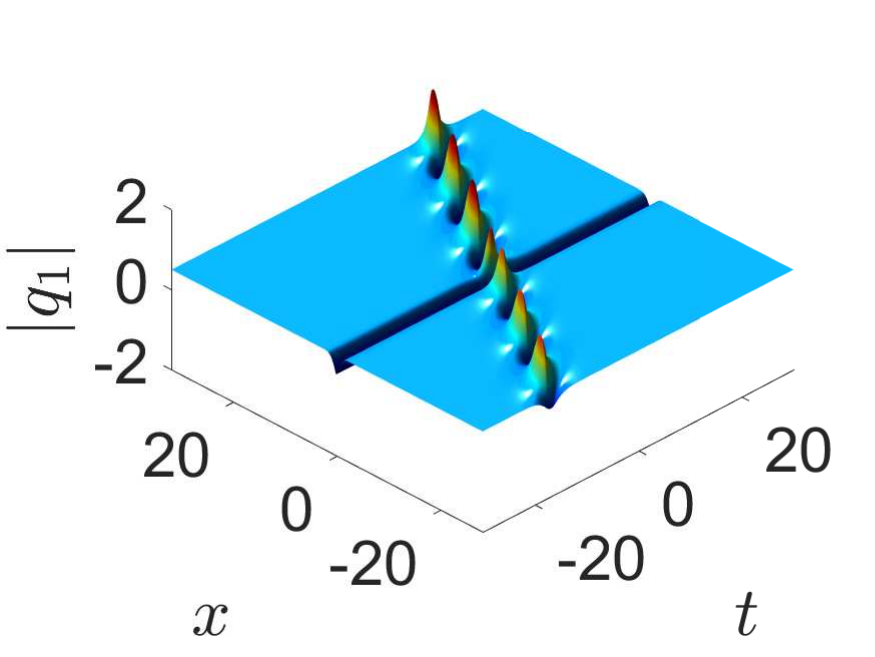} &
\includegraphics[width=0.22\textwidth]{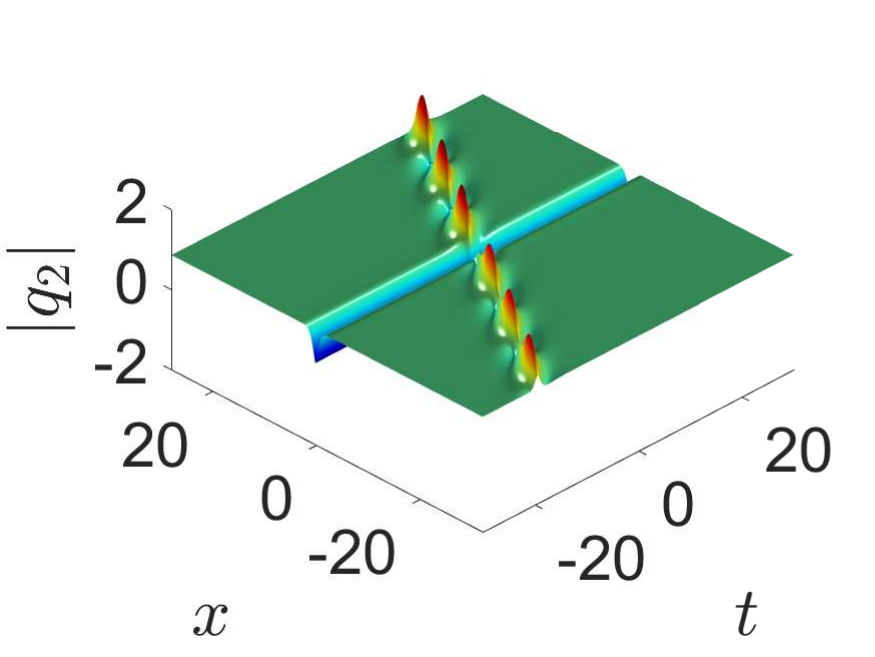} \\
		(c1) & (c2) & (d1) & (d2)
	\end{tabular}
\caption{\small (a1) and (a2): One parallel dark-bright and one parallel bright-bright solitons solution by taking $\mathbf{q}_{-}=(1,0)^{T}$, $\sigma=10^{-3}$, $K_{6}=0.5$, $\kappa_{5}=\chi_{6}=1$, $\kappa_{6}=\mathrm{e}^{\frac{7}{10}+\frac{3}{5}\mathrm{i}}$, $\alpha_{5}=\alpha_{6}=\frac{1}{2}\pi$, $k=1$. (b1) and (b2): Two parallel bright-bright solitons solution by taking $\mathbf{q}_{-}=(1,0)^{T}$, $\sigma=10^{-3}$, $K_{6}=0.5$, $\kappa_{5}=\mathrm{e}^{2-2\mathrm{i}}$, $\kappa_{6}=\mathrm{e}^{\frac{3}{5}-\frac{6}{5}\mathrm{i}}$, $\chi_{6}=1.5$, $\alpha_{5}=\alpha_{6}=\frac{1}{2}\pi$, $k=0$. (c1) and (c2): One parallel bright-bright and one parallel dark-bright solitons solution by taking $\mathbf{q}_{-}=(1,0)^{T}$, $\sigma=10^{-3}$, $K_{6}=0.5$, $\kappa_{5}=\mathrm{e}^{\frac{7}{10}+\frac{3}{10}\mathrm{i}}$, $\kappa_{6}=2.8$, $\chi_{6}=-0.2$, $\alpha_{5}=\alpha_{6}=\frac{1}{2}\pi$, $k=0$. (d1) and (d2): One dark-dark and one breather-breather solitons solution by taking $\mathbf{q}_{-}=(\frac{1}{2}\mathrm{e}^{-\frac{1}{10}\mathrm{i}\pi},
-\frac{\sqrt{3}}{2}\mathrm{e}^{-\frac{1}{10}\mathrm{i}\pi})^{T}$, $\sigma=10^{-3}$, $K_{6}=0.98$, $\kappa_{5}=\chi_{6}=1$, $\kappa_{6}=0.5$, $\alpha_{5}=\frac{1}{2}\pi$, $\alpha_{6}=\frac{4}{5}\pi$, $k=1$.
}\label{fig:4}
\end{figure}

Following that, we examine the case where both eigenvalues are located outside the circle ($G_{1}=0$ and $G_{2}=2$) and define additional parameters accordingly.
\begin{subequations}\label{4.43}
\begin{align}
\theta_{1}&=K_{7}\mathrm{e}^{\mathrm{i}\alpha_{7}}, \quad f_{1}=\mathrm{e}^{\kappa_{7}+\mathrm{i}\chi_{7}}, \quad 0<K_{7}<q_{0}, \quad
\chi_{7}\in \mathbb{R},  \\
\theta_{2}&=K_{8}\mathrm{e}^{\mathrm{i}\alpha_{8}}, \quad f_{2}=\mathrm{e}^{\kappa_{8}+\mathrm{i}\chi_{8}}, \quad 0<K_{8}<q_{0}, \quad
\chi_{8}\in \mathbb{R}.
\end{align}
\end{subequations}
Through the expressions~\eqref{4.43}, soliton solutions and the reflectionless potentials, it can be inferred that the different structures of the two-soliton solutions are obtained. For $q_{1-}\times q_{2-}=0$ and $\kappa_{7},\kappa_{8}\in \mathbb{R}$, two dark-bright solitons solution is given by (b1) and (b2) in Fig.~\ref{fig:5}. Moreover, setting $q_{1-}\times q_{2-}=0$ and $\kappa_{7},\kappa_{8}\in \mathbb{C}$ generates two bright-bright solitons solution in (c1) and (c2) of Fig.~\ref{fig:5}. In addition, two breather-breather solitons solution is obtained by selecting parameters  $q_{1-}\times q_{2-}\neq0$ in (d1) and (d2) of Fig.~\ref{fig:5}.

\begin{figure}[htb]
\centering
\begin{tabular}{cccc}
\includegraphics[width=0.22\textwidth]{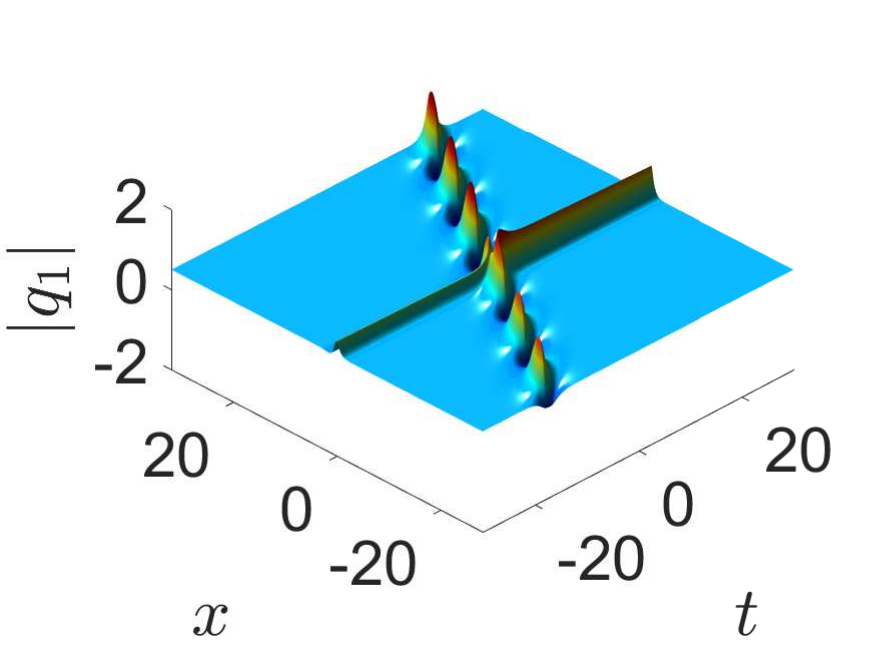} &
\includegraphics[width=0.22\textwidth]{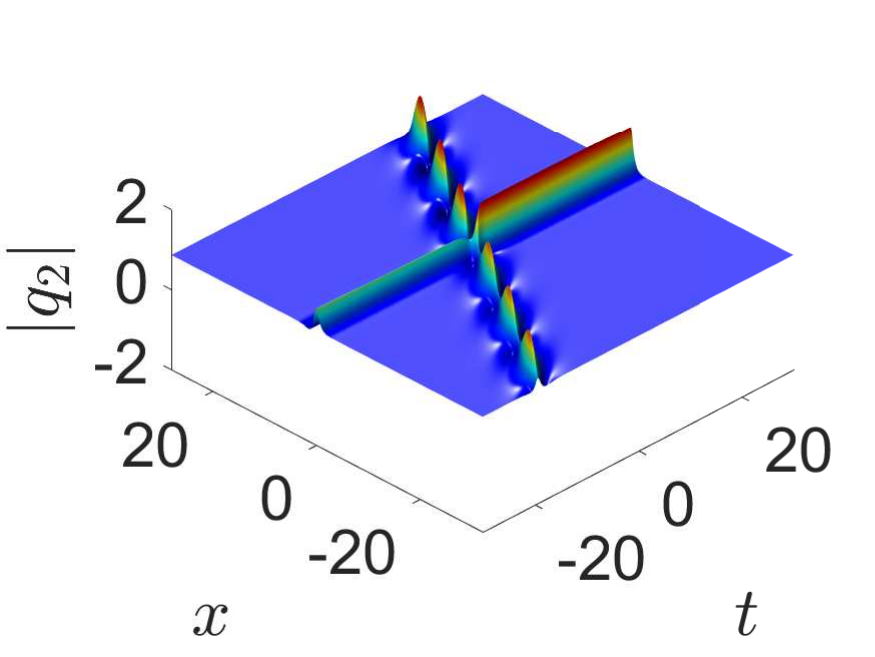} &
\includegraphics[width=0.22\textwidth]{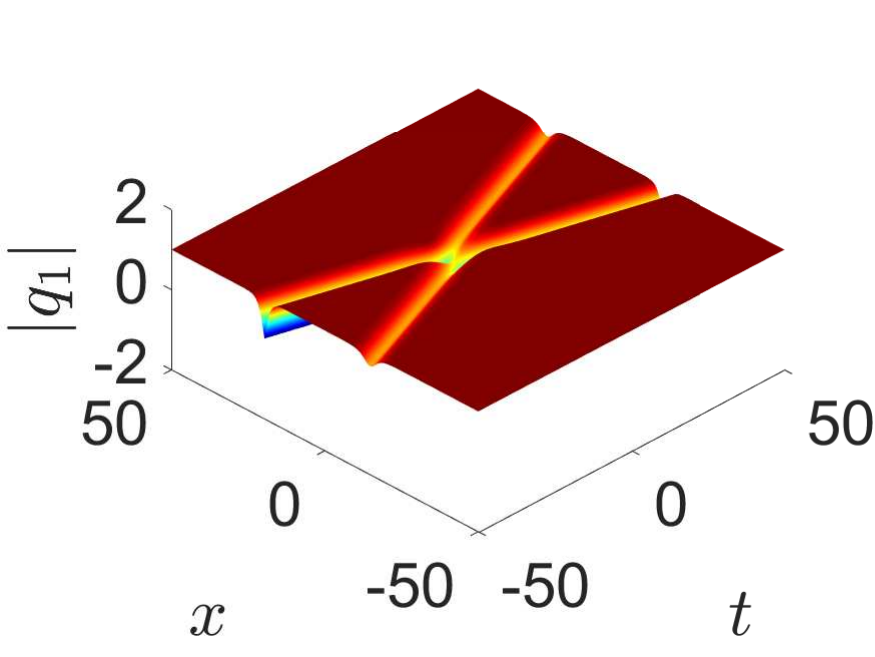} &
\includegraphics[width=0.22\textwidth]{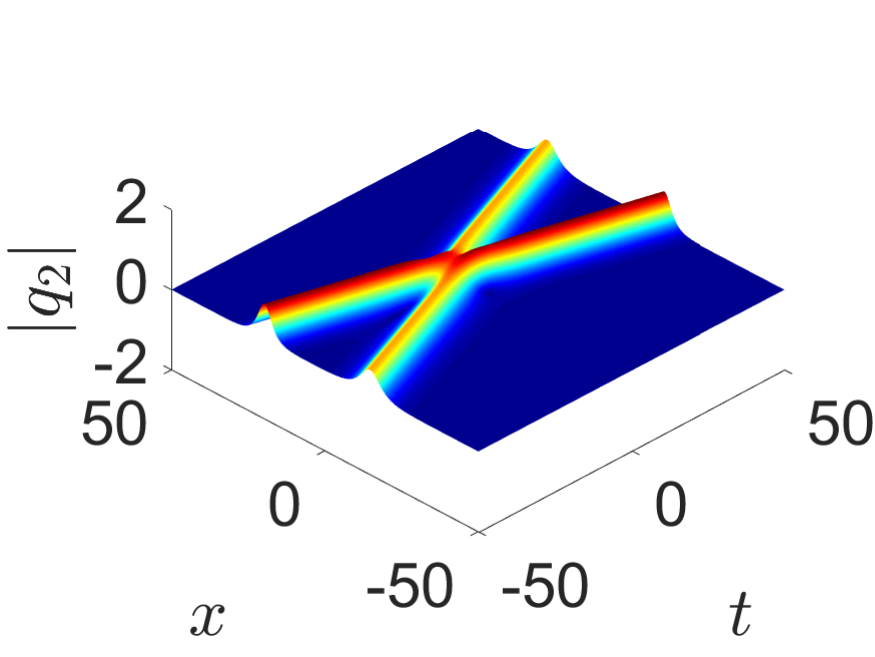} \\
		(a1) & (a2) & (b1) & (b2)
	\end{tabular}
\begin{tabular}{cccc}
\includegraphics[width=0.22\textwidth]{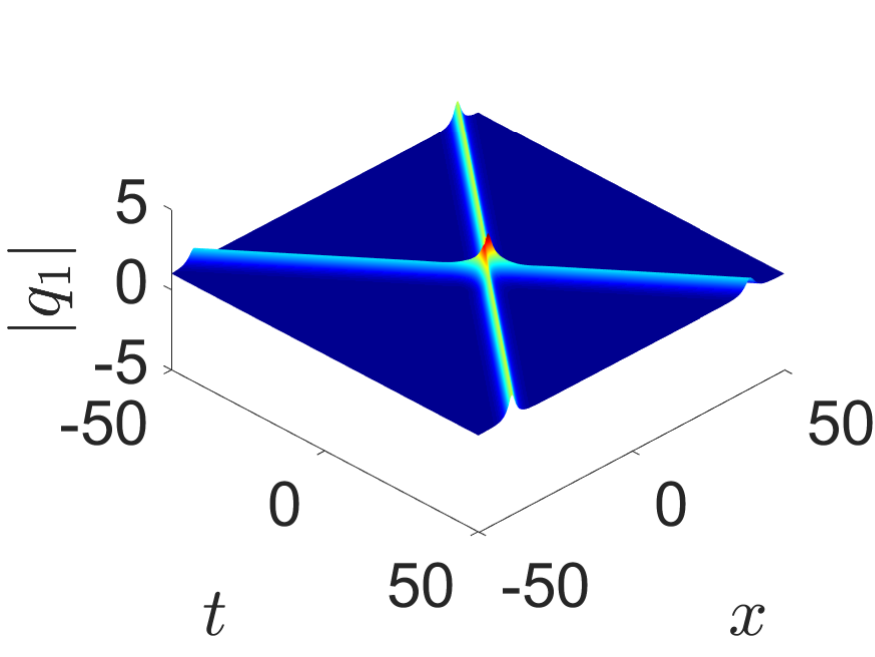} &
\includegraphics[width=0.22\textwidth]{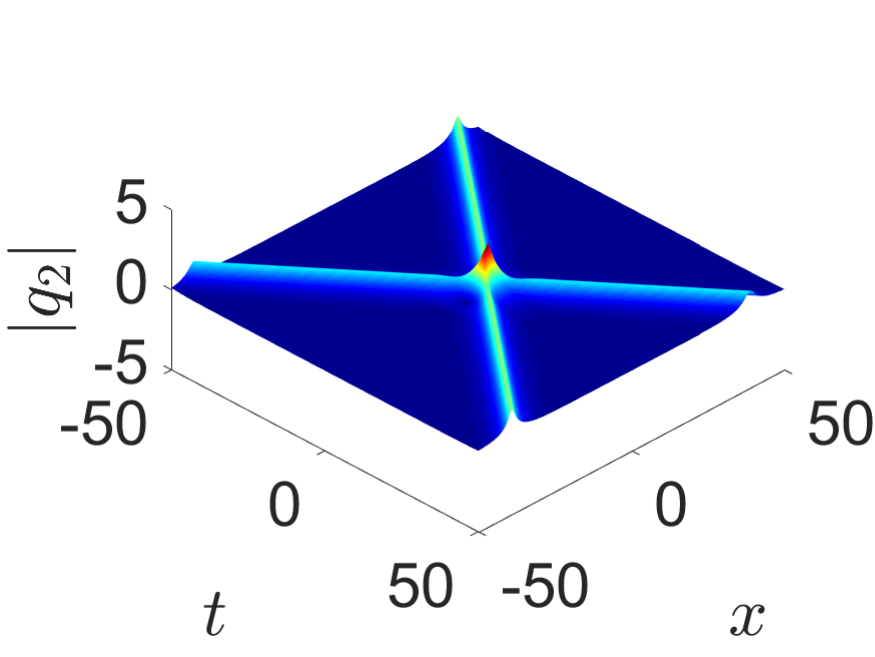} &
\includegraphics[width=0.22\textwidth]{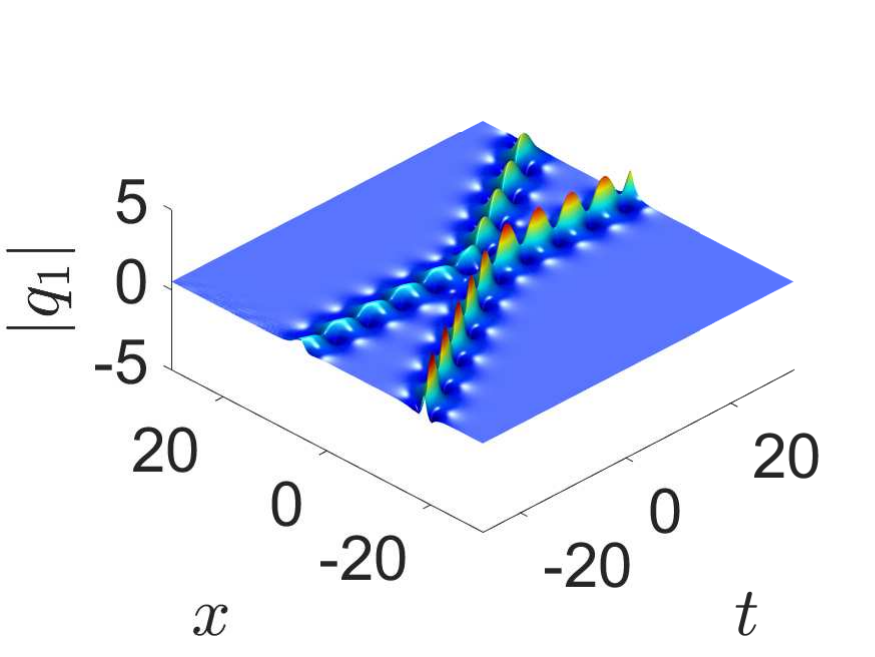} &
\includegraphics[width=0.22\textwidth]{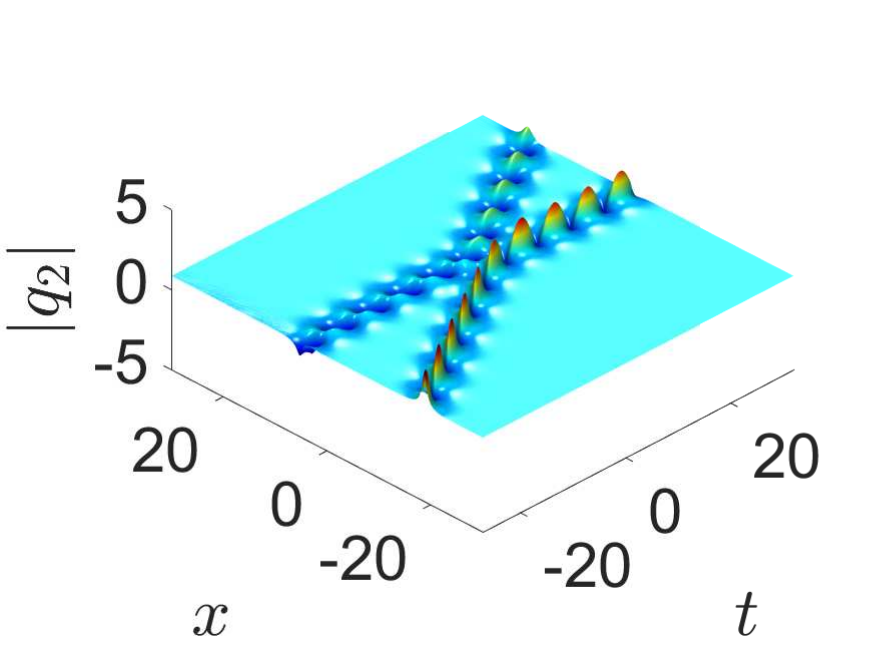} \\
		(c1) & (c2) & (d1) & (d2)
	\end{tabular}
\caption{\small (a1) and (a2): One bright-bright and one breather-breather solitons solution by taking $\mathbf{q}_{-}=(\frac{1}{2}\mathrm{e}^{-\frac{1}{10}\mathrm{i}\pi},
-\frac{\sqrt{3}}{2}\mathrm{e}^{-\frac{1}{10}\mathrm{i}\pi})^{T}$, $\sigma=10^{-3}$, $K_{6}=0.98$, $\kappa_{5}=\mathrm{e}^{\frac{1}{5}-\mathrm{i}}$, $\chi_{6}=1$, $\kappa_{6}=0.2$, $\alpha_{5}=\frac{1}{2}\pi$, $\alpha_{6}=\frac{4}{5}\pi$, $k=0$. (b1) and (b2): Two dark-bright solitons solution by taking $\mathbf{q}_{-}=(1,0)^{T}$, $\sigma=10^{-1}$, $K_{7}=K_{8}=0.5$, $\chi_{7}=1$, $\kappa_{7}=\kappa_{8}=\chi_{8}=-1$, $\alpha_{7}=\frac{1}{2}\pi$, $\alpha_{8}=\frac{3}{4}\pi$. (c1) and (c2): Two bright-bright solitons solution by taking $\mathbf{q}_{-}=(1,0)^{T}$, $\sigma=10^{-3}$, $K_{7}=K_{8}=0.5$, $\chi_{7}=1$, $\chi_{8}=0.8$, $\kappa_{7}=\mathrm{e}^{\frac{1}{5}-\frac{3}{2}\mathrm{i}}$, $\kappa_{8}=\mathrm{e}^{\frac{1}{2}+\frac{7}{10}\mathrm{i}}$, $\alpha_{7}=\frac{1}{5}\pi$, $\alpha_{8}=\frac{4}{5}\pi$. (d1) and (d2): Two breather-breather solitons solution by taking $\mathbf{q}_{-}=(\frac{1}{2}\mathrm{e}^{\frac{1}{5}\mathrm{i}\pi},
-\frac{\sqrt{3}}{2}\mathrm{e}^{\frac{1}{5}\mathrm{i}\pi})^{T}$, $\sigma=10^{-3}$, $K_{7}=K_{8}=0.98$, $\chi_{7}=0.5$, $\chi_{8}=1.5$, $\kappa_{7}=-1.5$, $\kappa_{8}=2$, $\alpha_{7}=\frac{1}{2}\pi$, $\alpha_{8}=\frac{3}{5}\pi$.
}\label{fig:5}
\end{figure}

\section{Double-pole solutions}
\label{s:Double-pole solutions}

The situation of the defocusing-defocusing coupled Hirota equations with NZBC~\eqref{1.3} when the analytical scattering coefficient has double zeros was obtained~\cite{B1}. We shall denote the pertinent solutions as the ``double-pole" solutions pertaining to the equations~\eqref{1.3}.

\subsection{Behavior of the eigenfunctions at a double pole}

\begin{proposition}\label{pro:9}
Suppose that $h_{11}(\theta_{g})=h_{11}^{\prime}(\theta_{g})=0$ and  $h_{11}^{\prime\prime}(\theta_{g})\neq 0$, with $\left|\theta_{g}\right|<q_{0}$, then exist constants $b_{g}$, $\hat{b}_{g}$, $\check{b}_{g}$, $\bar{b}_{g}$, $f_{g}$, $\hat{f}_{g}$, $\check{f}_{g}$, $\bar{f}_{g}$, $e_{g}$, $\hat{e}_{g}$, $\check{e}_{g}$ and $\bar{e}_{g}$ such that
\begin{subequations}\label{5.1}
\begin{align}
\psi_{+,1}^{\prime}(\theta_{g};x,t)&=\frac{f_{g}}{s_{33}(\theta_{g})}
\widetilde{\gamma}^{\prime}(\theta_{g})+b_{g}\widetilde{\gamma}(\theta_{g})
+e_{g}\psi_{-,3}(\theta_{g}), \\
\widetilde{\gamma}^{\prime}(\frac{q_{0}^{2}}{\theta_{g}^{*}};x,t)&=
\hat{f}_{g}\psi_{-,3}^{\prime}(\frac{q_{0}^{2}}{\theta_{g}^{*}})
+\hat{b}_{g}\psi_{-,3}(\frac{q_{0}^{2}}{\theta_{g}^{*}})
+\hat{e}_{g}\psi_{+,1}(\frac{q_{0}^{2}}{\theta_{g}^{*}}), \\
\gamma^{\prime}(\theta_{g}^{*};x,t)&=\bar{f}_{g}\psi_{-,1}^{\prime}(\theta_{g}^{*})
+\bar{b}_{g}\psi_{-,1}(\theta_{g}^{*})+\bar{e}_{g}\psi_{+,3}(\theta_{g}^{*}), \\
\psi_{+,3}^{\prime}(\frac{q_{0}^{2}}{\theta_{g}};x,t)&=
\check{f}_{g}\gamma^{\prime}(\frac{q_{0}^{2}}{\theta_{g}})
+\check{b}_{g}\gamma(\frac{q_{0}^{2}}{\theta_{g}})
+\check{e}_{g}\psi_{-,1}(\frac{q_{0}^{2}}{\theta_{g}}),
\end{align}
\end{subequations}
and corresponding modified eigenfunctions are
\begin{equation}\label{5.2}
\begin{split}
\nu_{+,1}^{\prime}(\theta_{g};x,t)&=-\mathrm{i}\delta_{1}^{\prime}(\theta_{g})
\nu_{+,1}(\theta_{g})+e_{g}\nu_{-,3}(\theta_{g})
\mathrm{e}^{-2\mathrm{i}\delta_{1}(\theta_{g})} \\
&+\left[ \left(\frac{\mathrm{i}f_{g}\delta_{2}^{\prime}(\theta_{g})}{s_{33}(\theta_{g})}+b_{g} \right) \widetilde{d}(\theta_{g})
+\frac{f_{g}}{s_{33}(\theta_{g})}\widetilde{d}^{\prime}(\theta_{g}) \right] \mathrm{e}^{\mathrm{i}[\delta_{2}(\theta_{g})-\delta_{1}(\theta_{g})]},
\end{split}
\end{equation}
\begin{equation}\label{5.3}
\begin{split}
\widetilde{d}^{\prime}(\frac{q_{0}^{2}}{\theta_{g}^{*}};x,t)&
=-\mathrm{i}\delta_{2}^{\prime}(\frac{q_{0}^{2}}{\theta_{g}^{*}})
\widetilde{d}(\frac{q_{0}^{2}}{\theta_{g}^{*}})
+\hat{e}_{g}\nu_{+,1}(\frac{q_{0}^{2}}{\theta_{g}^{*}})
\mathrm{e}^{-\mathrm{i}[\delta_{1}(\theta_{g}^{*})+\delta_{2}(\theta_{g}^{*})]} \\
&+\left[ \left( \hat{b}_{g}-\mathrm{i}\hat{f}_{g}\delta_{1}^{\prime}(\frac{q_{0}^{2}}{\theta_{g}^{*}}) \right) \nu_{-,3}(\frac{q_{0}^{2}}{\theta_{g}^{*}})
+\hat{f}_{g}\nu_{-,3}^{\prime}(\frac{q_{0}^{2}}{\theta_{g}^{*}}) \right] \mathrm{e}^{\mathrm{i}[\delta_{1}(\theta_{g}^{*})-\delta_{2}(\theta_{g}^{*})]},
\end{split}
\end{equation}
\begin{equation}\label{5.4}
\begin{split}
d^{\prime}(\theta_{g}^{*};x,t)&=-\mathrm{i}\delta_{2}^{\prime}(\theta_{g}^{*})
d(\theta_{g}^{*})+\bar{e}_{g}\nu_{+,3}(\theta_{g}^{*})
\mathrm{e}^{-\mathrm{i}[\delta_{1}(\theta_{g}^{*})+\delta_{2}(\theta_{g}^{*})]} \\
&+\left[ \left( \mathrm{i}\bar{f}_{g}\delta_{1}^{\prime}(\theta_{g}^{*})+\bar{b}_{g} \right) \nu_{-,1}(\theta_{g}^{*})
+\bar{f_{g}}\nu_{-,1}^{\prime}(\theta_{g}^{*}) \right] \mathrm{e}^{\mathrm{i}[\delta_{1}(\theta_{g}^{*})-\delta_{2}(\theta_{g}^{*})]},
\end{split}
\end{equation}
\begin{equation}\label{5.5}
\begin{split}
\nu_{+,3}^{\prime}(\frac{q_{0}^{2}}{\theta_{g}};x,t)&
=\mathrm{i}\delta_{1}^{\prime}(\frac{q_{0}^{2}}{\theta_{g}})\nu_{+,3}(\frac{q_{0}^{2}}{\theta_{g}})
+\check{e}_{g}\nu_{-,1}(\frac{q_{0}^{2}}{\theta_{g}})
\mathrm{e}^{-2\mathrm{i}\delta_{1}(\theta_{g})} \\
&+\left[ \left(
\check{b}_{g}+\mathrm{i}\check{f}_{g}\delta_{2}^{\prime}(\frac{q_{0}^{2}}{\theta_{g}}) \right) d(\frac{q_{0}^{2}}{\theta_{g}})
+\check{f}_{g}d^{\prime}(\frac{q_{0}^{2}}{\theta_{g}}) \right] \mathrm{e}^{\mathrm{i}[\delta_{2}(\theta_{g})-\delta_{1}(\theta_{g})]},
\end{split}
\end{equation}
where $f_{g}$, $\hat{f}_{g}$, $\check{f}_{g}$ and $\bar{f}_{g}$ are the same norming constants in the symmetry relationship~\eqref{3.7}, whereas $b_{g}$, $\hat{b}_{g}$, $\check{b}_{g}$, $\bar{b}_{g}$, $e_{g}$, $\hat{e}_{g}$, $\check{e}_{g}$ and $\bar{e}_{g}$ appear as a result of the double multiplicity.
\end{proposition}

\begin{proposition}\label{pro:10}
Assuming that $y(z)$ and $h(z)$ are analytic in $\mathbb{D}^{+}$ and $\theta_{g}$ is a double zero of $h(z)$, then expanding $y(z)$ and $h(z)$ as Taylor expansions at $z=\theta_{g}$
\begin{equation}\label{5.6}
\begin{split}
\frac{y(z)}{h(z)}=\left[ \frac{2y^{\prime}(\theta_{g})}{h^{\prime\prime}(\theta_{g})}- \frac{2y(\theta_{g})h^{\prime\prime\prime}(\theta_{g})}{3[h^{\prime\prime}(\theta_{g})]^{2}} \right]\frac{1}{(z-\theta_{g})}
+\frac{2y(\theta_{g})}{h^{\prime\prime}(\theta_{g})}\frac{1}{(z-\theta_{g})^{2}}+\cdots.
\end{split}
\end{equation}
Therefore, it can be seen that the coefficients of $(z-\theta_{g})^{-1}$ and $(z-\theta_{g})^{-2}$ in the series expansion of $y(z)/h(z)$ near $z=\theta_{g}$ are as follows
\begin{equation}\label{5.7}
\begin{split}
\mathop{\rm{Res}}_{z=\theta_{g}}\left[ \frac{y(z)}{h(z)} \right]=
\frac{2y^{\prime}(\theta_{g})}{h^{\prime\prime}(\theta_{g})}
-\frac{2y(\theta_{g})h^{\prime\prime\prime}(\theta_{g})}{3[h^{\prime\prime}(\theta_{g})]^{2}}, \quad
\mathop{\rm{Y_{-2}}}_{z=\theta_{g}}\left[ \frac{y(z)}{h(z)} \right]=
\frac{2y(\theta_{g})}{h^{\prime\prime}(\theta_{g})}.
\end{split}
\end{equation}
\end{proposition}

\begin{corollary}\label{cor:9}
The generalization of the negative second power coefficients and negative first power coefficients will be obtained.
\begin{subequations}\label{5.8}
\begin{align}
\mathop{\rm{Y_{-2}}}_{z=\theta_{g}}\left[ \frac{\nu_{+,1}(z)}{h_{11}(z)} \right]&=
W_{g}\widetilde{d}(\theta_{g}), \quad
\mathop{\rm{Y_{-2}}}_{z=\theta_{g}^{*}}\left[ \frac{d(z)}{s_{11}(z)} \right]=
\bar{W}_{g}\nu_{-,1}(\theta_{g}^{*}),  \\
\mathop{\rm{Y_{-2}}}_{z=q_{0}^{2}/\theta_{g}}\left[ \frac{\nu_{+,3}(z)}{h_{33}(z)} \right]&=
\check{W}_{g}d(\frac{q_{0}^{2}}{\theta_{g}}), \quad
\mathop{\rm{Y_{-2}}}_{z=q_{0}^{2}/\theta_{g}^{*}}\left[ -\frac{\widetilde{d}(z)}{s_{33}(z)} \right]=-\hat{W}_{g}\nu_{-,3}(\frac{q_{0}^{2}}{\theta_{g}^{*}}),
\end{align}
\end{subequations}
\begin{equation}\label{5.9}
\begin{split}
\mathop{\rm{Res}}_{z=\theta_{g}}\left[ \frac{\nu_{+,1}(z)}{h_{11}(z)} \right]&=
W_{g}\widetilde{d}^{\prime}(\theta_{g})+D_{g}\nu_{-,3}(\theta_{g})
+\left[ B_{g}-\mathrm{i}x-\mathrm{i}t \left( 2\theta_{g}+3\sigma(q_{0}^{2}+\theta_{g}^{2}) \right) \right]W_{g}\widetilde{d}(\theta_{g}),
\end{split}
\end{equation}
\begin{equation}\label{5.10}
\begin{split}
\mathop{\rm{Res}}_{z=\theta_{g}^{*}}\left[ \frac{d(z)}{s_{11}(z)} \right]&=
\bar{W}_{g}\nu_{-,1}^{\prime}(\theta_{g}^{*})+\bar{D}_{g}\nu_{+,3}(\theta_{g}^{*})
+\left[ \bar{B}_{g}+\mathrm{i}x+\mathrm{i}t \left( 2\theta_{g}^{*}
+3\sigma(q_{0}^{2}+(\theta_{g}^{*})^{2}) \right) \right]\bar{W}_{g}\nu_{-,1}(\theta_{g}^{*}),
\end{split}
\end{equation}
\begin{equation}\label{5.11}
\begin{split}
\mathop{\rm{Res}}_{z=q_{0}^{2}/\theta_{g}}\left[ \frac{\nu_{+,3}(z)}{h_{33}(z)} \right]&=
\check{W}_{g}d^{\prime}(\frac{q_{0}^{2}}{\theta_{g}})
+\check{D}_{g}\nu_{-,1}(\frac{q_{0}^{2}}{\theta_{g}})
+\left[ \check{B}_{g}+\frac{\mathrm{i}\theta_{g}^{2}}{q_{0}^{2}} \left(x+ t(2\theta_{g}+3\sigma(q_{0}^{2}+\theta_{g}^{2})) \right) \right]\check{W}_{g}d(\frac{q_{0}^{2}}{\theta_{g}}),
\end{split}
\end{equation}
\begin{equation}\label{5.12}
\begin{split}
\mathop{\rm{Res}}_{z=q_{0}^{2}/\theta_{g}^{*}}\left[ -\frac{\widetilde{d}(z)}{s_{33}(z)} \right]&=\left[\frac{\mathrm{i}(\theta_{g}^{*})^{2}}{q_{0}^{2}} \left(x+ t(2\theta_{g}^{*}+3\sigma(q_{0}^{2}+(\theta_{g}^{*})^{2})) \right) -\hat{B}_{g} \right]\hat{W}_{g}\nu_{-,3}(\frac{q_{0}^{2}}{\theta_{g}^{*}})
-\hat{W}_{g}\nu_{-,3}^{\prime}(\frac{q_{0}^{2}}{\theta_{g}^{*}}) \\
&-\hat{D}_{g}\nu_{+,1}(\frac{q_{0}^{2}}{\theta_{g}^{*}}),
\end{split}
\end{equation}
where
\begin{subequations}\label{5.13}
\begin{align}
W_{g}(x,t)&=\frac{2f_{g}\mathrm{e}^{\mathrm{i}[\delta_{2}(\theta_{g})
-\delta_{1}(\theta_{g})]}}{s_{33}(\theta_{g})h_{11}^{\prime\prime}(\theta_{g})}, \quad
B_{g}=\frac{b_{g}}{f_{g}}s_{33}(\theta_{g})
-\frac{h_{11}^{\prime\prime\prime}(\theta_{g})}{3h_{11}^{\prime\prime}(\theta_{g})},  \\
\bar{W}_{g}(x,t)&=\frac{2\bar{f}_{g}\mathrm{e}^{\mathrm{i}[\delta_{1}(\theta_{g}^{*})
-\delta_{2}(\theta_{g}^{*})]}}{s_{11}^{\prime\prime}(\theta_{g}^{*})}, \quad
\bar{D}_{g}(x,t)=\frac{2\bar{e}_{g}\mathrm{e}^{-\mathrm{i}[\delta_{1}(\theta_{g}^{*})
+\delta_{2}(\theta_{g}^{*})]}}{s_{11}^{\prime\prime}(\theta_{g}^{*})},  \\
\check{W}_{g}(x,t)&=\frac{2\check{f}_{g}\mathrm{e}^{\mathrm{i}[\delta_{2}(\theta_{g})
-\delta_{1}(\theta_{g})]}}{h_{33}^{\prime\prime}(q_{0}^{2}/\theta_{g})}, \quad
\check{D}_{g}(x,t)=\frac{2\check{e}_{g}\mathrm{e}^{-2\mathrm{i}
\delta_{1}(\theta_{g})}}{h_{33}^{\prime\prime}(q_{0}^{2}/\theta_{g})}, \quad
\bar{B}_{g}=\frac{\bar{b}_{g}}{\bar{f}_{g}}
-\frac{s_{11}^{\prime\prime\prime}(\theta_{g}^{*})}{3s_{11}^{\prime\prime}(\theta_{g}^{*})}, \\
\hat{W}_{g}(x,t)&=\frac{2\hat{f}_{g}\mathrm{e}^{\mathrm{i}[\delta_{1}(\theta_{g}^{*})
-\delta_{2}(\theta_{g}^{*})]}}{s_{33}^{\prime\prime}(q_{0}^{2}/\theta_{g}^{*})}, \quad
\hat{D}_{g}(x,t)=\frac{2\hat{e}_{g}\mathrm{e}^{-\mathrm{i}[\delta_{1}(\theta_{g}^{*})
+\delta_{2}(\theta_{g}^{*})]}}{s_{33}^{\prime\prime}(q_{0}^{2}/\theta_{g}^{*})}, \\
D_{g}(x,t)&=\frac{2e_{g}
\mathrm{e}^{-2\mathrm{i}\delta_{1}(\theta_{g})}}{h_{11}^{\prime\prime}(\theta_{g})}, \quad
\hat{B}_{g}=\frac{\hat{b}_{g}}{\hat{f}_{g}}-\frac{s_{33}^{\prime\prime\prime}(q_{0}^{2}
/\theta_{g}^{*})}{3s_{33}^{\prime\prime}(q_{0}^{2}/\theta_{g}^{*})}, \quad
\check{B}_{g}=\frac{\check{b}_{g}}{\check{f}_{g}}-\frac{h_{33}^{\prime\prime\prime}(q_{0}^{2}
/\theta_{g})}{3h_{33}^{\prime\prime}(q_{0}^{2}/\theta_{g})}.
\end{align}
\end{subequations}
\end{corollary}

\subsection{Symmetries with double poles}

\begin{proposition}\label{pro:11}
Suppose that $h_{11}(\theta_{g})=h_{11}^{\prime}(\theta_{g})=0$ and $h_{11}^{\prime \prime}(\theta_{g})\neq 0$, with $\left| \theta_{g} \right|<q_{0}$, then analytic scattering coefficients has the following symmetry relationship:
\begin{subequations}\label{5.14}
\begin{align}
h_{11}^{\prime\prime}(\theta_{g})
&=\left.\left[s_{11}^{\prime\prime}(z)\right]^{*}\right|_{z=\theta_{g}^{*}}, \quad s_{33}^{\prime\prime}(\theta_{g})
=\left.\left[h_{33}^{\prime\prime}(z)\right]^{*}\right|_{z=\theta_{g}^{*}}, \quad
s_{33}^{\prime\prime}(\frac{q_{0}^{2}}{\theta_{g}^{*}})
=\left.\frac{(\theta_{g}^{*})^{4}}{q_{0}^{4}}s_{11}^{\prime\prime}(z)\right|_{z=\theta_{g}^{*}}, \\
h_{11}^{\prime\prime\prime}(\theta_{g})
&=\left.\left[s_{11}^{\prime\prime\prime}(z)\right]^{*}\right|_{z=\theta_{g}^{*}}, \quad s_{33}^{\prime\prime\prime}(\theta_{g})
=\left.\left[h_{33}^{\prime\prime\prime}(z)\right]^{*}\right|_{z=\theta_{g}^{*}}, \quad
h_{11}^{\prime\prime}(\theta_{g})
=\left.\frac{q_{0}^{4}}{\theta_{g}^{4}}h_{33}^{\prime\prime}(z)\right|_{z=q_{0}^{2}/\theta_{g}}, \\
s_{33}^{\prime\prime\prime}(\frac{q_{0}^{2}}{\theta_{g}^{*}})
&=-\left.\frac{(\theta_{g}^{*})^{5}}{q_{0}^{6}}\left[6s_{11}^{\prime\prime}(z)+\theta_{g}^{*} s_{11}^{\prime\prime\prime}(z)\right]\right|_{z=\theta_{g}^{*}}, \quad
h_{11}^{\prime\prime\prime}(\theta_{g})
=-\left.\frac{q_{0}^{4}}{\theta_{g}^{5}}\left[6h_{33}^{\prime\prime}(z)+\frac{q_{0}^{2}}{\theta_{g}} h_{33}^{\prime\prime\prime}(z)\right]\right|_{z=q_{0}^{2}/\theta_{g}}.
\end{align}
\end{subequations}
The eigenfunctions has the following symmetry relationship:
\begin{subequations}\label{5.15}
\begin{align}
\psi_{-,1}^{\prime}(\theta_{g}^{*})&=-\frac{\mathrm{i}q_{0}}{(\theta_{g}^{*})^{2}}
\left[ \psi_{-,3}(\frac{q_{0}^{2}}{\theta_{g}^{*}})
+\frac{q_{0}^{2}}{\theta_{g}^{*}}\psi_{-,3}^{\prime}(\frac{q_{0}^{2}}{\theta_{g}^{*}}) \right], \quad
\psi_{+,1}^{\prime}(\theta_{g})=-\frac{\mathrm{i}q_{0}}{\theta_{g}^{2}}
\left[ \psi_{+,3}(\frac{q_{0}^{2}}{\theta_{g}})
+\frac{q_{0}^{2}}{\theta_{g}}\psi_{+,3}^{\prime}(\frac{q_{0}^{2}}{\theta_{g}}) \right],
 \\
\psi_{-,3}^{\prime}(\theta_{g})&=\frac{\mathrm{i}q_{0}}{\theta_{g}^{2}}
\left[ \psi_{-,1}(\frac{q_{0}^{2}}{\theta_{g}})
+\frac{q_{0}^{2}}{\theta_{g}}\psi_{-,1}^{\prime}(\frac{q_{0}^{2}}{\theta_{g}}) \right], \quad
\psi_{+,3}^{\prime}(\theta_{g}^{*})=\frac{\mathrm{i}q_{0}}{(\theta_{g}^{*})^{2}}
\left[ \psi_{+,1}(\frac{q_{0}^{2}}{\theta_{g}^{*}})
+\frac{q_{0}^{2}}{\theta_{g}^{*}}\psi_{+,1}^{\prime}(\frac{q_{0}^{2}}{\theta_{g}^{*}}) \right], \\
\widetilde{\gamma}^{\prime}(\theta_{g})
&=\frac{q_{0}^{2}}{\theta_{g}^{2}}\gamma^{\prime}(\frac{q_{0}^{2}}{\theta_{g}}), \quad
\gamma^{\prime}(\theta_{g}^{*})=\frac{q_{0}^{2}}{(\theta_{g}^{*})^{2}}
\widetilde{\gamma}^{\prime}(\frac{q_{0}^{2}}{\theta_{g}^{*}}).
\end{align}
\end{subequations}
\end{proposition}

\begin{corollary}\label{cor:10}
The norming constants follows the following symmetry relationship:
\begin{subequations}\label{5.16}
\begin{align}
\bar{f}_{g}&=-\frac{f_{g}^{*}}{\rho(\theta_{g}^{*})}, \quad \hat{f}_{g}=\frac{\mathrm{i}q_{0}}{\theta_{g}^{*}\rho(\theta_{g}^{*})}f_{g}^{*}, \quad
\check{f}_{g}=\frac{\mathrm{i}\theta_{g}}{q_{0}s_{33}(\theta_{g})}f_{g}, \quad e_{g}=\check{e}_{g}=\bar{e}_{g}=\hat{e}_{g}=0, \\
b_{g}&=\frac{\mathrm{i}q_{0}}{\theta_{g}}\left[ \frac{\check{b}_{g}}{\theta_{g}}
+\frac{q_{0}^{2}}{\theta_{g}^{2}}\check{b}_{g} \right],  \quad
b_{g}=-\frac{\bar{b}_{g}^{*}\rho(\theta_{g})}{s_{33}(\theta_{g})}
-\left.\bar{f}_{g}^{*}\left[\frac{\rho(z)}{s_{33}(z)}\right]^{\prime}\right|_{z=\theta_{g}}, \\
\bar{b}_{g}&=\frac{\mathrm{i}\theta_{g}^{*}}{q_{0}}\left[ \frac{\hat{b}_{g}}{\theta_{g}^{*}}
-\frac{q_{0}^{2}}{(\theta_{g}^{*})^{2}}\hat{b}_{g} \right], \quad
\bar{b}_{g}^{*}=-\frac{b_{g}s_{33}(\theta_{g})}{\rho(\theta_{g})}
+\left.\bar{f}_{g}^{*} \ln \left[\frac{s_{33}(z)}{\rho(z)}\right]^{\prime}\right|_{z=\theta_{g}},
\end{align}
\end{subequations}
and
\begin{subequations}\label{5.17}
\begin{align}
\check{W}_{g}(x,t)&=\frac{\mathrm{i}q_{0}^{3}}{\theta_{g}^{3}}W_{g}(x,t), \quad
\check{B}_{g}=\frac{\theta_{g}}{q_{0}^{2}}-\frac{\theta_{g}^{2}}{q_{0}^{2}}\bar{B}_{g}^{*}
-\left.\frac{\theta_{g}^{2}s_{33}(\theta_{g})}{q_{0}^{2}\rho(\theta_{g})}
\left[\frac{\rho(z)}{s_{33}(z)}\right]^{\prime}\right|_{z=\theta_{g}}, \quad D_{g}(x,t)=0, \\
\bar{W}_{g}(x,t)&=-\frac{\left[s_{33}(\theta_{g})\right]^{*}}{\rho(\theta_{g}^{*})}W_{g}^{*}(x,t), \quad  B_{g}=\bar{B}_{g}^{*}+\left.\frac{s_{33}(\theta_{g})}{\rho(\theta_{g})}
\left[\frac{\rho(z)}{s_{33}(z)}\right]^{\prime}\right|_{z=\theta_{g}}, \quad \check{D}_{g}(x,t)=0, \\
\hat{W}_{g}(x,t)&=\frac{\mathrm{i}q_{0}^{5}\left[s_{33}(\theta_{g})\right]^{*}}{(\theta_{g}^{*})^{5}
\rho(\theta_{g}^{*})}W_{g}^{*}(x,t), \quad \hat{B}_{g}=\frac{3\theta_{g}^{*}}{q_{0}^{2}}
-\frac{(\theta_{g}^{*})^{2}}{q_{0}^{2}}\bar{W}_{g}, \quad \bar{D}_{g}(x,t)=\hat{D}_{g}(x,t)=0.
\end{align}
\end{subequations}
\end{corollary}

Therefore, the residue conditions are obtained through the Corollary~\ref{cor:9}:
\begin{subequations}\label{5.18}
\begin{align}
\mathbf{R}_{-1,\theta_{g}}^{+}(x,t)&=\left[\mathop{\rm{Res}}_{z=\theta_{g}}\left[ \frac{\nu_{+,1}(z)}{h_{11}(z)} \right], \mathbf{0}, \mathbf{0}\right], \quad
\mathbf{R}_{-1,q_{0}^{2}/\theta_{g}^{*}}^{+}(x,t)=\left[ \mathbf{0}, \mathop{\rm{Res}}_{z=q_{0}^{2}/
\theta_{g}^{*}} \left[ -\frac{\widetilde{d}(z)}{s_{33}(z)} \right], \mathbf{0}\right], \\
\mathbf{R}_{-2,\theta_{g}}^{+}(x,t)&=\left[\mathop{\rm{Y_{-2}}}_{z=\theta_{g}}\left[ \frac{\nu_{+,1}(z)}{h_{11}(z)} \right], \mathbf{0}, \mathbf{0}\right], \quad
\mathbf{R}_{-2,q_{0}^{2}/\theta_{g}^{*}}^{+}(x,t)=\left[ \mathbf{0}, \mathop{\rm{Y_{-2}}}_{z=q_{0}^{2}/
\theta_{g}^{*}} \left[ -\frac{\widetilde{d}(z)}{s_{33}(z)} \right], \mathbf{0}\right], \\
\mathbf{R}_{-1,\theta_{g}^{*}}^{-}(x,t)&=\left[\mathbf{0}, \mathop{\rm{Res}}_{z=\theta_{g}^{*}}\left[ \frac{d(z)}{s_{11}(z)} \right], \mathbf{0}\right], \quad
\mathbf{R}_{-1,q_{0}^{2}/\theta_{g}}^{-}(x,t)=\left[ \mathbf{0}, \mathbf{0}, \mathop{\rm{Res}}_{z=q_{0}^{2}/\theta_{g}}\left[\frac{\nu_{+,3}(z)}{h_{33}(z)}\right]\right], \\
\mathbf{R}_{-2,\theta_{g}^{*}}^{-}(x,t)&=\left[\mathbf{0},\mathop{\rm{Y_{-2}}}_{z=\theta_{g}^{*}}\left[
\frac{d(z)}{s_{11}(z)} \right], \mathbf{0}\right], \quad
\mathbf{R}_{-2,q_{0}^{2}/\theta_{g}}^{-}(x,t)=\left[ \mathbf{0}, \mathbf{0}, \mathop{\rm{Y_{-2}}}_{z=q_{0}^{2}/\theta_{g}}\left[\frac{\nu_{+,3}(z)}{h_{33}(z)}\right]\right].
\end{align}
\end{subequations}

\subsection{Reflectionless solutions with double poles}

\begin{theorem}\label{thm:8}
Suppose that $h_{11}(\theta_{g})=h_{11}^{\prime}(\theta_{g})=0$ and $h_{11}^{\prime \prime}(\theta_{g})\neq 0$, with $\left| \theta_{g} \right|<q_{0}$, the double-poles solutions of the RH problem are shown below:
\begin{equation}\label{5.19}
\begin{split}
\mathbf{R}(z;x,t)&=\mathbf{Y}_{-}(z)-\frac{1}{2\pi\mathrm{i}}
\int_{\mathbb{R}}\frac{\mathbf{R}^{-}(\xi)\widetilde{\mathbf{L}}(\xi)}{\xi-z}\,\mathrm{d}\xi
+\sum_{j=1}^{G}\left[ \frac{\mathbf{R}_{-1,\theta_{j}}^{+}}{z-\theta_{j}} +\frac{\mathbf{R}_{-1,\theta_{j}^{*}}^{-}}{z-\theta_{j}^{*}}
+\frac{\mathbf{R}_{-2,\theta_{j}}^{+}}{[z-\theta_{j}]^{2}}
+\frac{\mathbf{R}_{-2,\theta_{j}^{*}}^{-}}{[z-\theta_{j}^{*}]^{2}} \right] \\
&+\sum_{j=1}^{G}\left[ \frac{\mathbf{R}_{-1,q_{0}^{2}/\theta_{j}^{*}}^{+}}{z-(q_{0}^{2}/\theta_{j}^{*})} +\frac{\mathbf{R}_{-1,q_{0}^{2}/\theta_{j}}^{-}}{z-(q_{0}^{2}/\theta_{j})}
+\frac{\mathbf{R}_{-2,q_{0}^{2}/\theta_{j}^{*}}^{+}}{[z-(q_{0}^{2}/\theta_{j}^{*})]^{2}} +\frac{\mathbf{R}_{-2,q_{0}^{2}/\theta_{j}}^{-}}{[z-(q_{0}^{2}/\theta_{j})]^{2}} \right],
\end{split}
\end{equation}
where $\widetilde{\mathbf{L}}(z)=\mathbf{e}^{\mathrm{i}\mathbf{\Delta}(z;x,t)}
\mathbf{L}(z)\mathbf{e}^{-\mathrm{i}\mathbf{\Delta}(z;x,t)}$ and $\mathbf{R}(x,t,z)=\mathbf{R}^{\pm}(z;x,t)$ for $\mathfrak{Im}z\gtrless0$. Moreover, the eigenfunctions are given by
\begin{equation}\label{5.20}
\begin{split}
\nu_{-,1}(\theta_{g}^{*};x,t)&=\begin{pmatrix}
     \mathrm{i}   \\
     \mathrm{i}\mathbf{q}_{-}/\theta_{g}^{*}
  \end{pmatrix}
-\frac{1}{2\pi\mathrm{i}} \int_{\mathbb{R}}\frac{\left[\mathbf{R}^{-}(\xi)\widetilde{\mathbf{L}}(\xi) \right]_{1}}{\xi-\theta_{g}^{*}}\,\mathrm{d}\xi  \\
&+\sum_{j=1}^{G}\left[ \frac{\left[ B_{j}-\mathrm{i}x-\mathrm{i}t \left( 2\theta_{j}+3\sigma(q_{0}^{2}+\theta_{j}^{2}) \right) \right]W_{j}\widetilde{d}(\theta_{j})}{\theta_{g}^{*}-\theta_{j}}
+\frac{W_{j}\widetilde{d}^{\prime}(\theta_{j})}{\theta_{g}^{*}-\theta_{j}}
+\frac{W_{j}\widetilde{d}(\theta_{j})}{[\theta_{g}^{*}-\theta_{j}]^{2}}
 \right],
\end{split}
\end{equation}
\begin{equation}\label{5.21}
\begin{split}
\nu_{-,1}^{\prime}(\theta_{g}^{*};x,t)&=\begin{pmatrix}
     \mathrm{i}   \\
     -\mathrm{i}\mathbf{q}_{-}/(\theta_{g}^{*})^{2}
  \end{pmatrix}
-\frac{1}{2\pi\mathrm{i}} \int_{\mathbb{R}}\frac{\left[\mathbf{R}^{-}(\xi)\widetilde{\mathbf{L}}(\xi) \right]_{1}}{[\xi-\theta_{g}^{*}]^{2}}\,\mathrm{d}\xi  \\
&-\sum_{j=1}^{G}\left[ \frac{\left[ B_{j}-\mathrm{i}x-\mathrm{i}t \left( 2\theta_{j}+3\sigma(q_{0}^{2}+\theta_{j}^{2}) \right) \right]W_{j}\widetilde{d}(\theta_{j})}{[\theta_{g}^{*}-\theta_{j}]^{2}}
+\frac{W_{j}\widetilde{d}^{\prime}(\theta_{j})}{[\theta_{g}^{*}-\theta_{j}]^{2}}
+\frac{2W_{j}\widetilde{d}(\theta_{j})}{[\theta_{g}^{*}-\theta_{j}]^{3}}
 \right],
\end{split}
\end{equation}
\begin{equation}\label{5.22}
\begin{split}
-\frac{\widetilde{d}(\theta_{g};x,t)}{s_{33}(\theta_{g})}&=\begin{pmatrix}
     0  \\
     \mathbf{q}_{-}^{\perp}/q_{0}
  \end{pmatrix}
-\frac{1}{2\pi\mathrm{i}} \int_{\mathbb{R}}\frac{\left[\mathbf{R}^{-}(\xi)\widetilde{\mathbf{L}}(\xi) \right]_{2}}{\xi-\theta_{g}}\,\mathrm{d}\xi  \\
&+\sum_{j=1}^{G} \left\{ \frac{\bar{W}_{j}}{(\theta_{g}-\theta_{j}^{*})^{2}} \left[1+ (\theta_{g}-\theta_{j}^{*})\left[ \bar{B}_{j}+\mathrm{i}x+\mathrm{i}t \left( 2\theta_{j}^{*}+3\sigma(q_{0}^{2}+(\theta_{j}^{*})^{2}) \right) \right] \right] \right. \\
&+\frac{\mathrm{i}\theta_{j}^{*}\hat{W}_{j}}{q_{0}[\theta_{g}-(q_{0}^{2}/\theta_{j}^{*})]^{2}}
\left[ 1-(\theta_{g}-\frac{q_{0}^{2}}{\theta_{j}^{*}})\left[
\frac{\mathrm{i}(\theta_{j}^{*})^{2}}{q_{0}^{2}}t \left( 2\theta_{j}^{*}+3\sigma(q_{0}^{2}+(\theta_{j}^{*})^{2}) \right) -\hat{B}_{j}
 \right. \right.   \\
& \left. \left. \left. +\frac{\mathrm{i}(\theta_{j}^{*})^{2}}{q_{0}^{2}}x +\frac{\theta_{j}^{*}}{q_{0}^{2}} \right] \right] \right\} \nu_{-,1}(\theta_{j}^{*})
+\sum_{j=1}^{G}\left[ \frac{\bar{W}_{j}}{\theta_{g}-\theta_{j}^{*}}
-\frac{\mathrm{i}(\theta_{j}^{*})^{3}\hat{W}_{j}}{q_{0}^{3}[\theta_{g}-(q_{0}^{2}/\theta_{j}^{*})]}
 \right]\nu_{-,1}^{\prime}(\theta_{j}^{*}),
\end{split}
\end{equation}
\begin{equation}\label{5.23}
\begin{split}
-\frac{\widetilde{d}^{\prime}(\theta_{g};x,t)}{s_{33}(\theta_{g})}&=
-\frac{\widetilde{d}(\theta_{g})s_{33}^{\prime}(\theta_{g})}{[s_{33}(\theta_{g})]^{2}}
-\frac{1}{2\pi\mathrm{i}} \int_{\mathbb{R}}\frac{\left[\mathbf{R}^{-}(\xi)\widetilde{\mathbf{L}}(\xi) \right]_{2}}{(\xi-\theta_{g})^{2}}\,\mathrm{d}\xi  \\
&-\sum_{j=1}^{G} \left\{ \frac{\bar{W}_{j}}{(\theta_{g}-\theta_{j}^{*})^{3}} \left[2+ (\theta_{g}-\theta_{j}^{*})\left[ \bar{B}_{j}+\mathrm{i}x+\mathrm{i}t \left( 2\theta_{j}^{*}+3\sigma(q_{0}^{2}+(\theta_{j}^{*})^{2}) \right) \right] \right] \right. \\
&+\frac{\mathrm{i}\theta_{j}^{*}\hat{W}_{j}}{q_{0}[\theta_{g}-(q_{0}^{2}/\theta_{j}^{*})]^{3}}
\left[ 2-(\theta_{g}-\frac{q_{0}^{2}}{\theta_{j}^{*}})\left[
\frac{\mathrm{i}(\theta_{j}^{*})^{2}}{q_{0}^{2}}t \left( 2\theta_{j}^{*}+3\sigma(q_{0}^{2}+(\theta_{j}^{*})^{2}) \right) -\hat{B}_{j}
 \right. \right.   \\
& \left. \left. \left. +\frac{\mathrm{i}(\theta_{j}^{*})^{2}}{q_{0}^{2}}x +\frac{\theta_{j}^{*}}{q_{0}^{2}} \right] \right] \right\} \nu_{-,1}(\theta_{j}^{*})
+\sum_{j=1}^{G}\left[ \frac{\mathrm{i}(\theta_{j}^{*})^{3}\hat{W}_{j}}{q_{0}^{3}[\theta_{g}-(q_{0}^{2}/\theta_{j}^{*})]^{2}}
-\frac{\bar{W}_{j}}{(\theta_{g}-\theta_{j}^{*})^{2}} \right]\nu_{-,1}^{\prime}(\theta_{j}^{*}).
\end{split}
\end{equation}
\end{theorem}

\begin{theorem}[Reconstruction formula]\label{thm:9}
Double-poles solutions $\mathbf{q}(x,t)$ of the defocusing-defocusing coupled Hirota equations with NZBC~\eqref{1.3} are reconstructed as
\begin{equation}\label{5.24}
\begin{split}
q_{k}(x,t)&=q_{k-}-\frac{1}{2\pi} \int_{\mathbb{R}}\left[
\mathbf{R}^{-}(\xi;x,t)\widetilde{\mathbf{L}}(\xi;x,t) \right]_{(k+1)1}\,\mathrm{d}\xi
-\sum_{j=1}^{G}\mathrm{i}W_{j}\widetilde{d}^{\prime}_{(k+1)}(\theta_{j}) \\
&-\sum_{j=1}^{G}\mathrm{i}W_{j}\left[ B_{j}-\mathrm{i}x-\mathrm{i}t \left( 2\theta_{j}+3\sigma(q_{0}^{2}+\theta_{j}^{2}) \right) \right]\widetilde{d}_{(k+1)}(\theta_{j}), \quad k=1,2.
\end{split}
\end{equation}
\end{theorem}

\subsection{Trace formulae and the double-poles solutions}
The construction method of the trace formula for the double-poles is different from that for the single-poles, therefore it is assumed that
\begin{equation}\label{5.25}
\begin{split}
\chi_{a}^{+}(z)&=h_{11}(z)\prod_{g=1}^{G}\frac{(z-\theta_{g}^{*})^{2}}{(z-\theta_{g})^{2}}, \quad z\in \mathbb{D}^{+},  \quad
\chi_{a}^{-}(z)=s_{11}(z)\prod_{g=1}^{G}\frac{(z-\theta_{g})^{2}}{(z-\theta_{g}^{*})^{2}}, \quad z\in \mathbb{D}^{-}.
\end{split}
\end{equation}
Using~\eqref{2.61} and the scattering coefficients, then we have
\begin{equation}\label{5.26}
\begin{split}
\ln \chi_{a}^{+}(z)+\ln \chi_{a}^{-}(z)=-\ln \left[1-\left| \beta_{1}(z) \right|^{2}
-\frac{\left| \beta_{2}(z) \right|^{2}}{\rho(z)} \right], \quad z\in \mathbb{R}.
\end{split}
\end{equation}

By combining~\eqref{5.25} with~\eqref{5.26} and using the Plemelj formula, it can be concluded that
\begin{subequations}\label{5.27}
\begin{align}
\chi_{a}^{+}(z)&=\exp\left[ -\frac{1}{2\pi\mathrm{i}} \int_{\mathbb{R}} \ln \left[1-\left| \beta_{1}(\xi) \right|^{2}-\frac{\left| \beta_{2}(\xi) \right|^{2}}{\rho(\xi)} \right] \,\frac{\mathrm{d}\xi}{\xi-z} \right],  \quad z\in \mathbb{D}^{+}, \\
\chi_{a}^{-}(z)&=\exp\left[ \frac{1}{2\pi\mathrm{i}} \int_{\mathbb{R}} \ln \left[1-\left| \beta_{1}(\xi) \right|^{2}-\frac{\left| \beta_{2}(\xi) \right|^{2}}{\rho(\xi)} \right] \,\frac{\mathrm{d}\xi}{\xi-z} \right],  \quad \,\,\,\, z\in \mathbb{D}^{-}.
\end{align}
\end{subequations}
By substituting expression~\eqref{5.27} into definition~\eqref{5.25}, the scattering coefficient display expression can be solved
\begin{subequations}\label{5.28}
\begin{align}
h_{11}(z)&=\prod_{g=1}^{G}\frac{(z-\theta_{g})^{2}}{(z-\theta_{g}^{*})^{2}}
\exp\left[ -\frac{1}{2\pi\mathrm{i}} \int_{\mathbb{R}} \ln \left[1-\left| \beta_{1}(\xi) \right|^{2}-\frac{\left| \beta_{2}(\xi) \right|^{2}}{\rho(\xi)} \right] \,\frac{\mathrm{d}\xi}{\xi-z} \right],  \quad z\in \mathbb{D}^{+}, \\
s_{11}(z)&=\prod_{g=1}^{G}\frac{(z-\theta_{g}^{*})^{2}}{(z-\theta_{g})^{2}}
\exp\left[ \frac{1}{2\pi\mathrm{i}} \int_{\mathbb{R}} \ln \left[1-\left| \beta_{1}(\xi) \right|^{2}-\frac{\left| \beta_{2}(\xi) \right|^{2}}{\rho(\xi)} \right] \,\frac{\mathrm{d}\xi}{\xi-z} \right],  \quad \,\,\,\, z\in \mathbb{D}^{-}.
\end{align}
\end{subequations}
Subsequently, the asymptotic phase difference $\Delta\delta=\delta_{+}-\delta_{-}$ between the respective double-poles is considered.
\begin{equation}\label{5.29}
\begin{split}
\Delta\delta=4\sum_{g=1}^{G}\arg(\theta_{g})+\frac{1}{2\pi} \int_{\mathbb{R}} \ln \left[1-\left| \beta_{1}(\xi) \right|^{2}-\frac{\left| \beta_{2}(\xi) \right|^{2}}{\rho(\xi)} \right] \,\frac{\mathrm{d}\xi}{\xi}.
\end{split}
\end{equation}

\begin{theorem}\label{thm:10}
In the reflectionless case, the double-poles solutions~\eqref{5.24} of the defocusing-defocusing coupled Hirota equation with NZBC~\eqref{1.3} may be written
\begin{equation}\label{5.30}
\begin{split}
\widehat{\mathbf{q}}(x,t)=\frac{1}{\det\widehat{\mathbf{K}}} \begin{pmatrix}
\det \widehat{\mathbf{K}}_{1}^{\mathrm{aug}} \\
\det \widehat{\mathbf{K}}_{2}^{\mathrm{aug}}
\end{pmatrix},  \quad
\widehat{\mathbf{K}}_{n}^{\mathrm{aug}}=\left(\begin{array}{ll}
q_{n-} & \widehat{\mathbf{E}} \\
\widehat{\mathbf{A}}_{n} & \widehat{\mathbf{K}}
\end{array}\right), \quad n=1,2,
\end{split}
\end{equation}
the vectors $\widehat{\mathbf{A}}_{n}$, $\widehat{\mathbf{E}}$ and matrix $\widehat{\mathbf{K}}$ are
\begin{equation}\label{5.31}
\begin{split}
\widehat{\mathbf{A}}_{n}&=\left(\widehat{A}_{n1}, \ldots, \widehat{A}_{n(2G)}\right)^{T}, \quad
\widehat{\mathbf{E}}=\left(\widehat{E}_{1}, \ldots, \widehat{E}_{2G}\right), \quad \widehat{\mathbf{K}}=\mathbf{I}+\widehat{\mathbf{P}},
\end{split}
\end{equation}
where
\begin{align}\label{5.32}
\begin{split}
\widehat{E}_{g}=\left\{\begin{array}{ll}
\mathrm{i}W_{g}(x,t)B_{g}^{(1)}(x,t),  &g=1, \cdots, G, \\
\mathrm{i}W_{g-G}(x,t), & g=G+1, \cdots, 2G,
\end{array}\right.
\end{split}
\end{align}
For $j=1, \ldots, G$ and $k=1, \ldots, G$, the entries of matrix $\widehat{\mathbf{P}}=(\widehat{P}_{jk}(x,t))$ are
\begin{equation}\label{5.33}
\begin{split}
P_{jk}&=s_{33}(\theta_{j})\sum_{a=1}^{G}\left[W_{a}^{(3)}(\theta_{j})W_{k}^{(1)}(\theta_{a}^{*})
-W_{a}^{(6)}(\theta_{j})W_{k}^{(2)}(\theta_{a}^{*}) \right].
\end{split}
\end{equation}
For $j=1, \ldots, G$ and $k=G+1, \ldots, 2G$, the entries of matrix $\widehat{\mathbf{P}}=(\widehat{P}_{jk}(x,t))$ are
\begin{equation}\label{5.34}
\begin{split}
P_{jk}&=s_{33}(\theta_{j})\sum_{a=1}^{G}\left[W_{a}^{(3)}(\theta_{j})b_{k-G}^{(1)}(\theta_{a}^{*}) -\frac{W_{a}^{(6)}(\theta_{j})b_{k-G}^{(1)}(\theta_{a}^{*})}{\theta_{a}^{*}-\theta_{k-G}} \right].
\end{split}
\end{equation}
For $j=G+1, \ldots, 2G$, and $k=1, \ldots, G$, the entries of matrix $\widehat{\mathbf{P}}=(\widehat{P}_{jk}(x,t))$ are
\begin{equation}\label{5.35}
\begin{split}
P_{jk}&=-\sum_{a=1}^{G}\left[W_{a}^{(9)}(\theta_{j-G})W_{k}^{(1)}(\theta_{a}^{*})
+W_{a}^{(10)}(\theta_{j-G})W_{k}^{(2)}(\theta_{a}^{*}) \right].
\end{split}
\end{equation}
For $j=G+1, \ldots, 2G$, and $k=G+1, \ldots, 2G$, the entries of matrix $\widehat{\mathbf{P}}=(\widehat{P}_{jk}(x,t))$ are
\begin{equation}\label{5.36}
\begin{split}
P_{jk}&=-\sum_{a=1}^{G}\left[W_{a}^{(9)}(\theta_{j-G})b_{k-G}^{(1)}(\theta_{a}^{*}) +\frac{W_{a}^{(10)}(\theta_{j-G})b_{k-G}^{(1)}(\theta_{a}^{*})}{\theta_{a}^{*}-\theta_{k-G}} \right],
\end{split}
\end{equation}
where
\begin{equation}\label{5.37}
\everymath{\displaystyle}
\begin{split}
\widehat{A}_{ni^{\prime}}=\left\{\begin{array}{ll}
(-1)^{n}s_{33}(\theta_{i^{\prime}})\frac{q_{\bar{n}-}^{*}}{q_{0}}
+\sum_{a=1}^{G}\mathrm{i}q_{n-}W_{a}^{(7)}(\theta_{i^{\prime}}),  & i^{\prime}=1, \ldots, G, \\
(-1)^{n}s_{33}^{\prime}(\theta_{i^{\prime}-G})\frac{q_{\bar{n}-}^{*}}{q_{0}}
+\sum_{a=1}^{G}\mathrm{i}q_{n-}W_{a}^{(8)}(\theta_{i^{\prime}-G}), & i^{\prime}=G+1, \ldots, 2G,
\end{array}\right.
\end{split}
\end{equation}
and $\bar{n}=n+(-1)^{n+1}$.
\end{theorem}

\begin{proof}
For $j=1, \ldots, G$, define
\begin{subequations}\label{5.38}
\begin{align}
b_{j}^{(1)}(z)&=\frac{W_{j}}{z-\theta_{j}}, \quad
b_{j}^{(2)}(z)=\frac{\bar{W}_{j}}{z-\theta_{j}^{*}}, \quad
b_{j}^{(3)}(z)=\frac{\hat{W}_{j}}{z-(q_{0}^{2}/\theta_{j}^{*})}, \\
B_{j}^{(1)}(x,t)&=B_{j}-\mathrm{i}x-\mathrm{i}t \left[ 2\theta_{j}+3\sigma(q_{0}^{2}+\theta_{j}^{2}) \right], \quad
B_{j}^{(2)}(x,t)=\bar{B}_{j}+\mathrm{i}x+\mathrm{i}t \left[ 2\theta_{j}^{*}
+3\sigma(q_{0}^{2}+(\theta_{j}^{*})^{2}) \right], \\
B_{j}^{(3)}(x,t)&=\frac{\theta_{j}^{*}}{q_{0}^{2}}-\hat{B}_{j}
+\frac{\mathrm{i}(\theta_{j}^{*})^{2}}{q_{0}^{2}} \left[x+ t(2\theta_{j}^{*}+3\sigma(q_{0}^{2}+(\theta_{j}^{*})^{2})) \right],
\end{align}
\end{subequations}
and
\begin{subequations}\label{5.39}
\begin{align}
W_{j}^{(1)}(z)&=b_{j}^{(1)}(z) \left[ B_{j}^{(1)}+ \frac{1}{z-\theta_{j}} \right], \quad
W_{j}^{(2)}(z)=\frac{b_{j}^{(1)}(z)}{z-\theta_{j}} \left[ B_{j}^{(1)}+ \frac{2}{z-\theta_{j}} \right], \\
W_{j}^{(3)}(z)&=\frac{b_{j}^{(2)}(z)}{z-\theta_{j}^{*}}\left[ 1+(z-\theta_{j}^{*})B_{j}^{(2)} \right]+\frac{\mathrm{i}\theta_{j}^{*}b_{j}^{(3)}(z)}{q_{0}[z-(q_{0}^{2}/\theta_{j}^{*})]}\left[ 1-(z-\frac{q_{0}^{2}}{\theta_{j}^{*}})B_{j}^{(3)} \right], \\
W_{j}^{(4)}(z)&=\frac{b_{j}^{(2)}(z)}{(z-\theta_{j}^{*})^{2}}\left[ 2+(z-\theta_{j}^{*})B_{j}^{(2)} \right]+\frac{\mathrm{i}\theta_{j}^{*}b_{j}^{(3)}(z)}{q_{0}[z-(q_{0}^{2}/\theta_{j}^{*})]^{2}}\left[ 2-(z-\frac{q_{0}^{2}}{\theta_{j}^{*}})B_{j}^{(3)} \right], \\
W_{j}^{(5)}(z)&=
\frac{\mathrm{i}(\theta_{j}^{*})^{3}b_{j}^{(3)}(z)}{q_{0}^{3}[z-(q_{0}^{2}/\theta_{j}^{*})]}
-\frac{b_{j}^{(2)}(z)}{z-\theta_{j}^{*}},  \quad
W_{j}^{(6)}(z)=b_{j}^{(2)}(z)-\frac{\mathrm{i}(\theta_{j}^{*})^{3}b_{j}^{(3)}(z)}{q_{0}^{3}}, \\
W_{j}^{(7)}(z)&=s_{33}(z) \left[ W_{j}^{(6)}(z)/(\theta_{j}^{*})^{2}
-W_{j}^{(3)}(z)/\theta_{j}^{*} \right], \\
W_{j}^{(8)}(z)&=s_{33}(z) \left[ \frac{W_{j}^{(5)}(z)}{(\theta_{j}^{*})^{2}}
+\frac{W_{j}^{(4)}(z)}{\theta_{j}^{*}} \right]+s_{33}^{\prime}(z) \left[ \frac{W_{j}^{(6)}(z)}{(\theta_{j}^{*})^{2}}-\frac{W_{j}^{(3)}(z)}{\theta_{j}^{*}} \right], \\
W_{j}^{(9)}(z)&=s_{33}(z)W_{j}^{(4)}(z)-s_{33}^{\prime}(z)W_{j}^{(3)}(z), \quad
W_{j}^{(10)}(z)=s_{33}(z)W_{j}^{(5)}(z)+s_{33}^{\prime}(z)W_{j}^{(6)}(z).
\end{align}
\end{subequations}
Under the condition of the reflectionless, there are equations that hold true

\begin{equation}\label{5.40}
\begin{split}
\nu_{-,21}(\theta_{j}^{*};x,t)&=\frac{\mathrm{i}q_{1-}}{\theta_{j}^{*}}
+\sum_{j^{\prime}=1}^{G} \left[ \left[B_{j^{\prime}}^{(1)}+\frac{1}{\theta_{j}^{*}-\theta_{j^{\prime}}} \right]b_{j^{\prime}}^{(1)}(\theta_{j}^{*})\widetilde{d}_{2}(\theta_{j^{\prime}})
+b_{j^{\prime}}^{(1)}(\theta_{j}^{*})\widetilde{d}_{2}^{\prime}(\theta_{j^{\prime}}) \right],
\end{split}
\end{equation}
\begin{equation}\label{5.41}
\begin{split}
\nu_{-,21}^{\prime}(\theta_{j}^{*};x,t)&=-\frac{\mathrm{i}q_{1-}}{(\theta_{j}^{*})^{2}}
-\sum_{j^{\prime}=1}^{G} \left[ \left[B_{j^{\prime}}^{(1)}+\frac{2}{\theta_{j}^{*}-\theta_{j^{\prime}}} \right] \frac{b_{j^{\prime}}^{(1)}(\theta_{j}^{*})\widetilde{d}_{2}(\theta_{j^{\prime}})}
{\theta_{j}^{*}-\theta_{j^{\prime}}}
+\frac{b_{j^{\prime}}^{(1)}(\theta_{j}^{*})\widetilde{d}_{2}^{\prime}(\theta_{j^{\prime}})}
{\theta_{j}^{*}-\theta_{j^{\prime}}} \right],
\end{split}
\end{equation}
\begin{equation}\label{5.42}
\begin{split}
\widetilde{d}_{2}(\theta_{i^{\prime}};x,t)&=-s_{33}(\theta_{i^{\prime}}) \left[ \frac{q_{2-}^{*}}{q_{0}}
+\sum_{j=1}^{G} W_{j}^{(6)}(\theta_{i^{\prime}})\nu_{-,21}^{\prime}(\theta_{j}^{*})
+W_{j}^{(3)}(\theta_{i^{\prime}})\nu_{-,21}(\theta_{j}^{*})  \right],
\end{split}
\end{equation}
\begin{equation}\label{5.43}
\begin{split}
\widetilde{d}_{2}^{\prime}(\theta_{i^{\prime}};x,t)&=
-s_{33}^{\prime}(\theta_{i^{\prime}})\frac{q_{2-}^{*}}{q_{0}}
+\sum_{j=1}^{G} \left[ s_{33}(\theta_{i^{\prime}})W_{j}^{(4)}(\theta_{i^{\prime}})  -s_{33}^{\prime}(\theta_{i^{\prime}})W_{j}^{(3)}(\theta_{i^{\prime}}) \right]\nu_{-,21}(\theta_{j}^{*}) \\
&-\sum_{j=1}^{G} \left[ s_{33}(\theta_{i^{\prime}})W_{j}^{(5)}(\theta_{i^{\prime}})  +s_{33}^{\prime}(\theta_{i^{\prime}})W_{j}^{(6)}(\theta_{i^{\prime}}) \right]\nu_{-,21}^{\prime}(\theta_{j}^{*}),
\end{split}
\end{equation}
where $i^{\prime}=1, \ldots, G$.
Substituting~\eqref{5.40} and~\eqref{5.41} into equations~\eqref{5.42} and~\eqref{5.43} respectively yields
\begin{equation}\label{5.44}
\begin{split}
\widetilde{d}_{2}(\theta_{i^{\prime}};x,t)&=-s_{33}(\theta_{i^{\prime}})\frac{q_{2-}^{*}}{q_{0}}
+\sum_{j=1}^{G}\sum_{j^{\prime}=1}^{G}s_{33}(\theta_{i^{\prime}}) \left[
W_{j}^{(6)}(\theta_{i^{\prime}})W_{j^{\prime}}^{(2)}(\theta_{j}^{*})
-W_{j}^{(3)}(\theta_{i^{\prime}})W_{j^{\prime}}^{(1)}(\theta_{j}^{*})
\right]\widetilde{d}_{2}(\theta_{j^{\prime}}) \\
&+\sum_{j=1}^{G}\mathrm{i}q_{1-}W_{j}^{(7)}(\theta_{i^{\prime}})
+\sum_{j=1}^{G}\sum_{j^{\prime}=1}^{G}s_{33}(\theta_{i^{\prime}})b_{j^{\prime}}^{(1)}(\theta_{j}^{*}) \left[ \frac{W_{j}^{(6)}(\theta_{i^{\prime}})}{\theta_{j}^{*}-\theta_{j^{\prime}}}
-W_{j}^{(3)}(\theta_{i^{\prime}}) \right]\widetilde{d}_{2}^{\prime}(\theta_{j^{\prime}}),
\end{split}
\end{equation}
and
\begin{equation}\label{5.45}
\begin{split}
\widetilde{d}_{2}^{\prime}(\theta_{i^{\prime}};x,t)&
=-s_{33}^{\prime}(\theta_{i^{\prime}})\frac{q_{2-}^{*}}{q_{0}}
+\sum_{j=1}^{G}\sum_{j^{\prime}=1}^{G} \left[
W_{j}^{(9)}(\theta_{i^{\prime}})W_{j^{\prime}}^{(1)}(\theta_{j}^{*})
+W_{j}^{(10)}(\theta_{i^{\prime}})W_{j^{\prime}}^{(2)}(\theta_{j}^{*})
\right]\widetilde{d}_{2}(\theta_{j^{\prime}}) \\
&+\sum_{j=1}^{G}\mathrm{i}q_{1-}W_{j}^{(8)}(\theta_{i^{\prime}})
+\sum_{j=1}^{G}\sum_{j^{\prime}=1}^{G} b_{j^{\prime}}^{(1)}(\theta_{j}^{*}) \left[
\frac{W_{j}^{(10)}(\theta_{i^{\prime}})}{\theta_{j}^{*}
-\theta_{j^{\prime}}} +W_{j}^{(9)}(\theta_{i^{\prime}}) \right]\widetilde{d}_{2}^{\prime}(\theta_{j^{\prime}}).
\end{split}
\end{equation}
The equations for $\widetilde{d}_{2}(\theta_{i^{\prime}};x,t)$ and $\widetilde{d}_{2}^{\prime}(\theta_{i^{\prime}};x,t)$ form a closed system of $2G$ equations with $2G$ unknowns. In the same way, a closed system containing $\widetilde{d}_{3}(\theta_{i^{\prime}};x,t)$ and $\widetilde{d}_{3}^{\prime}(\theta_{i^{\prime}};x,t)$ can be found. These two systems can be written as $\widehat{\mathbf{K}}\widehat{\mathbf{X}}_{n}=\widehat{\mathbf{A}}_{n}$ for $n=1,2$, while $\widehat{\mathbf{X}}_{n}=\left(\widehat{X}_{n1}, \ldots, \widehat{X}_{n(2G)}\right)^{T}$ and
\begin{equation}\label{5.46}
\begin{split}
\widehat{X}_{n i^{\prime}}=\left\{\begin{array}{ll}
\widetilde{d}_{(n+1)}(\theta_{i^{\prime}};x,t), & i^{\prime}=1, \ldots, G, \\
\widetilde{d}_{(n+1)}^{\prime}(\theta_{i^{\prime}-G};x,t), & i^{\prime}=G+1, \ldots, 2G.
\end{array}\right.
\end{split}
\end{equation}
Using Cramer's rule, we have
\begin{equation}\label{5.47}
\begin{split}
\widehat{X}_{n i}=\frac{\det\widehat{\mathbf{K}}_{ni}^{\mathrm{aug}}}{\det\widehat{\mathbf{K}}}, \quad i=1, \ldots, 2G, \quad n=1,2,
\end{split}
\end{equation}
where $\widehat{\mathbf{K}}_{ni}^{\mathrm{aug}}=\left( \widehat{\mathbf{K}}_{1}, \ldots, \widehat{\mathbf{K}}_{i-1}, \widehat{\mathbf{A}}_{n}, \widehat{\mathbf{K}}_{i+1}, \ldots, \widehat{\mathbf{K}}_{2G} \right)$. By inserting the determinant representation of the solutions from equation~\eqref{5.47} into equation~\eqref{5.24}, one obtains~\eqref{5.30}.
\end{proof}

\begin{figure}[htb]
\centering
\begin{tabular}{cccc}
\includegraphics[width=0.22\textwidth]{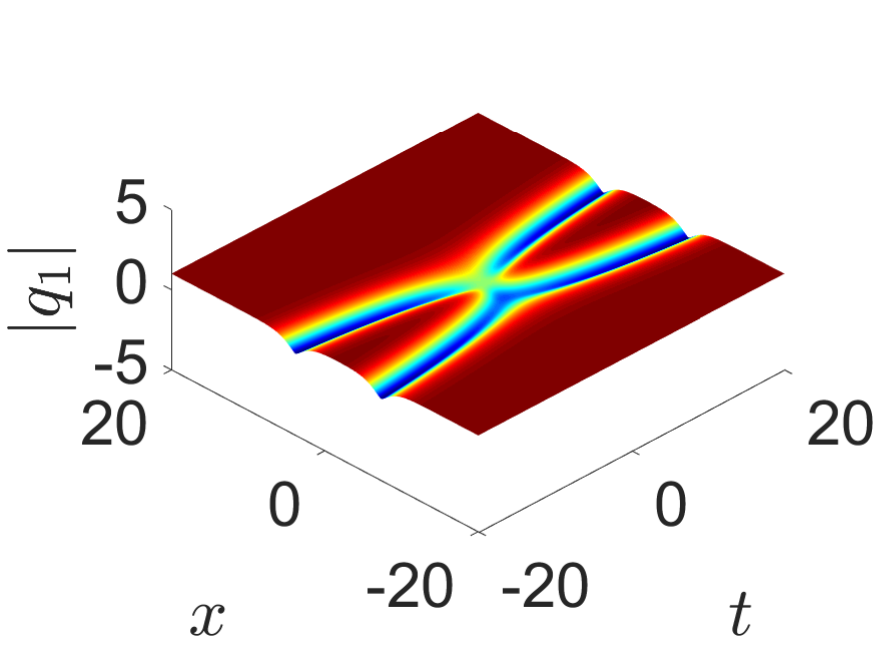} &
\includegraphics[width=0.22\textwidth]{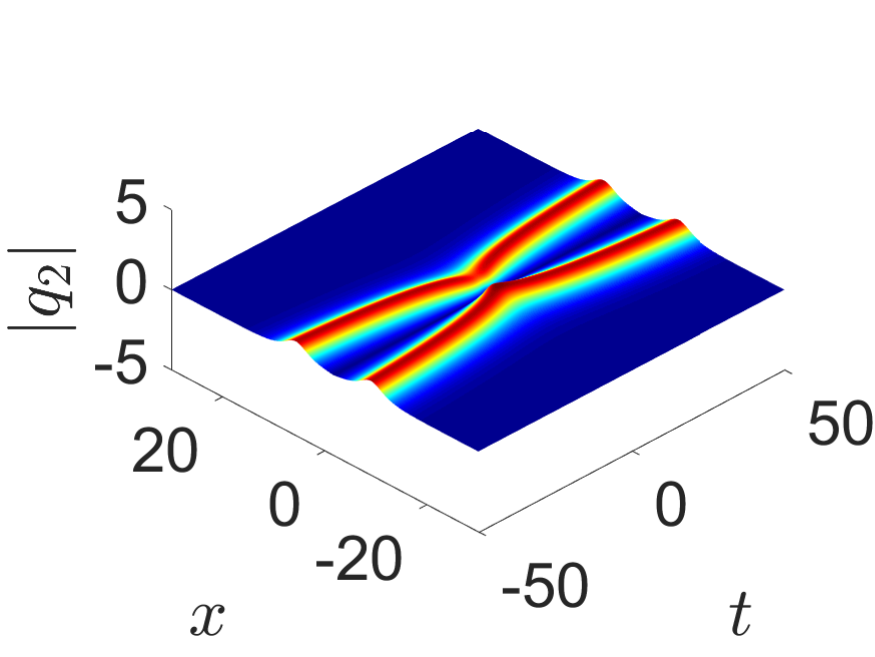} &
\includegraphics[width=0.22\textwidth]{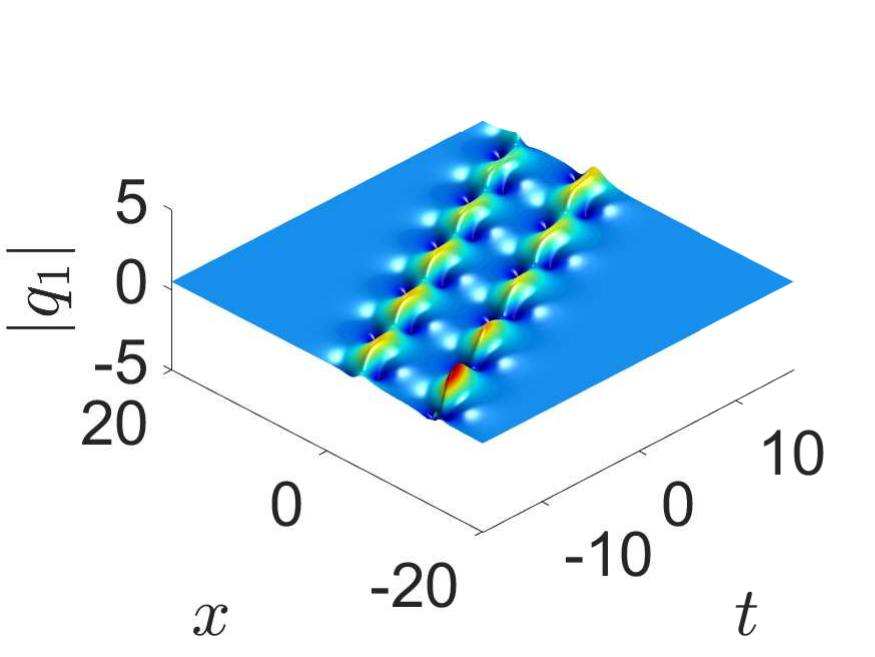} &
\includegraphics[width=0.22\textwidth]{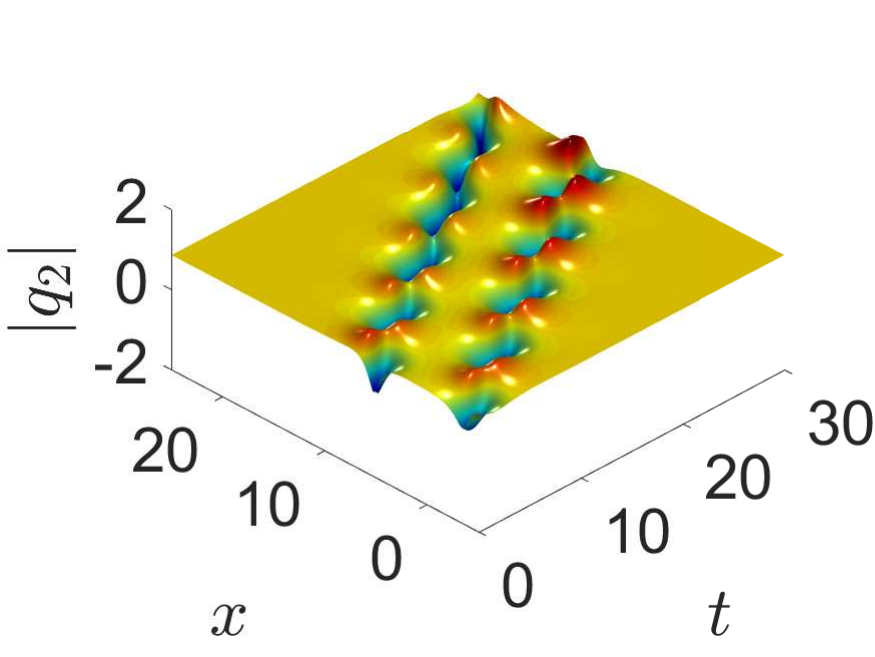} \\
		(a1) & (a2) & (b1) & (b2)
	\end{tabular}
\caption{\small (a1) and (a2): One dark-bright double-pole solution by taking $\mathbf{q}_{-}=(1,
0)^{T}$, $\sigma=10^{-3}$, $\widehat{K}_{1}=0.5$, $\widehat{K}_{2}=\widehat{K}_{3}=2$, $\widehat{\kappa}_{1}=\widehat{\chi}_{2}=0.5$, $\widehat{\chi}_{1}=-2$, $\widehat{\kappa}_{2}=-0.5$, $\widehat{\alpha}_{1}=\frac{1}{2}\pi$. (b1) and (b2): One bright-breather-dark-breather double-pole solution by taking $\mathbf{q}_{-}=(\frac{1}{2}\mathrm{e}^{-\frac{1}{5}\mathrm{i}\pi},
-\frac{\sqrt{3}}{2}\mathrm{e}^{-\frac{1}{5}\mathrm{i}\pi})^{T}$, $\sigma=10^{-2}$, $\widehat{K}_{1}=0.98$, $\widehat{\kappa}_{2}=\widehat{K}_{2}=\widehat{K}_{3}=2$, $\widehat{\kappa}_{1}=-1.5$, $\widehat{\chi}_{2}=-0.2$, $\widehat{\chi}_{1}=1$, $\widehat{\alpha}_{1}=\frac{3}{5}\pi$.
}\label{fig:6}
\end{figure}

We have now derived the double-poles solutions of the defocusing-defocusing coupled Hirota equations with NZBC~\eqref{1.3}. Considering the double-poles solutions and parametrizing the discrete eigenvalues and normalization constants as
\begin{equation}\label{5.48}
\begin{split}
\theta_{1}=\widehat{K}_{1}\mathrm{e}^{\mathrm{i}\widehat{\alpha}_{1}}, \quad
f_{1}=\widehat{K}_{2}\mathrm{e}^{\widehat{\kappa}_{1}+\mathrm{i}\widehat{\chi}_{1}}, \quad
\bar{b}_{1}=\widehat{K}_{3}\mathrm{e}^{\widehat{\kappa}_{2}+\mathrm{i}\widehat{\chi}_{2}}, \quad 0<\widehat{K}_{1}<q_{0}.
\end{split}
\end{equation}
Through the expression~\eqref{5.48} and the reflectionless potentials, it can be inferred that the different structures of the double-poles solutions are obtained. For $q_{1-}\times q_{2-}=0$, one dark-bright double-pole solution is given by (a1) and (a2) of Fig.~\ref{fig:6}. Moreover, setting $q_{1-}\times q_{2-}\neq0$ generates one bright-breather-dark-breather double-pole solution in (b1) and (b2) of Fig.~\ref{fig:6}.

\section{Conclusion}
\label{s:Conclusion}

We apply the IST tool to the defocusing-defocusing coupled Hirota equations with NZBC~\eqref{1.3} and derive some interesting results by constructing the matrix RH problem. We delve into the analytic properties of Jost eigenfunctions and scattering coefficients, while examining particular potential conditions that guarantee such analyticity. Innovative analytical eigenfunctions for the defocusing-defocusing coupled Hirota equations with NZBC~\eqref{1.3} adhere to two symmetry conditions. These relations are subsequently utilized to provide a rigorous characterization of the discrete spectrum. The discrete spectrum yields discrete eigenvalues in two distinct scenarios, each linked to a diverse array of soliton solution types. The characteristics of their soliton interactions are depicted through graphical illustrations. By discussing the different eigenvalues on the circle and off the circle, corresponding combinations of dark solitons, bright solitons, and breather solitons are obtained. Derived the double-pole solutions of Eq.~\eqref{1.3} and proved the rationality of its algebraic closed system. For the defocusing-defocusing coupled Hirota equations, a novel bright-breather-dark-breather double-pole solution is found for the first time, which may help explain and predict certain characteristics in optical solitons and fluid dynamics phenomena.

In this paper, the pure soliton solutions and the double-pole solutions are obtained under the condition of reflectionless potentials. Although this method is simple, it inevitably has some limitations. In the face of the existence of the reflection potential, even if we can derive the corresponding solutions, they often contain implicit integral terms, which brings challenges to the analysis and application. In the future research, we plan to explore how to use the IST technique to eliminate these integral terms, so as to construct explicit soliton solutions. For the numerical inverse scattering~\cite{11}, the computational efficiency and practicability of the solution are improved while maintaining the accuracy of the analytical solutions. By combining numerical methods, we hope to understand and solve the soliton problem in the presence of reflection potential more comprehensively, and further expand the application range of IST technology in nonlinear physical phenomena.

It is worth noting that the IST technique also plays an important role in exploring the soliton solutions of nonlinear integrable flows involving coordinate reflection points~\cite{15,C8,C9}. This technique provides a new perspective for the analysis of soliton solutions with its unique advantages~\cite{16}. Recently, with the in-depth study of the $4\times4$ matrix spectral problem, the IST has not only been successfully applied to the generation of coupled and combined integrable models~\cite{17}, but also has been analyzed in detail from the perspective of double Hamiltonian systems~\cite{18}. These models not only have significant integrable properties in mathematics, but also show great potential and value in the application of physics, mechanical engineering and materials science~\cite{19}.

At present, the IST technique has been successfully applied to deal with the parallel boundary conditions of the defocusing-defocusing coupled Hirota equations with NZBC~\eqref{1.3} at infinity. Similarly, this technique can also be extended to the study of non-parallel boundary conditions at infinity~\cite{20}, further enriching its application in the field of mathematical physics. As the field of mathematics and physics continues to advance, we look forward to seeing more accurate algorithms and more in-depth theoretical analysis. These advances will help to solve more complex physical problems, such as boundary problems in quantum field theory, nonlinear dynamics and fluid mechanics. At the same time, the application of the robust IST method~\cite{21} in engineering and applied science will also be further developed to provide more effective solutions for practical problems.

\section*{Acknowledgement}
Zhang's work was partially supported by the National Natural Science Foundation of China (Grant Nos. 11371326 and 12271488). Ma's work was partially supported by the Ministry of Science and Technology of China (G2021016032L and G2023016011L).

\section*{Conflict of interests}
The authors declare that there is no conflict of interests regarding the research effort and the publication of this paper.

\section*{Data availability statements}
All data generated or analyzed during this study are included in this published article.

\section*{ORCID}
{\setlength{\parindent}{0cm}
Peng-Fei Han https://orcid.org/0000-0003-1164-5819 \\
Wen-Xiu Ma https://orcid.org/0000-0001-5309-1493 \\
Ru-Suo Ye https://orcid.org/0000-0001-7209-2336 \\
Yi Zhang https://orcid.org/0000-0002-8483-4349}

\end{document}